\documentclass[reqno,12pt,letterpaper]{amsart}
\usepackage{amsmath,amssymb,amsthm,graphicx,mathrsfs,url}\usepackage[usenames,dvipsnames]{color}
\usepackage[colorlinks=true,linkcolor=Red,citecolor=Green]{hyperref}
\usepackage{amsxtra}
\usepackage{subcaption}

\def\arXiv#1{\href{http://arxiv.org/abs/#1}{arXiv:#1}}

\def\?[#1]{\textbf{[#1]}\marginpar{\Large{\textbf{??}}}} 
\def\smallsection#1{\smallskip\noindent\textbf{#1}.}
\let\epsilon=\varepsilon 

\newcommand{\RR}{{\mathbb R}}

\newcommand{\CC}{{\mathbb C}}

\newcommand{\ZZ}{{\mathbb Z}}

\newtheorem{theo}{Theorem}
\newtheorem{prop}{Proposition}[section]

\newtheorem{lemm}[prop]{Lemma}

\numberwithin{equation}{section}

\DeclareMathOperator{\Spec}{Spec}

\let\Im=\Imag

\let\Re=\Real

\DeclareMathOperator{\vol}{vol}

\DeclareMathOperator{\tr}{tr}

\usepackage{scalerel}

\def\codefontsize{\small}

\newcommand\reallywidehat[1]{\arraycolsep=0pt\relax%
\begin{array}{c}
\stretchto{
  \scaleto{
    \scalerel*[\widthof{\ensuremath{#1}}]{\kern-.5pt\bigwedge\kern-.5pt}
    {\rule[-\textheight/2]{1ex}{\textheight}} 
  }{\textheight} %
}{0.5ex}\\           
#1\\                 
\rule{-1ex}{0ex}
\end{array}
}
\title[Mathematics of magic angles in a model of twisted bilayer graphene]{Mathematics of magic angles\\ in a model of twisted bilayer graphene}
\author{Simon Becker}
\email{simon.becker@damtp.cam.ac.uk}
\address{Department of Applied Mathematics and Theoretical Physics, University of Cambridge, Wilberforce Road, Cambridge, CB3 0WA, United Kingdom.}

\author{Mark Embree}
\email{embree@vt.edu}
\address{Department of Mathematics, Virginia Tech, 
Blacksburg, VA 24061, USA}

\author{Jens Wittsten}
\email{jens.wittsten@math.lu.se}
\address{Centre for Mathematical Sciences, Lund University, Box 118, SE-221 00 Lund, Sweden, and Department of Engineering, University of Bor{\aa}s, SE-501 90 Bor{\aa}s, Sweden}

\author{Maciej Zworski}
\email{zworski@math.berkeley.edu}
\address{Department of Mathematics, University of California,
Berkeley, CA 94720, USA.}

\begin{document}

\begin{abstract}
We provide a mathematical account of the recent Physical Reviews Letter by 
Tarnopolsky--Kruchkov--Vishwanath \cite{magic}. The new contributions are a 
spectral characterization of magic angles, its accurate numerical implementation and an exponential estimate on the squeezing of all bands as the angle decreases. Pseudospectral phenomena 
\cite{dsz},\cite{trem},
 due to the non-hermitian nature of operators appearing in the model considered in \cite{magic} play a crucial role in our analysis.
 \end{abstract}

\maketitle

\section{Introduction and statement of results}
\begin{figure}
\begin{center}
\includegraphics[width=\textwidth]{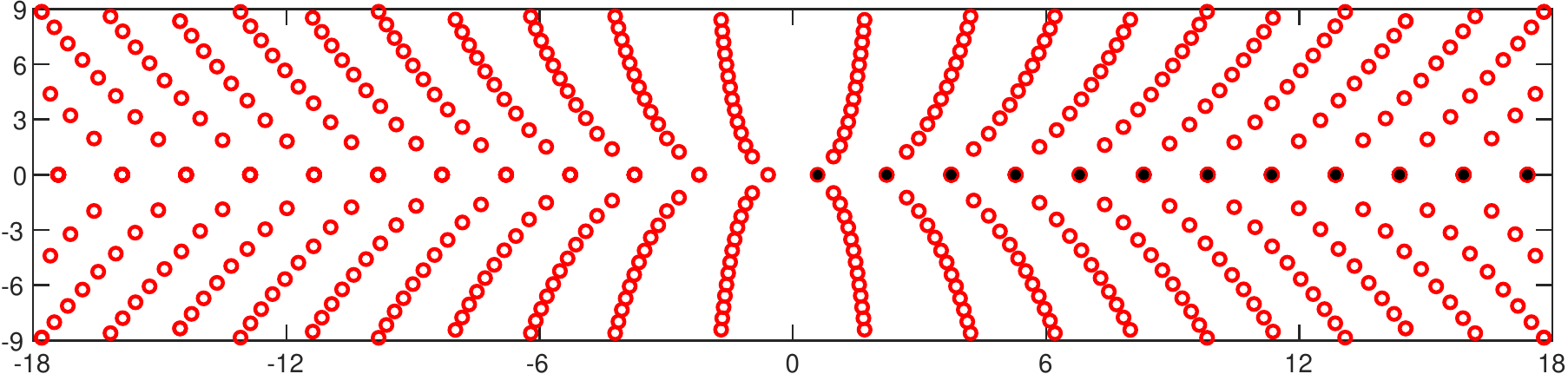}
\end{center}
\caption{\label{f:resa}
Reciprocals of magic angles for the specific potential
\eqref{eq:defU}: resonant $ \alpha$'s  (red circles) come from the full spectrum
of the compact operator \eqref{eq:defT} defining magic angles,
and the magic $ \alpha$'s (black dots) are the reciprocals
of the ``physically relevant" positive angles.}
\end{figure}
Following a recent Physical Review Letter by 
Tarnopolsky--Kruchkov--Vishwanath \cite{magic} we consider the following 
Hamiltonian modeling twisted bilayer graphene:
\begin{equation}
\label{eq:defH}    H ( \alpha)  := \begin{pmatrix} 0 & D(\alpha) ^* \\
D ( \alpha) &  0 \end{pmatrix}  , \ \ \  
D ( \alpha ) := \begin{pmatrix} { 2 }  D_{\bar z}  &   \alpha U(z) \\  
\alpha U(-z) & { 2 } D_{\bar z }   \end{pmatrix} ,
\end{equation} 
where $ z = x_1 + i x_2$, $ D_{\bar z } := \tfrac1{2 i } ( \partial_{x_1} + 
i \partial_{x_2} )$ and 
\begin{equation}
\label{eq:defU} 
U (z ) = U ( z, \bar z ) := \sum_{ k = 0}^2 \omega^k e^{ \frac12 (  z  
\bar \omega^k - \bar z  \omega^k )} , \ \ \ 
\omega := e^{ 2 \pi i /3 } . 
\end{equation}
(We abuse the notation in the argument of $ U $ for the sake of brevity
and write $ U ( z ) $ rather than $ U ( z, \bar z ) $.)
The dimensionless parameter $ \alpha $ is essentially the reciprocal of the 
angle of twisting between the two layers. {When two honeycomb lattices are twisted against one another, a periodic honeycomb superlattice,  called the moir\'e lattice, becomes visible. 
(This name comes from the patterns formed when two fabrics lie on top of each other.) 
  Bistritzer and MacDonald in \cite{BM11} predicted that the symmetries of the periodic moir\'e {lattice} lead to dramatic flattening of the band spectrum.
  The operator \eqref{eq:defH} and in particular potential \eqref{eq:defU} were obtained in 
  \cite{magic} by
removing certain interaction terms from the operator constructed in \cite{BM11}.
  }

{In this paper we consider any potential having the symmetries of \eqref{eq:defU}:
\begin{equation}
\label{eq:symmU2} 
\begin{gathered}  \mathbf a = \tfrac{4}3 \pi i \omega^\ell,  \ \ell =1,2 \ \Longrightarrow \ 
U ( z + \mathbf a ) = \bar \omega U ( z ) , \text{ and }\\
U ( \omega z )  = 
\omega U ( z )   .
\end{gathered}
\end{equation}
The only exception is Theorem \ref{t:squeeze} which requires a non-triviality 
assumption, see \eqref{eq:conditionbracket}. Such potentials are explored further in Section \ref{s:exp}.}

The Hamiltonian $ H $ is periodic with respect to a lattice $ \Gamma $
(see \eqref{eq:Gam} below)
and {\em magic angles} are defined as the $ \alpha $'s (or rather their
reciprocals) at which 
\begin{equation}
\label{eq:defmal}
0 \in \bigcap_{ \mathbf k \in \CC } 
\Spec_{ L^2 ( \CC/\Gamma ) } ( H_{\mathbf k } ( \alpha ) ) , \ \ \
H_{\mathbf k } ( \alpha ) := \begin{pmatrix} 0 & D(\alpha) ^* -\bar {\mathbf k}  \\
D ( \alpha) - \mathbf k &  0 \end{pmatrix}  . 
\end{equation}
The Hamiltonian $ H_{\mathbf k} ( \alpha ) $ comes from the {\em Floquet theory}
of $ H( \alpha ) $ and \eqref{eq:defmal} means that $ H ( \alpha ) $ has a
{\em flat band} at $ 0 $ (see Proposition \ref{p:flat} below). {Since the Bloch electrons have the same energy
at the flat bands, strong electron-electron interactions leading to effects such as superconductivity have been observed at magic angles.} We refer to \cite{magic} for physical motivation and references. Some aspects of this paper carry over to more general models such as the Bistritzer--MacDonald \cite{BM11} and that is discussed in 
\cite{suppl}.

The first theorem is, essentially, the main mathematical result of \cite{magic}. To formulate it we define the Wronskian of two 
$ \CC^2 $-valued $ \Gamma$-periodic functions:
\begin{equation}
\label{eq:Wr}  W ( \mathbf u , \mathbf v ) = \det [ \mathbf u , \mathbf v ] , 
\ \ \ \mathbf u , \mathbf v \in \CC^2, 
\end{equation}
noting that
if $ D ( \alpha ) \mathbf u = D ( \alpha ) \mathbf v = 0$, then 
$ W $ is constant (applying 
$ \partial_{\bar z } $ shows that $ W $ is holomorphic and periodic). 
We also define an involution $ \mathscr E $ satisfying 
$ \mathscr E D ( \alpha ) = D ( \alpha ) \mathscr E $:
\begin{equation}
\label{eq:defE} \mathscr E \mathbf u (\alpha, z ) := 
\begin{pmatrix}  0 & -1 \\  1 & \ \ 0   \end{pmatrix}\mathbf u (\alpha, - z ).\end{equation}
We then have 
{
\begin{theo}
\label{t:magic}
Suppose that $ D ( \alpha ) $ is given by \eqref{eq:defH} with 
$ U \in C^\infty ( \CC/ \Gamma ; \CC ) $ satisfying \eqref{eq:symmU2}.
Then there exists a real-analytic function $f$ on $\RR$ such that 
\begin{equation*}
\label{eq:magic}
0 \in \bigcap_{ \mathbf k \in \CC } 
\Spec_{ L^2 ( \CC/\Gamma ) } ( H_{\mathbf k } ( \alpha ) ) \ \Longleftrightarrow f(\alpha)=0.
\end{equation*}
The function $f$ is defined using a projectively unique
family  $ \RR \ni \alpha \mapsto \mathbf u (\alpha)  \in C^\infty ( \CC/\Gamma; 
 \CC^2 ) $  such that
$ \mathbf u ( 0 ) =( 1, 0 )^t $, $ D ( \alpha ) \mathbf u ( \alpha ) = 0 $. Then 
$f(\alpha):=W ( \mathbf u ( \alpha ) ,  \mathscr E \mathbf u ( \alpha    ) ) $,
where $ W $ is given by \eqref{eq:Wr} and $ \mathscr E $ is defined in \eqref{eq:defE}.
\end{theo}}

\begin{figure}
\includegraphics[width=13cm]{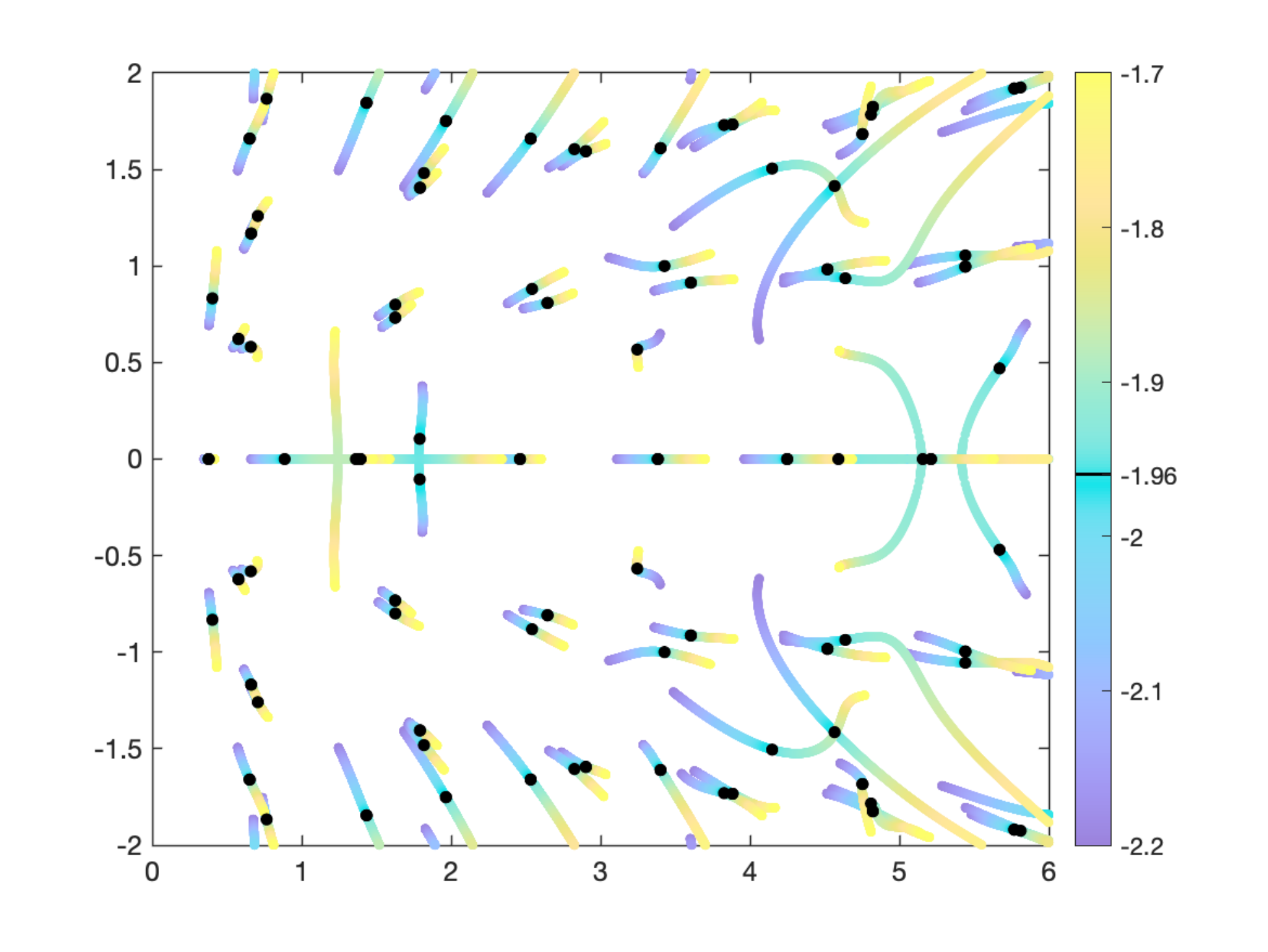}
\caption{\label{f:mu}  For 
$ U_\mu ( z ) = U ( z ) +  \mu \sum_{k=0}^2 \omega^k e^{\bar z \omega^k-  z \bar \omega^k} $, 
with $ U $ given by 
\eqref{eq:defU} and $ \mu = -1.96 $, we 
show set $ \mathcal A $ (indicated by 
$ \bullet $).
 The distribution 
is much less regular than for $ \mu = 0 $ 
shown in Figure~\ref{f:resa}, and nothing like 
\eqref{eq:alphaj} can be expected. The coloured paths trace
the dynamics of magic $\alpha$'s for $-2.2 \leq \mu \leq -1.7 $:
to understand the dependence of ``physically relevant" real $ \alpha$'s
complex values should be considered.}
\end{figure}

A more precise,  representation theoretical, description of $ \mathbf u ( \alpha ) $ will be given in \S \ref{s:ham}. Projective uniqueness means uniqueness up to a multiplicative factor.
In \S \ref{s:alph} we show that (after possibly switching 
$ \mathbf u $ and $ \mathscr E \mathbf u $)
\begin{equation}
\label{eq:stack1}
v ( \alpha ) :=  W ( \mathbf u ( \alpha ) , 
\mathscr E \mathbf u ( \alpha ) ) = 0 \ \Longleftrightarrow \ 
\mathbf u ( \alpha,z_S ) = 0, \ \ \  z_S := \tfrac{ 4 \sqrt {3} } 9 \pi , 
\end{equation}
which then provides a recipe \cite{magic} for constructing the
zero eigenfunctions of $ H_{\mathbf k}  ( \alpha ) $: if $ v ( \alpha ) = 0$
then $ ( D( \alpha ) - \mathbf k ) \mathbf u_{\mathbf k}  ( \alpha) = 0 $, 
$\mathbf u_{\mathbf k }(\alpha) \in C^\infty ( \CC/\Gamma ; \CC^2 ) $, 
where 
\begin{equation}
\label{eq:recipe} 
\mathbf u_{\mathbf k } (z) = e^{  \frac i 2 ( z \bar {\mathbf k }+  \bar z {\mathbf k}) } \frac{ \theta_{ -\frac 16 + k_1/3, \frac 16 - k_2/3 } ( 3 z/ 4 \pi i 
\omega | \omega )}{ \theta_{- \frac1 6 , +\frac 16 } ( 3 z/ 4\pi i \omega | \omega ) } 
\mathbf u ( z ) 
, 
\ \  \mathbf k = \tfrac{1}{\sqrt 3 } ( k_1  \omega^2 - k_2\omega ),
\end{equation}
where $ \zeta \mapsto \theta_{ a,b} ( \zeta | \omega ) $ is the 
Jacobi theta function -- 
see \S \ref{s:theta} 
for a brief review 
and \cite[Chapter I]{tata} for a proper introduction.
(Our convention is slightly different than that in \cite{magic} but the 
formulas are equivalent.)

The next theorem provides a simple spectral characterization of $ \alpha$'s
satisfying \eqref{eq:defmal}. Combined with some symmetry reductions 
(see \S\S \ref{s:ham},\ref{s:num}) this characterization allows a precise calculation of the leading magic $ \alpha$'s { -- see Table \ref{tbl:magic} for the values of the first 13 elements of
$ \mathcal A_{\rm{mag}} $ and Tables \ref{tbl:error}, \ref{tbl:back} for rigorous error bounds.}
As seen in Proposition \ref{l:imp8}, it also implies that the multiplicities of flat bands at $ 0$
is at least 18.

\begin{figure}
\hspace*{-10pt}
\includegraphics[width=8.15cm]{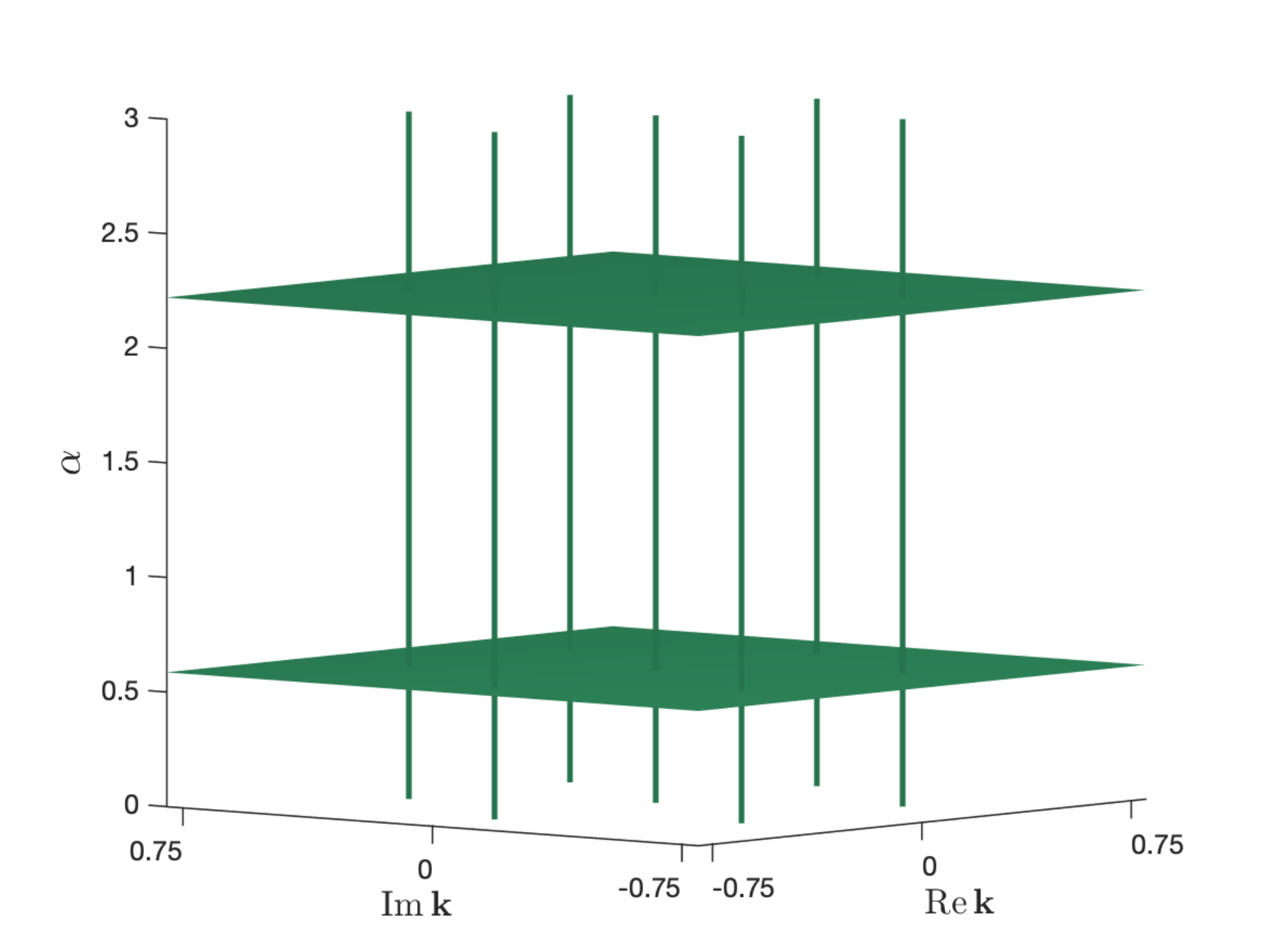}
\hspace*{-25pt}
\includegraphics[width=8.15cm]{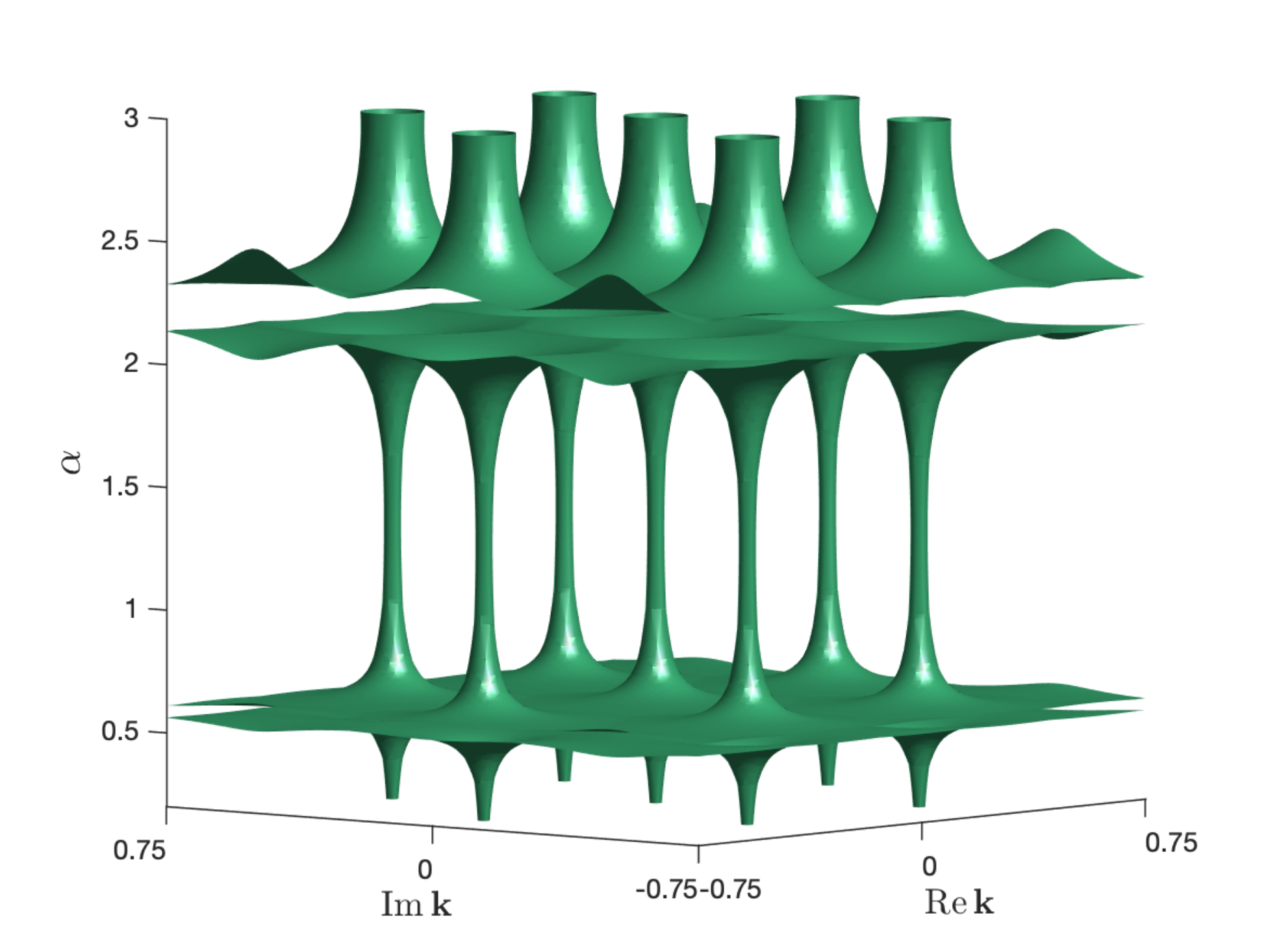} 
\caption{\label{f:psa} Left: spectrum of $ D ( \alpha ) $ 
as $ \alpha $ varies. Right: level surface of $ \mathbf k \mapsto \| ( D ( \alpha ) - \mathbf k )^{-1} \| = 10^2 $ 
as $ \alpha $ varies: we see that the norm of the resolvent 
 $  ( D ( \alpha ) - \mathbf k )^{-1} $ grows as we approach
the first two magic $ \alpha$'s (near $ 0.586$ and $2.221$), at which it blows up for
all $ k $. In any discretization that norm would be finite 
except on a finite set but it would blow up as the discretization improves.}
\end{figure}

\begin{theo}
\label{t:spec}
Let $ \Gamma^* $ be the dual lattice and define the family of compact operators
\begin{equation}
\label{eq:defT} 
T_{\mathbf k } := ( 2 D_{\bar z } - \mathbf k )^{-1} \begin{pmatrix}
0 & U ( z ) \\ U ( - z ) & 0 \end{pmatrix} , \ \ 
\mathbf k \notin  \Gamma^* ,
\end{equation}
where $ U( z) $ is given by \eqref{eq:defU}, or more generally satisfies
$ U \in C^\infty ( \CC/\Gamma ; \CC ) $ and \eqref{eq:symmU2}. 
Then the spectrum of $ T_{\mathbf k} $ is {\em independent} of
$ \mathbf k \notin \Gamma^* $, and the following statements are equivalent:
\begin{enumerate}
\item  $   1/\alpha \in \Spec_{ L^2 ( \CC /\Gamma ) } ( T_{\mathbf k } ) , \ \ \mathbf k \notin \Gamma^* $;
\item  $\Spec_{ L^2 ( \CC /\Gamma ) } D ( \alpha ) = \CC$ ;
\item  $ 0 \in \bigcap_{ \mathbf k \in \CC } 
\Spec_{ L^2 ( \CC/\Gamma ) } ( H_{\mathbf k } ( \alpha ) ) $, where
$ H_{\mathbf k } $ is defined in \eqref{eq:defmal}.
\end{enumerate}
\end{theo}

We denote the full set of {\em resonant} $ \alpha$'s and 
the set of {\em magic} $ \alpha $'s as
\begin{equation}
\label{eq:defResa}
\begin{gathered} 
\mathcal A := 1/ ( \Spec_{  L^2 ( \CC /\Gamma ) } ( T_{\mathbf k } ) 
\setminus \{ 0 \} ) , \ \ \mathbf k \notin \Gamma^* , \\
\kern6.25pt \mathcal A_{\rm{mag}} := \mathcal A \cap ( 0 , \infty ) = \{ \alpha_j \}_{j\geq 1}, \ \ \alpha_1 < \alpha_2 < \cdots ,
\end{gathered}
\end{equation}
respectively. The elements of $ \mathcal A $ are {\em included with 
their multiplicities} as multiplicities of eigenvalues of 
$ T_{\mathbf k}$. { Those multiplicities are at least $ 9 $ -- see Proposition \ref{l:imp8}.
Numerical evidence suggests that multiplicities of $ \mathcal A_{\rm{mag}} $ are exactly $ 9 $
and that is related to the question about zeros of $ u (\alpha ) $ -- see \eqref{eq:stack1}
and Remark 1 after Proof of Theorem \ref{t:magic} in \S \ref{s:alph}.}

As a simple byproduct of Theorems \ref{t:magic} and \ref{t:spec} we have
\[  \Spec_{ L^2 ( \CC /\Gamma ) } D ( \alpha ) = \Gamma^* , \ \ 
\alpha \notin \mathcal A . \]
Examples of operators which have either discrete spectra or all of $ \CC $
as spectrum, depending on analytic variation of 
coefficients, have been known before, see for instance Seeley \cite{see}.
The operator $ D ( \alpha ) $ provides a new striking example
of such  phenomena, showing that it is physically relevant and not merely 
pathological.

If we assume that
$ U ( z ) = \overline{ U ( \bar z ) } $, then Proposition \ref{p:sym} below (see also Figure~\ref{f:resa})
also gives 
$  \mathcal A = - \mathcal A = \overline {\mathcal A }$.

Mathematical description of $ \mathcal A $ remains open and
here we only contribute the following simple result:
\begin{theo}
\label{t:trace}
For the potential $ U $ given by \eqref{eq:defU} we have
\begin{equation}
\label{eq:ttrace}  \sum_{ \alpha \in \mathcal A  } \alpha^{-4} = { 72 \pi}
/{\sqrt 3 }, 
\end{equation}
where $ \alpha$'s are included according to their multiplicities. 
In particular, $ \mathcal A \neq \emptyset $.
\end{theo}

\begin{figure}
\begin{center}
\includegraphics[height=5.5cm]{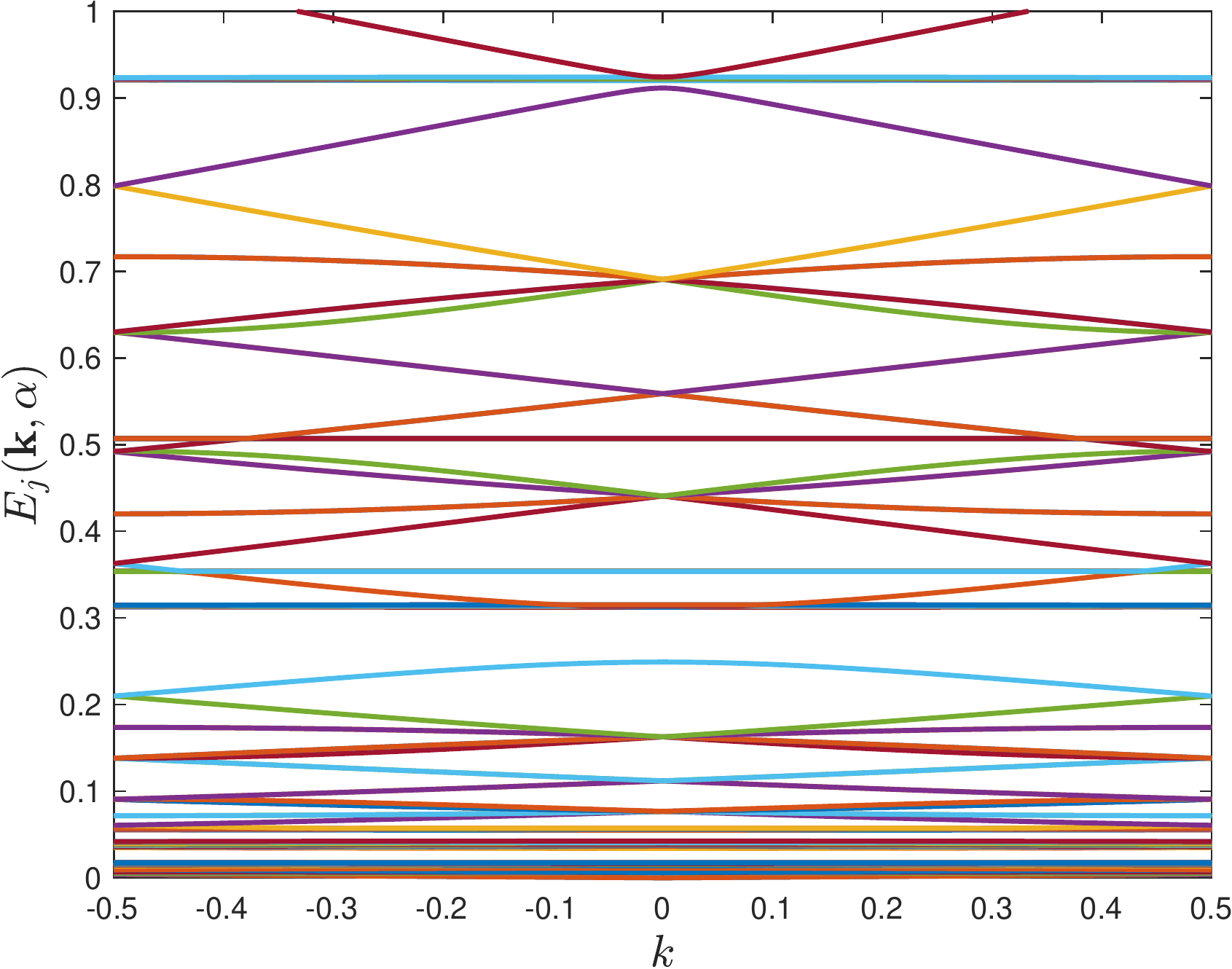}
\includegraphics[height=5.5cm]{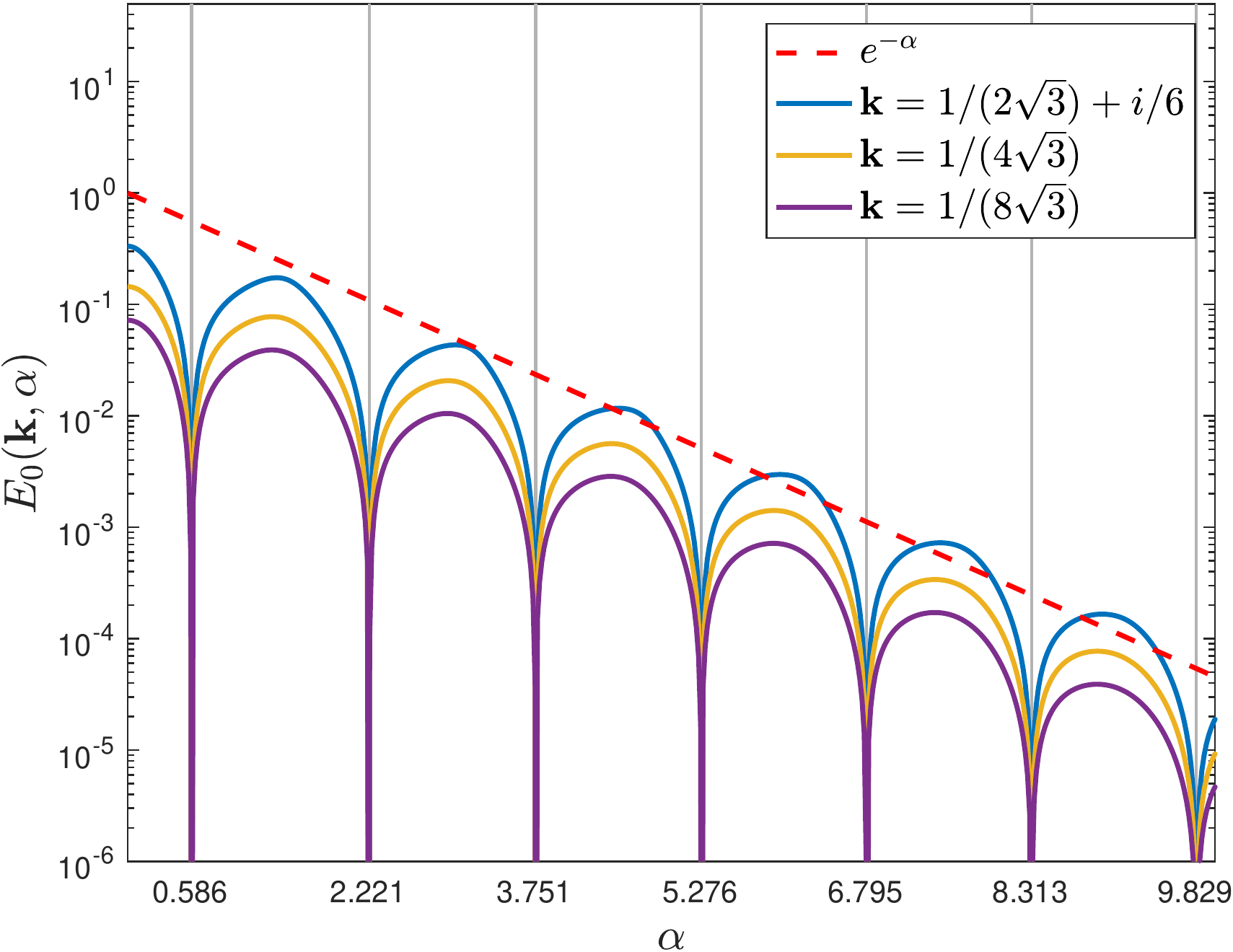}
\end{center}
\caption{\label{f:bands} On the left, the smallest non-negative eigenvalues of $ H_{\mathbf k } ( 
\alpha )$, $ \alpha = 5$, $ \mathbf k = k \omega/\sqrt 3$, $ -\frac12  \leq k \leq 
\frac12 $. On the right, $E_0({\mathbf k},\alpha)$ (log scale) for several values $ \mathbf k $.
(The point ${\mathbf k} = 1/(2\sqrt{3}) + i/6$ is farthest from an eigenvalue of $D(\alpha)$ for $\alpha\not\in{\mathcal A}$.)
The exponential squeezing of the bands described in Theorem \ref{t:squeeze}
is clearly visible.}
\end{figure}

Concerning $ \mathcal A_{\rm{mag}}$, an intriguing asymptotic relation for $ \alpha_j $'s for $ U $ given by \eqref{eq:defU} was suggested by the
numerics in \cite{magic}:
\begin{equation}
\label{eq:alphaj}  \alpha_{j+1} - \alpha_j \simeq \tfrac 3 2 , \ \ j \gg 1 . \end{equation}
 We do not address this problem here except numerically in \S \ref{s:num} and in Figure~\ref{f:mu}, which shows that
 regular spacing does not hold for general potentials.
 The following result based on  Dencker--Sj\"ostrand--Zworski \cite{dsz}  
indicates the mathematical subtlety underlying the distribution problem: for large values of 
$ \alpha $ the bands get exponentially squeezed, making it difficult to find the ones that are exactly zero; see Figure~\ref{f:bands} and the following
\begin{theo}
\label{t:squeeze}
Suppose that $ H_{\mathbf k} ( \alpha ) $ is given by 
\eqref{eq:defH} and \eqref{eq:defmal} with $ U $ given by 
\eqref{eq:defU} and that
\[  \Spec_{ L^2 ( \CC/\Gamma ) } H_{\mathbf k} ( \alpha ) = \{  E_j ( \mathbf k , \alpha ) \}_{ j \in \ZZ }  , \ \ \  
  E_{ j } ( \mathbf k , \alpha ) \leq E_{ j+1 } ( \mathbf k , \alpha ) ,  \ \  \mathbf k \in \CC , \ \ \alpha > 0 , \]
with the convention that $  E_0 ( \mathbf k , \alpha ) = 
\min_{ j} |E_j ( \mathbf k , \alpha ) | $.
Then there exist positive constants $c_0$, $c_1$, and $c_2$ such that
for all $ \mathbf k \in \CC $, 
\begin{equation}
\label{eq:lamb}  | E_j ( \mathbf k, \alpha) | \leq c_0 e^{ - c_1 \alpha } , 
\ \  | j | \leq c_2 \alpha, \ \  \alpha > 0 .
\end{equation}
\end{theo}
Numerical experiments presented in Figure~\ref{f:squeeze} (see also 
Figure~\ref{f:bands}) suggest that 
for any $ c_2 $ there exists $ c_0 $ for which \eqref{eq:lamb} holds, with 
$ c_1 = 1$.  The theorem is proved by showing that for large $ \alpha $
every point ``wants to be" in the spectrum of $ D ( \alpha ) $ modulo 
an exponentially small error. That is a typical {\em pseudospectral}
effect in the study of non-hermitian operators -- see 
Trefethen--Embree \cite{trem} for a broad description of such 
phenomena. Although $ H_{\mathbf k} ( \alpha ) $ is self-adjoint, having a zero eigenvalue is equivalent to 
$ \mathbf k \in \Spec_{{ L^2 ( \CC/\Gamma ) }}  ( D( \alpha ) )$ and $ D( \alpha ) $ is highly non-normal. {This is illustrated in Figure~\ref{f:psa}.}
In Section \ref{s:exp} we explore the situation for general potentials satisfying the symmetries \eqref{eq:symmU2}, and prove that a result corresponding to Theorem \ref{t:squeeze} continues to hold if an additional non-triviality assumption is imposed; see \eqref{eq:conditionbracket} and Theorem \ref{t:squeezegenpot}. 
(Some condition is clearly needed, as shown by the example of 
$ U \equiv 0 $.)

Watson and Luskin \cite{walu} have recently provided an alternative 
proof of Theorem \ref{t:magic} and implemented it numerically with precise 
error bounds. Assuming accuracy of singular value and polynomial calculations they proved
existence of $ \alpha_1 \in \mathcal A_{\rm{mag}} $, $ \alpha_1 \simeq 0.586 $.
Motivated by \cite{walu} we added error estimates for our calculations in 
\S \ref{s:error}. Assuming accuracy of singular value estimates for large sparse
matrices we show existence of $ \alpha_1 $ within $ 10^{-9}$ and $ \alpha_2 $
within $ 10^{-3} $ -- see Tables \ref{tbl:error} 
and \ref{tbl:back}. However, we do have high confidence in all digits shown
in Table \ref{tbl:magic}.

\section{Hamiltonian and its symmetries}
\label{s:ham}

In this section we discuss symmetries of $ D ( \alpha ) $ and 
$ H ( \alpha ) $ and prove basic results about their spectra.  

{Before entering mathematical  analysis of the model we provide a brief motivation for the Hamiltonian. 
Two basic symmetries are inherited from the honeycomb structure of the moir\'e lattice: a translation symmetry and a rotational symmetry by $2\pi/3$. In addition, the model exhibits a chiral symmetry which accounts for the massless and symmetric Dirac cones of the model that are preserved by the tunneling interaction.
The Dirac cones are effectively described by $2D$-massless Dirac operators. 
Therefore, the cones of two non-interacting sheets of graphene are described by a kinetic Hamiltonian 
\[ H_{\operatorname{kin}} = \operatorname{diag}(H_{\operatorname{Dirac}},H_{\operatorname{Dirac}})\text{, with } H_{\operatorname{Dirac}} =\begin{pmatrix}  0 & 2 D_z \\ 2 D_{\bar z} & 0 \end{pmatrix}.\]
Since honeycomb lattices are unions of two triangular lattices, we may distinguish between atoms of type $A$ and $B$. Considering then only the tunnelling interaction of atoms of different types between the layers gives rise to an off-diagonal tunnelling matrix 
\[ \tau(\alpha, z) = \begin{pmatrix} 0 & \alpha \overline{U(-z)} \\ \alpha U(z) & 0 \end{pmatrix}.\]
The tunnelling potential is then described by 
\[ H_{\operatorname{tun}}(\alpha) = \begin{pmatrix} 0 & \tau(\alpha,z) \\ \tau(\alpha,z)^* & 0 \end{pmatrix}.\]
Conjugating the sum of the two Hamiltonians by unitary operators yields, for $\sigma_1 = \begin{pmatrix} 0 & 1 \\ 1 & 0 \end{pmatrix},$
\[ H(\alpha) = \operatorname{diag}(1,\sigma_1,1) (H_{\operatorname{kin}}+ H_{\operatorname{tun}}(\alpha)) \operatorname{diag}(1,\sigma_1,1),\]
which is the operator introduced in \cite{magic} and studied in this article.}
\subsection{Symmetries of $ H ( \alpha ) $} \label{s:sH} 
The potential \eqref{eq:defU} satisfies the following properties:
\begin{equation}
\label{eq:symmU} 
\begin{gathered}  \mathbf a = \tfrac{4}3 \pi i \omega^\ell,  \ \ell =1,2 \ \Longrightarrow \ 
U ( z + \mathbf a ) = \bar \omega U ( z ) , \text{ and }\\
U ( \omega z )  = 
\omega U ( z )   .
\end{gathered}
\end{equation}
The first property in \eqref{eq:symmU} follows from the fact that (with $ k , \ell \in \ZZ_3 $)
\[ \tfrac12 ( {\mathbf a }  \bar \omega^k - \bar {\mathbf a } \omega^k )  =
 \tfrac23 \pi i (  \omega^{k-\ell} + \bar \omega^{k-\ell} ) = 
\left\{ \begin{array}{ll}  \tfrac43 \pi i \equiv - \tfrac23 \pi i \!\!\! \mod 2 \pi i , & k - \ell = 0; \\
\ \ \ \ -\tfrac23 \pi i,  & k - \ell \neq 0. 
\end{array} \right. \]
From this first property in \eqref{eq:symmU} we see that
\begin{equation}
\label{eq:Gam}
 U ( z + \gamma ) = U ( z) , \ \ \gamma \in 
\Gamma := 4 \pi \left( i \omega \ZZ \oplus i \omega^2 \ZZ \right). \end{equation}
The dual lattice consisting of $ \mathbf k $ satisfying 
 $ \tfrac12 (\gamma \bar{\mathbf k } + \bar \gamma \mathbf k ) \in 2 \pi \ZZ $ for $ \gamma\in \Gamma $, is given by 
$ \Gamma^* = \frac{ 1 } {\sqrt 3 } \left(  
  \omega  \ZZ  \oplus  {\omega^2}  \ZZ \right)$.

The second identity in \eqref{eq:symmU} shows that 
with 
$ L_{\mathbf a } \mathbf v ( z ) := \mathbf v ( z + \mathbf a ) $, 
\[  D(\alpha) L_{ \mathbf a } = 
L_{\mathbf a } \begin{pmatrix} 2   D_{\bar z } &  \omega \alpha U 
\\ \bar \omega \alpha U ( - \bullet ) & 2  D_{\bar z } \end{pmatrix} = 
 \begin{pmatrix}  \omega & 0 \\ 0 & 1 \end{pmatrix} 
L_{\mathbf a } D(\alpha) \begin{pmatrix} \bar  \omega & 0 \\ 0 & 1 \end{pmatrix} ,
\ \ \mathbf a = \tfrac{ 4}3 \pi i \omega^\ell , \ \ \ell = 1, 2 . \]
Hence,
\begin{equation}
\label{eq:defLa0} 
\mathscr L_{\mathbf a } D ( \alpha )  = D ( \alpha )  \mathscr L_{\mathbf a } , \ \ \ 
 {\mathscr L}_{\mathbf a } :=  \begin{pmatrix} \omega & 0 \\
0 & 1 \end{pmatrix} L_{\mathbf a } , \ \ \ 
\mathbf a = \tfrac{ 4}3 \pi i \omega^\ell , \ \ \ell = 1,2. \end{equation}
Putting
\begin{equation}
\label{eq:defGam3}
 \Gamma_3 := \Gamma/3 = \tfrac43 \pi ( i \omega \ZZ \oplus i \omega^2 \ZZ) , \ \ \ 
\Gamma_3/\Gamma \simeq \ZZ_3^2 , \end{equation}
and
\begin{equation*}
\label{eq:defLa} \mathscr L_{\mathbf a} := \begin{pmatrix}  \omega^{a_1 + a_2} & 0 \\ 0 & 1 \end{pmatrix}   L_{\mathbf a} , \ \ 
\mathbf a = \tfrac 43 \pi i ( \omega a_1 + \omega^2 a_2 ) , \end{equation*}
we obtain a unitary action of $\Gamma_3 $ on $ L^2 ( \CC ) $ or on 
$ L^2 ( \CC/\Gamma ) $, $ \Gamma_3 \ni \mathbf a \mapsto \mathscr L_{\mathbf a } $. 

We extend the action of $ \mathscr L_{\mathbf a } $ to 
$ L^2 ( \CC ; \CC^4 ) $ or $ L^2 ( \CC/\Gamma; \CC^4 ) $ block-diagonally
and we have $ \mathscr L_{\mathbf a} H ( \alpha ) = 
H ( \alpha ) \mathscr L_{\mathbf a } $. 

The second identity in \eqref{eq:symmU} shows that
$  [ D ( \alpha ) \mathbf u ( \omega \bullet ) ] ( z )  = \bar \omega 
[ D ( \alpha ) \mathbf u ] ( \omega z ) $. 
Hence, 
\begin{equation*}
\label{eq:CanD} 
\begin{gathered}
\mathscr C H ( \alpha ) = H ( \alpha ) \mathscr C ,   \ \ \ \ 
 \mathscr C \mathbf u ( z ) := \begin{pmatrix} 1 & 0 & 0 & 0 \\
 0 & 1 & 0 & 0 \\
 0 & 0 &  \bar \omega & 0 \\
0 & 0 & 0 &  \bar \omega \end{pmatrix} \mathbf u ( \omega z ) ,
\ \ \mathbf u \in  L^2 ( \CC ; \CC^4 ) .
\end{gathered}
\end{equation*}
Since $ \mathscr C \mathscr L_{\mathbf a } = \mathscr L_{\bar \omega \mathbf a } \mathscr C$, we combine the two actions into a unitary group action that
commutes with $ D ( \alpha ) $:
\begin{equation}
\label{eq:defG}  
\begin{gathered}  G := \Gamma_3 \rtimes \ZZ_3 , \ \ 
  \ZZ_3 \ni k : \mathbf a \to \bar \omega^k  \mathbf a , \ \ \ ( \mathbf a , k ) \cdot ( \mathbf a' , \ell ) = 
( \mathbf a + \bar \omega^k \mathbf a' , k + \ell ) ,
\\  ( \mathbf a, \ell ) \cdot \mathbf u = 
\mathscr L_{\mathbf a } \mathscr C^\ell \mathbf u . \ \
\end{gathered}
\end{equation}
By taking a quotient by $ \Gamma $ we obtain a finite group acting 
unitarily on $ L^2 ( \CC/\Gamma ) $ and commuting with $ H ( \alpha ) $:
\begin{equation}
\label{eq:defG3}
G_3 :=  G/\Gamma = \Gamma_3/\Gamma \rtimes \ZZ_3 \simeq \ZZ_3^2 \rtimes \ZZ_3. 
\end{equation}

By restriction to the first two components, $ G $ and $ G_3 $ act on 
$ L^2 ( \CC ; \CC ) $ and $ L^2 ( \CC/\Gamma; \CC^2 ) $ as well and we
use the same notation for those actions.

\noindent
{\bf Remark.} 
The group $ G_3 $ is naturally identified with the finite Heisenberg group
${\rm{He}}_3$:
\begin{gather*}   {\rm{He}}_3 := \left\{ \begin{pmatrix} 1 & x & t \\ 0 & 1 & y \\ 0 & 0 & 1 \end{pmatrix} , \ x,y,t \in \ZZ_3 \right\}, \\ 
\begin{pmatrix} 1 & x & t \\ 0 & 1 & y \\ 0 & 0 & 1 \end{pmatrix}
\begin{pmatrix} 1 & x' & t' \\ 0 & 1 & y' \\ 0 & 0 & 1 \end{pmatrix}
= \begin{pmatrix} 1 & x + x' & t + t' + x y' \\ 0 & 1 & y + y' \\ 0 & 0 & 1 \end{pmatrix} .
\end{gather*}
The identification of $ G_3 $ and $ {\rm{He}}_3 $ follows: with
$ \Gamma_3 / \Gamma \ni \mathbf a \mapsto  F ( \mathbf a ) := ( a_1 , a_2 ) \in \ZZ^2_3  $, $  \mathbf a = \tfrac 43 \pi i ( \omega a_1 + \omega^2 a_2 )$,
we have 
$ {\rm{He}}_3 \ni ( x, y , t ) \longmapsto (  F^{-1} ( t, y - t) , x ) \in G_3$. \qed

We record two more actions involving $ H ( \alpha ) $:
\begin{equation}
\label{eq:defW}
\begin{gathered}   H(\alpha) = - \mathscr W H(\alpha) \mathscr W^*, \ \ \ \mathscr W := \begin{pmatrix}
1 & 0 \\
0 & -1 \end{pmatrix} ,  
\ \ \ \mathscr W  \mathscr C  = 
 \mathscr C  \mathscr W , 
 \  
 \ \ \mathscr L_{\mathbf a } \mathscr W = \mathscr W \mathscr L_{\mathbf a },
\end{gathered}
\end{equation}
and   
\[  \mathscr Q H( \alpha) \mathscr Q^* =-H ( -\alpha ), \ \ \ 
\mathscr Q: = \operatorname{diag}(i,-i,-i,i), \ \ \ 
\mathscr Q  \mathscr C =   \mathscr C \mathscr Q, \ \ \
\mathscr Q \mathscr L_{\mathbf a } = \mathscr L_{\mathbf a } \mathscr Q.
 \]

We summarize these simple findings in
\begin{prop}
\label{p:specH}
The operator $ H ( \alpha ) : L^2 ( \CC; \CC^4 ) \to L^2 ( \CC; \CC^4 )$ is
an unbounded self-adjoint operator with the domain given by 
$ H^1 ( \CC; \CC^4 ) $. The operator $ H ( \alpha ) $ commutes with the 
unitary action of the group $ G $ given by \eqref{eq:defG} and 
\[ \Spec_{ L^2 ( \CC ) } H ( \alpha ) = 
- \Spec_{ L^2 ( \CC ) } H ( \alpha ) = \Spec_{ L^2 ( \CC ) } H ( - \alpha ) 
. \]
The same conclusions are valid when $ L^2 (\CC ) $ is replaced by $ L^2 (\CC/\Gamma) $
and $ G $ by $ G_3 $ given by \eqref{eq:defG3}. In addition, the spectrum is then discrete. 
\end{prop}

\subsection{Representation theory and protected states at \boldmath $ 0 $}
\label{s:rep}

Irreducible unitary representations of $ \ZZ_3^2 $ are one dimensional 
and are given by 
\begin{equation}
\label{eq:reprZ32}
\begin{gathered}   \pi_{\mathbf k} : \ZZ_3^2 \to {\mathsf U}(1) , \ \ \  \pi_{\bf k }(\mathbf a )  = e^{\frac i 2 ( \mathbf a \bar {\mathbf k } + 
\bar {\mathbf a} \mathbf k ) } ,\\ 
\mathbf a = \tfrac 43 \pi ( a_1 i \omega + a_2 i \omega^2 ) , \ \  \
a_j \in \ZZ_3 , \ \ \ 
\mathbf k = \tfrac{1}{\sqrt 3 } ( \omega^2 k_1 - \omega k_2 )  , \ k_j 
\in \ZZ_3, \\ 
\tfrac 1 2 ( \mathbf a \bar {\mathbf k } + 
\bar {\mathbf a} \mathbf k ) = \langle \mathbf a, \mathbf k \rangle =
\tfrac{2 \pi}3 (k_1 a_1 + k_2 a_2) .\end{gathered}
\end{equation}

Irreducible representations of $ G_3 $ are one dimensional for 
$ \mathbf k \in \Delta $ (given by $ \Delta(\ZZ_3 ) := \{ ( k, k ) , 
 k \in \ZZ_3 \}$ -- we note that $ \langle \mathbf k , 
 \omega \mathbf a \rangle = \langle \mathbf k , \mathbf a \rangle$, 
 $ \mathbf a \in \Gamma_3/\Gamma $, 
  if and
 only if $ \mathbf k \in \Delta  $), 
\[  \rho_{ k,p } ( ( \mathbf a , \ell ) ) = 
\bar \omega^{\ell p}  \pi_{(k,k)} ( \mathbf a ) ,\]
or three dimensional, for $ \mathbf k \notin \Delta  $:
\[   \rho_{\mathbf k}( ( \mathbf a , \ell ) ) 
= \begin{pmatrix} \omega^{ \langle \mathbf k , \mathbf a \rangle }
& 0 & 0 \\
0 & \omega^{\langle  \mathbf k , \omega \mathbf a \rangle }
 & 0 \\
0 & 0 &  \omega^{\langle  \mathbf k , \omega^2 \mathbf a \rangle }

\end{pmatrix} \begin{pmatrix} 0 & 1 & 0 \\
0 & 0 & 1 \\
1 & 0 & 0 \end{pmatrix}^{\ell}  \in {\mathsf U }( 3 ) . \]
The representations are equivalent for $ \mathbf k $ in the same orbit of
the transpose of $ \mathbf a \mapsto \omega \mathbf a $,
 and hence there are only two.

From this we see the well known fact that there are 11 irreducible representations: 9 one dimensional and 2 three dimensional.
 We can decompose 
$ L^2 ( \CC/\Gamma; \CC^4 ) $ into 11 orthogonal subspaces (since the groups are finite we do not have the usual Floquet theory difficulties!):
\[ L^2 ( \CC/\Gamma; \CC^4 ) = \bigoplus_{k, p \in 
\ZZ_3 }  L^2_{ \rho_{k,p}  } ( 
\CC / \Gamma; \CC^4 ) \oplus L^2_{  \rho_{(1,0)} } ( 
\CC / \Gamma; \CC^4 ) \oplus L^2_{  \rho_{(2,0)} } ( 
\CC / \Gamma; \CC^4 ).  \]
In view of Proposition \ref{p:specH} we have
\[ H_{k,p} ( \alpha ) := H ( \alpha )   : ( L^2_{ \rho_{k,p} }\cap H^1 )  (\CC / \Gamma; \CC^4 )  \to L^2_{ \rho_{k,p} } (\CC / \Gamma; \CC^4 )  ,\]
with similarly defined $ H_{ ( 1,0)} $ and $ H_{ ( 0,1)} $.

We now consider the case of $ \alpha = 0 $ and analyse 
$ \ker_{L^2 ( \CC/\Gamma ) }  H ( 0 ) $ decomposed into the corresponding representations:
\begin{equation*}
\label{eq:nullspace}
 \ker_{L^2 ( \CC/\Gamma ) } H( 0 ) = \{  \mathbf u = \mathbf e_j, \ j = 1, \dots , 4 \}, 
 \end{equation*}
where the $ \mathbf e_j $ form 
the standard basis elements of $ \CC^4 $. 
The action of $ G_3 = \ZZ_3^2 {\rtimes} \ZZ_3 $ is diagonal and, 
with $ \mathbf a = \frac 4 3 \pi ( a_1 i \omega + a_2 i \omega^2 ) $, 
\begin{gather*}  \mathscr L_{\mathbf a } \mathbf e_1 =  \omega^{a_1 + a_2 } \mathbf e_1 , \ \ \ \mathscr L_{\mathbf a } \mathbf e_2 =  \mathbf e_2 ,\ 
\ \ \mathscr L_{\mathbf a } \mathbf e_3 =  \omega^{a_1 + a_2 } \mathbf e_3, \ \ \
\mathscr L_{\mathbf a } \mathbf e_4 =  \mathbf e_4,
\\ \mathscr C  \mathbf e_1 = \mathbf e_1 , \ \ \
\mathscr C  \mathbf e_2 = \mathbf e_2 , \ \ \
\mathscr C  \mathbf e_3 = \bar \omega \mathbf e_3, \ \ \ 
\mathscr C  \mathbf e_4 = \bar \omega \mathbf e_4 .
 \end{gather*}
These observations imply that, with $ L^2_{\rho_{k,p} } := L^2_{\rho_{k,p} } 
 ( \CC/\Gamma; \CC^4 ) $, 
\[ \mathbf e_1 \in L^2_{ \rho_{1,0}} , \ \ \ 
\mathbf e_2 \in L^2_{ \rho_{0,0}} , \ \ \ 
\mathbf e_3 \in L^2_{ \rho_{1,1}}, \ \ \ 
\mathbf e_4 \in L^2_{ \rho_{0,1}} .
 \]
Hence for $ \alpha = 0 $, each of $ H_{0,0} (0) $, $ H_{1,0} (0) $, 
$ H_{0,1} (0) $ and $ H_{1,1} (0) $ has
a simple eigenvalue at $ 0 $. Since $ \mathscr W $ (see \eqref{eq:defW}) commutes with 
the action of $ G_3 $, the spectra of $ H_{k,\ell} ( \alpha ) $ are 
symmetric with respect to $ 0 $, it follows that $ H_{ k,\ell} ( \alpha ) $, 
$ k, \ell$ as above, 
each have an eigenvalue at $ 0 $.

Since 
$  \ker_{L^2 ( \CC/\Gamma ; \CC^4 ) } H ( \alpha ) = 
 \ker_{L^2 ( \CC/\Gamma ; \CC^2 ) } D( \alpha ) \oplus \{ 0_{\CC^2 } \} 
+ \{ 0_{\CC^2 } \} \oplus \ker_{L^2 ( \CC/\Gamma ; \CC^2 ) } D( \alpha )^* , 
$ 
we obtained the following result about a symmetry protected eigenstate at $ 0 $:
\begin{prop}
\label{p:prot}
For all $ \alpha \in \CC $, 
\[ \ker_{ L^2_{\rho_{ 1,0} } ( \CC/\Gamma; \CC^2 ) } D ( \alpha )  \neq
\{ 0 \} . \]
In the notation of \eqref{eq:defE},
$ \ker_{ L^2_{\rho_{ 0,0} } ( \CC/\Gamma; \CC^2 ) } D ( \alpha ) = \mathscr E 
 \ker_{ L^2_{\rho_{ 1,0} } ( \CC/\Gamma; \CC^2 ) } D ( \alpha ) \neq \{ 0 \} $. 
\end{prop}

\subsection{Floquet theory}
\label{s:flo}
Since the statement \eqref{eq:defmal} is interpreted as 
having a ``flat Floquet band'' at zero energy, we conclude this section with a brief account of Floquet theory. 

In principle, we could use the unitarity dual of $ G$ defined in 
\eqref{eq:defG} (and described similarly to the unitary dual of $ G_3 $
in \S \ref{s:rep}) and decompose $ L^2 ( \CC ) $ into
irreducible representations under the action of $ G $.
However, let us take the standard Floquet theory approach based on invariance
under $ \Gamma $ (see \eqref{eq:Gam})
\begin{gather*}  \Gamma \ni \mathbf a :  \psi \longmapsto \mathscr L_{\mathbf a } \mathbf \psi ( z)  = \mathbf \psi ( z + \mathbf a ) , \ \ 
\psi \in L^2 ( \CC ; \CC^2 ) , \ \ 
D ( \alpha )  \mathscr
L_{\mathbf a } = \mathscr L_{\mathbf a } D ( \alpha ) . 
 \end{gather*}
(This definition agrees with \eqref{eq:defLa0} when $ \mathbf a \in \Gamma$.) 

 We start by recording basic properties of the operator $ D ( \alpha ) $.
We first observe that 
\begin{equation}
\label{eq:spec0}
\Spec_{{ L^2 ( \CC/\Gamma ) } } D ( 0 ) = \Gamma^*, \ \ \ D ( 0 )  e_{\mathbf k } \mathbf e_j = 
\mathbf k e_{\mathbf k } \mathbf e_j , \ \ 
e_{\mathbf k} ( z) := 
e^{ \frac i 2 ( 
\bar{\mathbf k } z + \mathbf k \bar z ) }, \  \   \mathbf k \in \Gamma^*, 
\ \ j = 1, 2 , 
\end{equation}
where the exponentials $ e_\mathbf k / \vol( \CC/\Gamma)^{\frac12} $ form an 
orthonormal basis of $ L^2 ( \CC/\Gamma ) $ and $ \mathbf e_j $ are
the standard basis of $ \CC^2 $.

We then have the following simple  
\begin{prop}
\label{p:slight}
The family $ \CC \ni \alpha \mapsto D ( \alpha ) : H^1 ( \CC/\Gamma ; 
\CC^2 ) \to L^2 ( \CC/\Gamma; \CC^2 ) $
is a holomorphic family of elliptic Fredholm operators of index $ 0 $,
and for all $ \alpha $, { the spectrum of $D(\alpha)$ is $\Gamma^*$-periodic: }
\begin{equation}
\label{eq:per}
\Spec_{ L^2 ( \CC/ \Gamma ) } D ( \alpha ) = 
\Spec_{ L^2 ( \CC/ \Gamma ) } D ( \alpha ) + \mathbf k , \ \ 
\mathbf k \in \Gamma^* . 
\end{equation}
 \end{prop}
\begin{proof}  
Since $ D_{\bar z } $ is an elliptic operator in dimension 2,
existence of parametrices (see for instance \cite[Proposition E.32]{res}) immediately shows the Fredholm property (see for instance 
\cite[\S C.2]{res} for that and other basic properties of Fredholm operators). In view of \eqref{eq:spec0},  $  D ( 0 ) - \mathbf k $
is invertible for  $ \mathbf k \notin \Gamma^* $ and hence
 $ D ( 0 ) : H^1 ( \CC/\Gamma ) 
\to L^2 ( \CC/\Gamma ) $ is an operator of
index $ 0 $. The same is true for the Fredholm family
$ D ( \alpha ) $.
To see \eqref{eq:per}, note that 
if $ ( D ( \alpha ) - \lambda ) \mathbf u = 0 $ 
then $ ( D ( \alpha ) - ( \lambda + \mathbf k ) ) ( e_{\mathbf k} 
\mathbf u)  = 0 $, $ \mathbf k \in \Gamma^* $.
\end{proof}

For $ \mathbf k \in \CC/ \Gamma^* $ (or simply $ \mathbf k \in \CC $)
we defined the Floquet boundary condition as 
\[  \mathbf \psi  ( z + \mathbf a ) = 
e^{ - \frac i2 ( \mathbf a \bar{ \mathbf k} + 
\bar{\mathbf a } \mathbf k ) } \psi ( z ) , \ \
\ \psi  \in L^2_{\rm{loc}} ( \CC ; \CC^2 ) , \ \ \ 
\mathbf a \in \Gamma . \]
This means that
\[  \mathbf v ( z ) := e^{ \frac i 2 ( z \bar {\mathbf k}
+ \bar z \mathbf k ) } \psi ( z ) \]
satisfies 
\[  \mathbf v( z + \mathbf a ) = \mathbf v (z), \ \ \mathbf a \in \Gamma, \ \    e^{  \frac i 2 ( z \bar {\mathbf k}
+ \bar z \mathbf k ) }   D( \alpha ) \psi ( z ) =  ( D ( \alpha ) - \mathbf k ) 
\mathbf v ( z ) . \]
It follows that
\begin{equation}
\label{eq:Dba}   e^{  \frac i 2 ( z \bar {\mathbf k}
+ \bar z \mathbf k ) }  H ( \alpha ) e^{  \frac i 2 ( z \bar {\mathbf k}
+ \bar z \mathbf k ) } 
=  H_{\mathbf k } ( \alpha ) := \begin{pmatrix}  0 & D ( \alpha )^* - \bar {\mathbf k} \\
D( \alpha ) - {\mathbf k } & 0 \end{pmatrix} ,
\end{equation}
where $ H_{\mathbf k} ( \alpha ) $ is the operator in \eqref{eq:defmal}.

We now proceed with standard Floquet theory and introduce the 
unitary transformation
\[  \mathscr U : L^2 ( \CC ; \CC^4 ) \to L^2 ( \CC/ \Gamma^* ; L^2 ( \CC/ \Gamma ) ), \ \
\mathscr U  \mathbf u ( \mathbf k  , z ) := \sum_{\mathbf a \in \Gamma } 
u ( z + \mathbf a ) e^{ \frac i 2 (  ( z + \mathbf a ) \bar{\mathbf k } 
+ ( \bar z + \bar { \mathbf a } ) \mathbf k) } . \]
We then have
\[  \mathscr U H \mathscr U^* \mathbf v ( z , \mathbf k ) =  
H_{\mathbf k } \mathbf v ( z , \mathbf k ) , \ \ \ 
  \mathbf v ( \bullet, \mathbf k ) \in C^\infty ( \CC / \Gamma ; \CC^4 ) , \]
that is, for a fixed $ \mathbf k \in \CC / \Gamma^* $, 
$ \mathscr U H \mathscr U^* $ acts on {\em periodic functions with respect to $ \Gamma $} 
as the operator in \eqref{eq:Dba}. 
For each $ \mathbf k $, the operator $ H_{\mathbf k} ( \alpha ) $
is an elliptic differential system (see Proposition \ref{p:slight} above) and hence it has a discrete spectrum that then describes the spectrum of $ H( \alpha ) $ on $ L^2 ( \CC ) $:
\begin{equation}
\label{eq:specH} 
\begin{gathered}    \Spec_{L^2 ( \CC) }  ( H ( \alpha ) ) = \bigcup_{ \mathbf k \in \CC / \Gamma^* }
\Spec_{L^2 ( \CC/ \Gamma )}  ( H_{\mathbf k } ( \alpha ) ) , \\ 
\Spec_{L^2 ( \CC/ \Gamma )} ( H_{\mathbf k } ( \alpha ) ) = \{ \pm E_j 
( {\mathbf k } , \alpha ) \}_{ j=0}^\infty, \ \  E_{ j+1 } ( \mathbf k, \alpha ) \geq E_j ( \mathbf k, \alpha ) 
\geq 0 . \end{gathered} \end{equation}
To see the last statement we recall that 
\[ (\lambda - \mathscr A )^{-1} = \begin{pmatrix}  ( \lambda^2 - A^* A )^{-1} & 0 
\\ 0 & ( \lambda^2 - A A^* )^{-1} \end{pmatrix} \begin{pmatrix}
 \lambda & A^* \\
 A & \lambda \end{pmatrix}, \ \  \mathscr A := \begin{pmatrix}
 0 & A^* \\ A & 0 \end{pmatrix} . \]
Hence, the non-zero eigenvalues of $ H_{\mathbf k } $ are given by 
$ \pm $ the  non-zero singular values
of $ D( \alpha) + \mathbf k $ (that is, the eigenvalues of 
$ [ ( D ( \alpha ) + \mathbf k )^* ( D( \alpha ) + \mathbf k) ]^{\frac12} $), included according to their multiplicities). 
We need to check that the eigenvalue $ 0 $ of 
$ ( D ( \alpha ) + \mathbf k )^* ( D( \alpha ) + \mathbf k)$ has the 
same multiplicity as the zero eigenvalue of
$ ( D ( \alpha ) + \mathbf k ) ( D( \alpha ) + \mathbf k)^* $, 
so that eigenvalues $ E_j (  \mathbf k, \alpha ) = 0$ are included
exactly twice (for $ \pm $).

 For that we use Proposition \ref{p:slight}, which also shows that 
 $ D( \alpha ) + \mathbf k $ is a Fredholm operator of order zero,
 and hence 
 \[ \dim \ker_{{ L^2 ( \CC/\Gamma ; \CC^2 ) }} ( D ( \alpha)  + \mathbf k ) = \dim \ker_{
 { L^2 ( \CC/\Gamma ; \CC^2 ) }} ( D ( \alpha )^*  + \bar{\mathbf k} ) . \]
 In \eqref{eq:specH} we abuse notation by counting
$ \pm 0 $ twice in the spectrum of $ H_{\mathbf k } ( \alpha ) $.

From this discussion we can re-interpret \eqref{eq:defmal} as the existence of a flat band:

\begin{prop}
\label{p:flat}
In the notation of \eqref{eq:defmal} and \eqref{eq:specH}
\begin{equation}
\label{eq:flat} 0 \in \bigcap_{ \mathbf k \in \CC } \Spec_{{ L^2 ( \CC/\Gamma, \CC^4) }}  H_{\mathbf k}( \alpha ) 
\ \Longleftrightarrow \ E_0 ( \mathbf k , \alpha ) = 0 \text{ for all
$ \mathbf k \in \CC/\Gamma^* $.} \end{equation}
\end{prop}

\section{Resonant and magic angles}
\label{s:alph}

We now want to obtain a computable condition on $ \alpha $ guaranteeing \eqref{eq:defmal}, that is, the flatness of a band \eqref{eq:flat}.
In view of \eqref{eq:Dba} and \eqref{eq:specH}, \eqref{eq:defmal} is equivalent to 
$ \Spec_{ L^2 ( \CC/\Gamma ) } D ( \alpha ) = \CC $.

\subsection{Spectrum of $ D ( \alpha ) $}

To investigate the spectrum of $ D ( \alpha ) $ we use the operator
$ T_{\mathbf k } $ defined in \eqref{eq:defT}. We note that for 
$ \mathbf k \notin \Gamma^* $, \eqref{eq:spec0} shows that
\begin{equation}
\label{eq:D2T} 
D ( \alpha ) - \mathbf k = ( D ( 0 ) - \mathbf k ) ( I + 
\alpha T_{\mathbf k } ) , \ \ \ D ( 0 ) = 2 D_{\bar z } . 
\end{equation} 
The operator $ T_{\mathbf k } : L^2 ( \CC /\Gamma ; \CC^2 ) 
\to L^2 ( \CC / \Gamma ; \CC^2 ) $ is compact and hence its spectrum 
can only accumulate at $ 0 $. This means that
\begin{equation}
\label{eq:D2Ak}   \Gamma^* \not \ni \mathbf k  \in \Spec_{ L^2 ( \CC/\Gamma) } 
D ( \alpha ) \ \Longleftrightarrow \ \alpha \in \mathcal A_{\mathbf k}, \  \ \ 
\mathcal A_{\mathbf k} := 1/(\Spec_{{ L^2 ( \CC/\Gamma ) }}  ( T_{\mathbf k} ) \setminus \{ 0 \} ), 
\end{equation}
where $ \mathcal A_{\mathbf k } $ is a discrete subset of $ \CC $.

We now have a proposition proving the first part of Theorem \ref{t:spec}.
It also defines the family of functions appearing in Theorem \ref{t:magic}.
\begin{prop}
\label{p:Ak} 
For $ \mathbf k \notin \Gamma^* $, 
the discrete set $ \mathcal A = \mathcal A_{\mathbf k} $ 
is {\em independent} of $ \mathbf k $ and 
\begin{equation}
\label{eq:SpecD}
\Spec_{ L^2 ( \CC / \Gamma ) }  ( D ( \alpha ) ) = \left\{ 
\begin{array}{ll} \Gamma^*, & \alpha \notin \mathcal A; \\
\CC, & \alpha \in \mathcal A . \end{array} \right. 
\end{equation}
Moreover, for all $ \alpha \notin \mathcal A $, 
\begin{equation}
\label{eq:kerDa}
\ker_{ L^2 ( \CC/\Gamma ; \CC^2 ) } D ( \alpha ) = \CC \mathbf u ( \alpha ) \oplus \CC \mathscr E \mathbf u ( \alpha ) , \ \ 
\mathbf u ( \alpha ) \in L^2_{ \rho_{1,0}} ( \CC/\Gamma; \CC^2 ) , \ \
\mathbf u ( 0 ) = \mathbf e_1 , \end{equation}
where $ \mathscr E $ is defined in \eqref{eq:defE}  and $ \mathbf e_1 = 
( 1, 0 )^t $. For $ \alpha \in \RR $,  
$ \mathbf u $ extends to a real analytic family, 
 $\RR \ni \alpha \mapsto \mathbf u ( \alpha ) 
\in \ker_{ L^2_{\rho_{1,0}}  ( \CC/\Gamma ; \CC^2 )} D( \alpha ) $.
\end{prop}
\begin{proof}
Suppose $ \alpha \in \CC \setminus \mathcal A_{\mathbf k } $, $ \mathbf k \notin 
\Gamma^* $. Then 
$ ( D ( \alpha ) - \mathbf k )^{-1} : L^2 ( \CC/\Gamma ) \to 
H^1 ( \CC/\Gamma ) \hookrightarrow L^2 ( \CC/\Gamma ) $ is a compact operator
and hence $ D ( \alpha ) $ has discrete spectrum. By Proposition \ref{p:prot}, $ 0  \in \Spec_{{ L^2 ( \CC/\Gamma ) } }(D(\alpha))$ for all $\alpha \in \CC$, and thus together with the periodicity condition \eqref{eq:per} this implies $\Spec_{{ L^2 ( \CC/\Gamma ) }} (D(\alpha))\supset \Gamma^*.$
Recall now that $D(\alpha)$ depends on $\alpha$ holomorphically and $0$ is isolated in the spectrum for $\alpha \notin \mathcal A_{\mathbf k}$. Thus, $ \ker_{ L^2 ( \CC/\Gamma ; \CC^2 ) } D ( \alpha ) $ depends holomorphically on $\alpha \notin \mathcal A_{\bf k}$ \cite[VII. Theorem~$1.7$]{kato} and by Proposition \ref{p:prot} $\operatorname{dim}(\ker_{ L^2 ( \CC/\Gamma ; \CC^2 ) } D ( \alpha )) \ge 2$ for all $\alpha \in \mathbb{C}$, we find
\[\operatorname{dim}(\ker_{ L^2 ( \CC/\Gamma ; \CC^2 ) } D ( \alpha ))=\operatorname{dim}(\ker_{ L^2 ( \CC/\Gamma ; \CC^2 )} D ( 0 ))=2{ \text{ for all } \alpha \notin \mathcal A_{\bf k}}.\] 

The discreteness of the spectrum implies that the spectrum depends continuously on $\alpha$ \cite[II. §6]{kato} for $\alpha \notin \mathcal A_{\mathbf k }$. Since $\operatorname{dim}(\ker_{ L^2 ( \CC/\Gamma ; \CC^2 ) } D ( \alpha ))=2$ for all $\alpha  \notin \mathcal A_{\mathbf k }$ and by periodicity \eqref{eq:per}, this implies that $\Spec_{{ L^2 ( \CC/\Gamma, \CC^2  ) } }(D(\alpha))= \Gamma^*.$

Using \eqref{eq:D2Ak} and that $\Spec_{{  L^2 ( \CC/\Gamma, \CC^2 ) } }(D(\alpha))=\Gamma^*$ for all $\alpha \notin \mathcal A_{\mathbf k}$, it follows that 
\[ \exists \, \mathbf k \notin \Gamma^* \ \text{ such that } \ \alpha \notin \mathcal A_{\mathbf k } \ \Longrightarrow \ \forall 
\, \mathbf p \notin \Gamma^* \ \text{ we have } \ 
\alpha \notin \mathcal A_{\mathbf p } . \] 
This shows independence of $ \mathcal A_{\mathbf k } =: 
\mathcal A $ of $ {\mathbf k}.$

Since 
\[ \CC \ni \alpha \mapsto \widetilde H(\alpha) := \begin{pmatrix} 0 & D(\bar \alpha ) ^*
 \\ D( \alpha ) & 0\end{pmatrix} , \ \     \widetilde H( \alpha )  = 
 H ( \alpha ) , \ \ \alpha \in \RR , \] 
is a holomorphic operator family with compact resolvents, 
self-adjoint for $\alpha\in \RR$, Rellich's theorem \cite[VII. Theorem $3.9$]{kato} implies that all eigenvalues and eigenfunctions of $H(\alpha )
= \widetilde H ( \alpha ) $ { can be chosen to depend real-analytically on $ \alpha \in \RR$.}
If we let $\varphi(\alpha) := ( \mathbf u(\alpha),0,0)^t \in L^2_{\rho_{1,0}}$,  $ \alpha \in \RR \setminus \mathcal A  $,  
then $\varphi(0)= \mathbf e_1 \in \CC^4$ and by the discussion above 
$ \varphi ( \alpha ) $ extends to a real analytic family for all $ \alpha
\in \RR $. 
\end{proof}

The next proposition provides the symmetries of the set $ \mathcal A$. 
\begin{prop}
\label{p:sym}
Suppose that in addition to \eqref{eq:symmU} we have 
$ U ( z ) = \overline{ U ( \bar z  ) }$. 
Then, $  \Spec  D ( \alpha ) = \Spec D ( - \alpha ) = \Spec D ( \bar \alpha ) $ and hence
\begin{equation*}
\label{eq:sym}
\mathcal A = - \mathcal A = \overline {\mathcal A }.  
\end{equation*}
In these statements $\Spec$ can be 
either the spectrum on $ {L^2 ( \CC ) }$, $ \Spec_{L^2 ( \CC ) } $, or
on $ {L^2 ( \CC/\Gamma ) }$, $ \Spec_{L^2 ( \CC /\Gamma) }$.

\end{prop}
\begin{proof}
To see the symmetries of the spectrum, we note that since
$ Q \mathbf v ( z ) = \overline{ \mathbf v ( - z ) },$ the anti-linear
involution satisfies
\[
 D ( \alpha) Q \mathbf v = -Q D( - \alpha)^* \mathbf v ,\]
which in turn implies
$ \Spec D ( \alpha ) = - \overline {\Spec } D ( - \alpha )^* = - 
{ \Spec } D ( - \alpha ) $.
But then \eqref{eq:SpecD} shows that $ \Spec  D ( \alpha ) = \Spec D ( - \alpha ) $.

Next we notice that $\overline{U(\bar z)}=U(z)$. If we define the unitary map $F \mathbf v(z):=\overline{ \mathbf v(\bar z)}$,
then we find using
$(D_{\bar z} F \mathbf v)(z)=(D_z \overline{\mathbf v})(\bar z) =- (\overline{D_{\bar z}  \mathbf v})(\bar z)=-(FD_{\bar z} \mathbf v)(z)$
the relation
\[D(\alpha)(F \mathbf v)=-F(D(-\bar \alpha) \mathbf v),\] 
which implies that $\Spec(D(\alpha))=-\Spec(D(-\bar \alpha))=\Spec(D(\bar \alpha)).$
\end{proof}

The description of the kernel of $ D( \alpha ) $ gives us an expression 
for the inverse of $ D ( \alpha ) - \mathbf k $, $ \mathbf k \notin 
\Gamma^* $ and $ \alpha \notin \mathcal A$. We start with 
the following simple
\begin{prop}
\label{p:inverse}
Suppose that $ \mathbf u ( \alpha ) $ is given in \eqref{eq:kerDa} and define a two-by-two matrix 
\begin{equation*}
\label{eq:defVa}   \mathbf V ( \alpha ) := [ \mathbf u ( \alpha ) , \mathscr E \mathbf u ( \alpha ) ] ,   \ \ v ( \alpha ) := \det \mathbf V ( \alpha ) . \end{equation*}
Then $ v( \alpha ) \neq 0 $ and 
$ \mathbf k \notin \Gamma^* $ imply that{, with the cofactor matrix denoted by $\operatorname{adj}$,}
\begin{equation}
\label{eq:inv}
( D ( \alpha ) - \mathbf k )^{-1} = \frac{1}{ v ( \alpha ) }  
 {\rm{adj}} (\mathbf V ( \alpha ))  ( 2 D_{\bar z } - \mathbf k )^{-1}( \mathbf V ( \alpha ) ) . 
\end{equation}
For a fixed $ \mathbf k \notin \Gamma^*$, 
$ \alpha \mapsto ( D ( \alpha ) - \mathbf k )^{-1} $ is a meromorphic 
family of compact operators with poles of finite rank at $ \alpha \in \mathcal A $.
\end{prop}
\begin{proof}
If $ v ( \alpha ) \neq 0 $, then $ \mathbf V ( \alpha )^{-1} 
= {\rm{adj}} \mathbf V ( \alpha ) / v ( \alpha ) $ and
 \eqref{eq:inv} follows from a simple calculation 
($ \mathbf V ( \alpha ) $ provides a matrix-valued integrating factor).
In view of \eqref{eq:D2T},
\[ ( D ( \alpha ) - \mathbf k )^{-1} = ( I + \alpha T_{\mathbf k})^{-1}
( D( 0 ) - \mathbf k )^{-1} , \]
where, using analytic Fredholm theory (see for instance 
\cite[Theorem C.8]{res}), $ \alpha \mapsto ( I + \alpha T_{\mathbf k} )^{-1} $ is a meromorphic family of operators with poles of finite rank.
\end{proof} 

The proposition shows that $ \alpha \in \mathcal A $ implies that
$ v ( \alpha ) = 0$. To obtain the opposite implication (which then gives
Theorem \ref{t:magic}) we will use the theta function argument from 
\cite{magic}.

\begin{figure}
\begin{center}
\includegraphics[width=7.5cm]{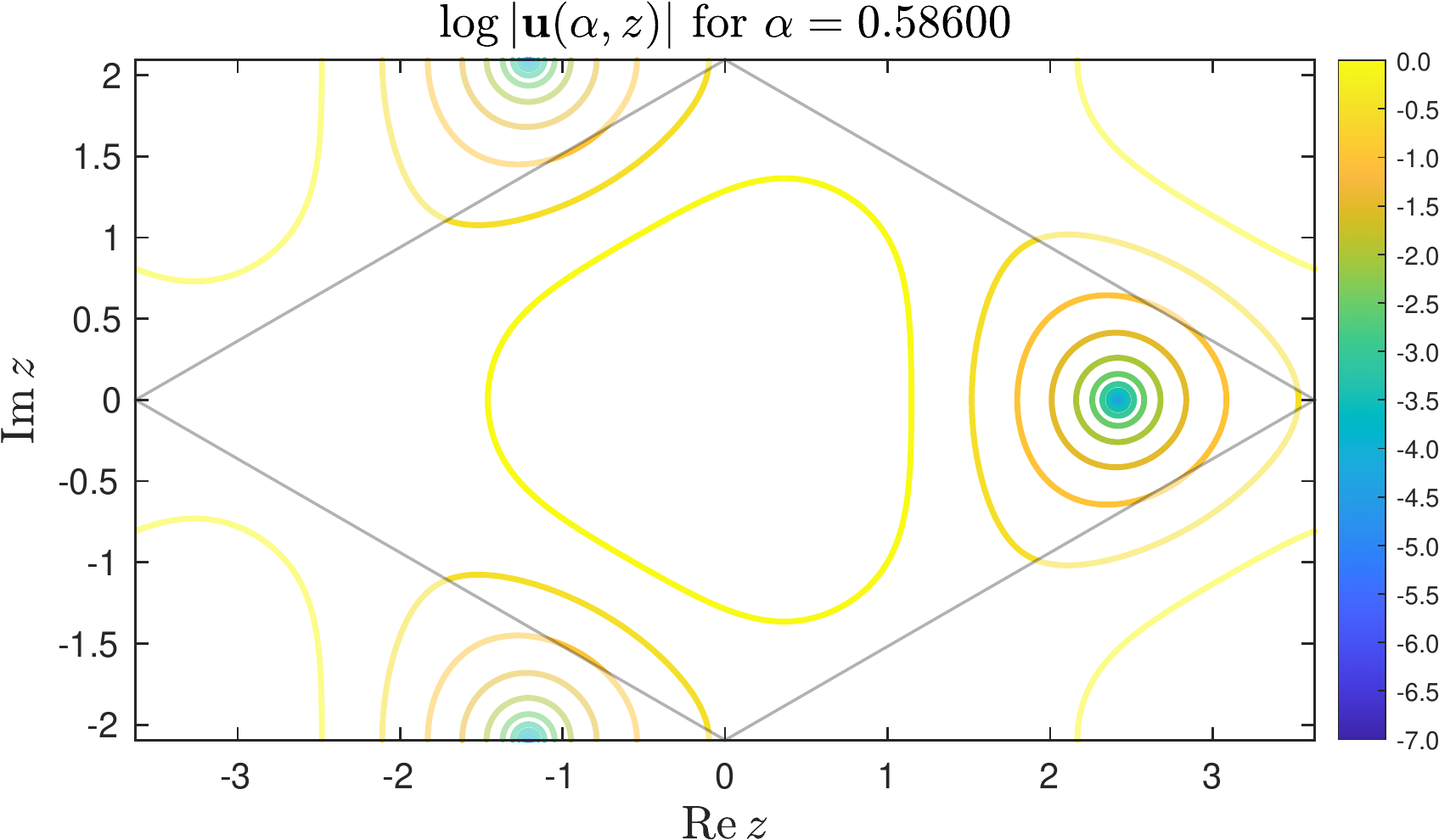}\quad \includegraphics[width=7.5cm]{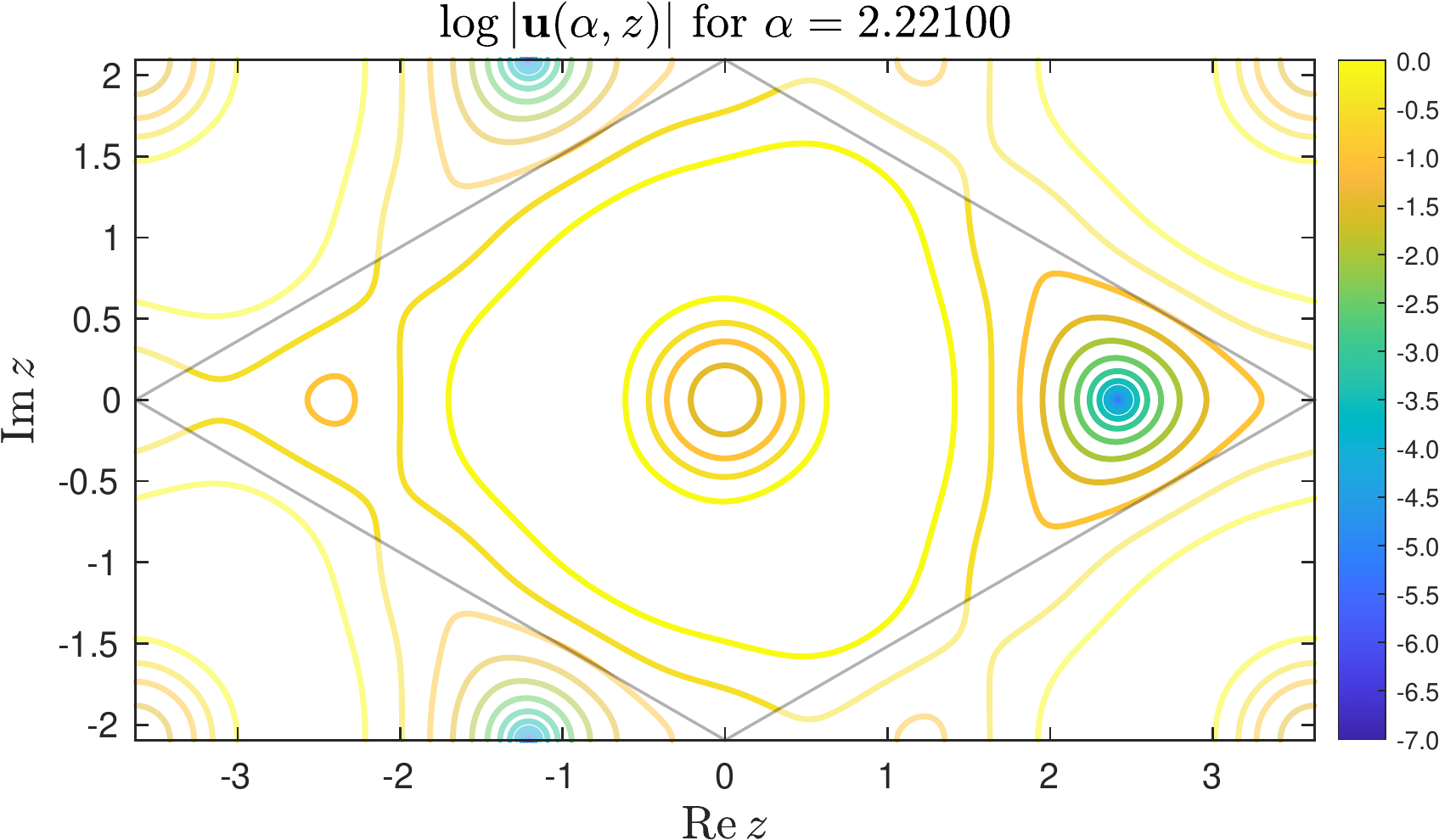}\\[5pt]
\includegraphics[width=7.5cm]{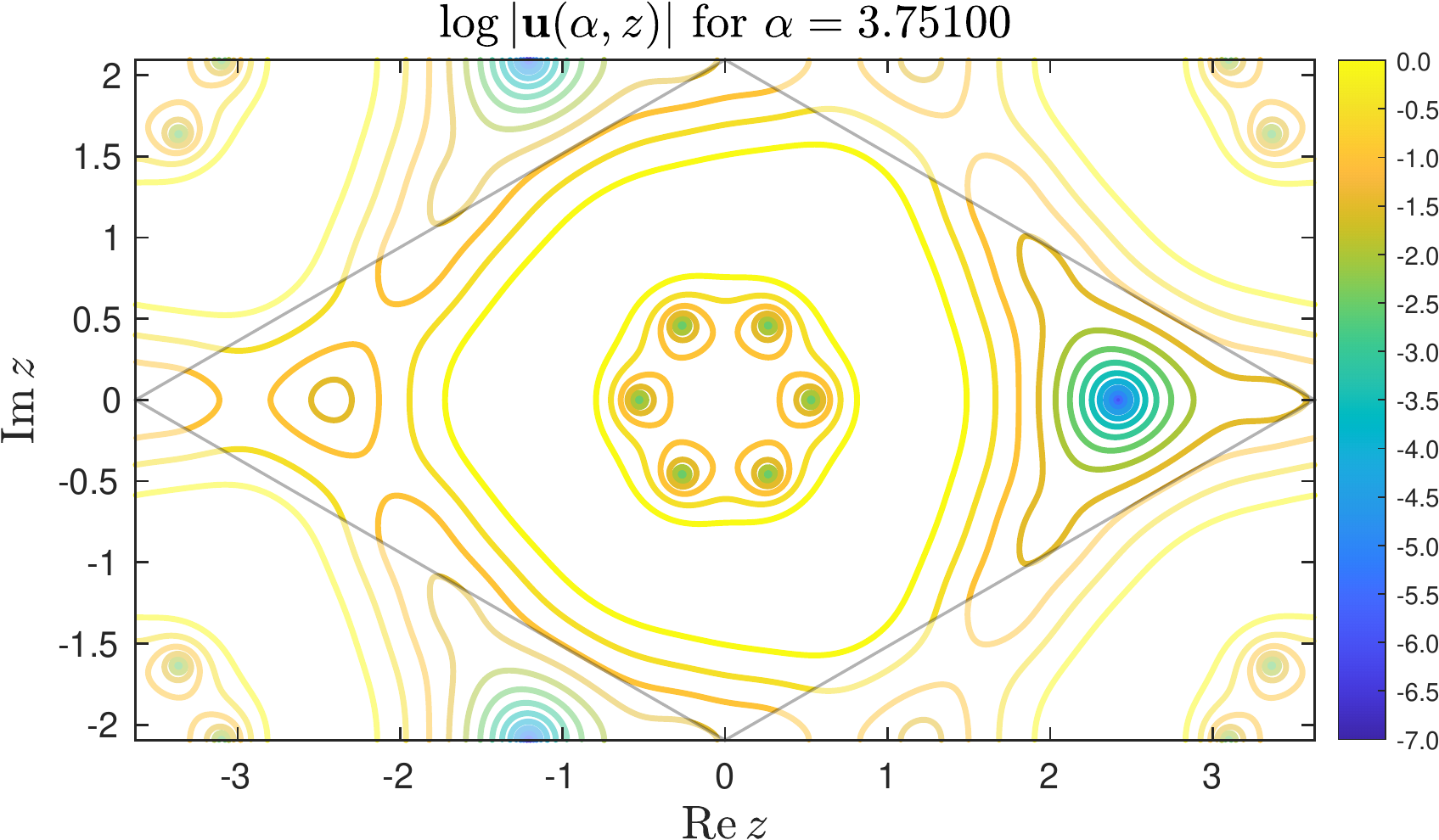}\quad \includegraphics[width=7.5cm]{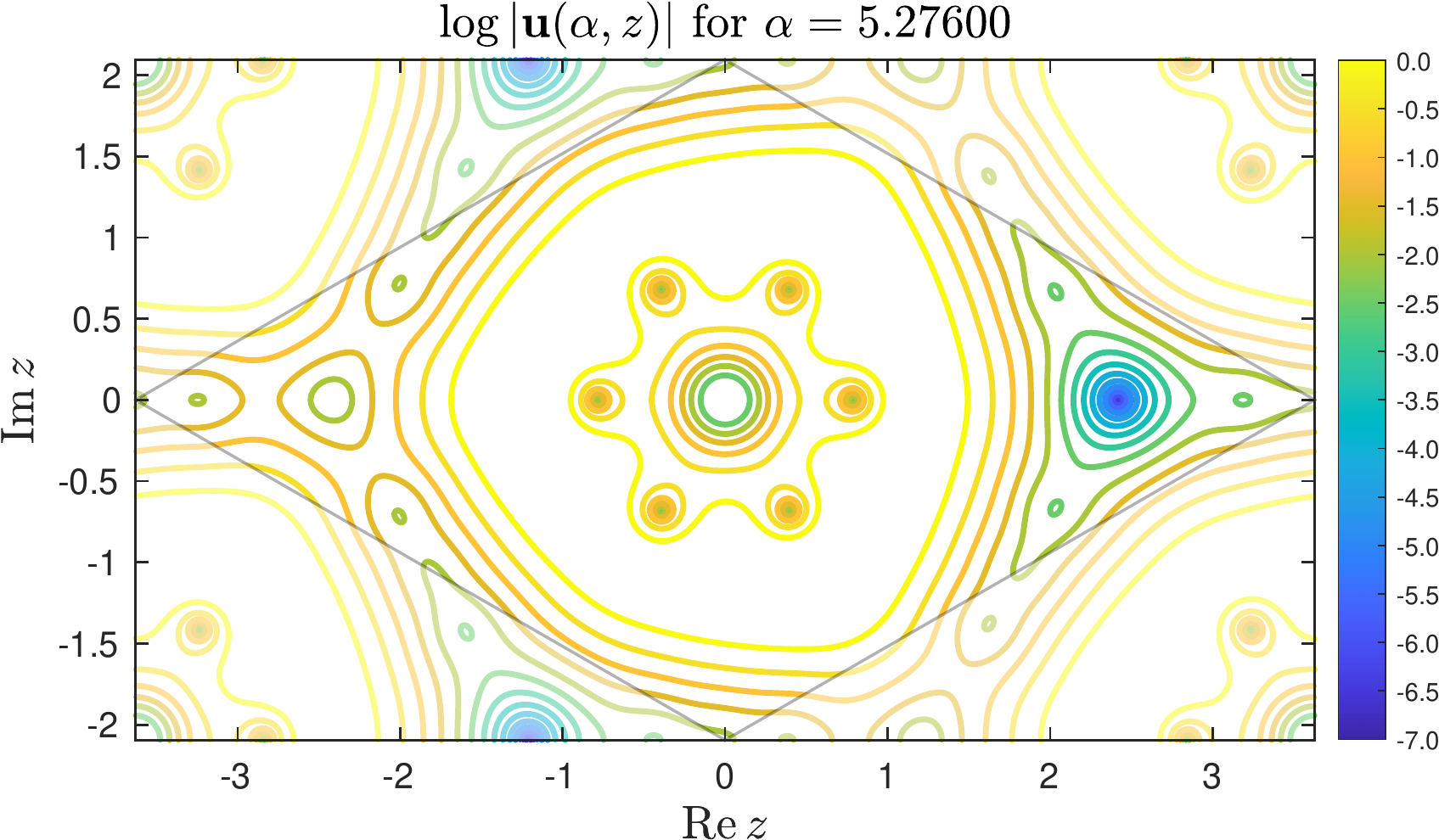}\\[5pt]
\includegraphics[width=7.5cm]{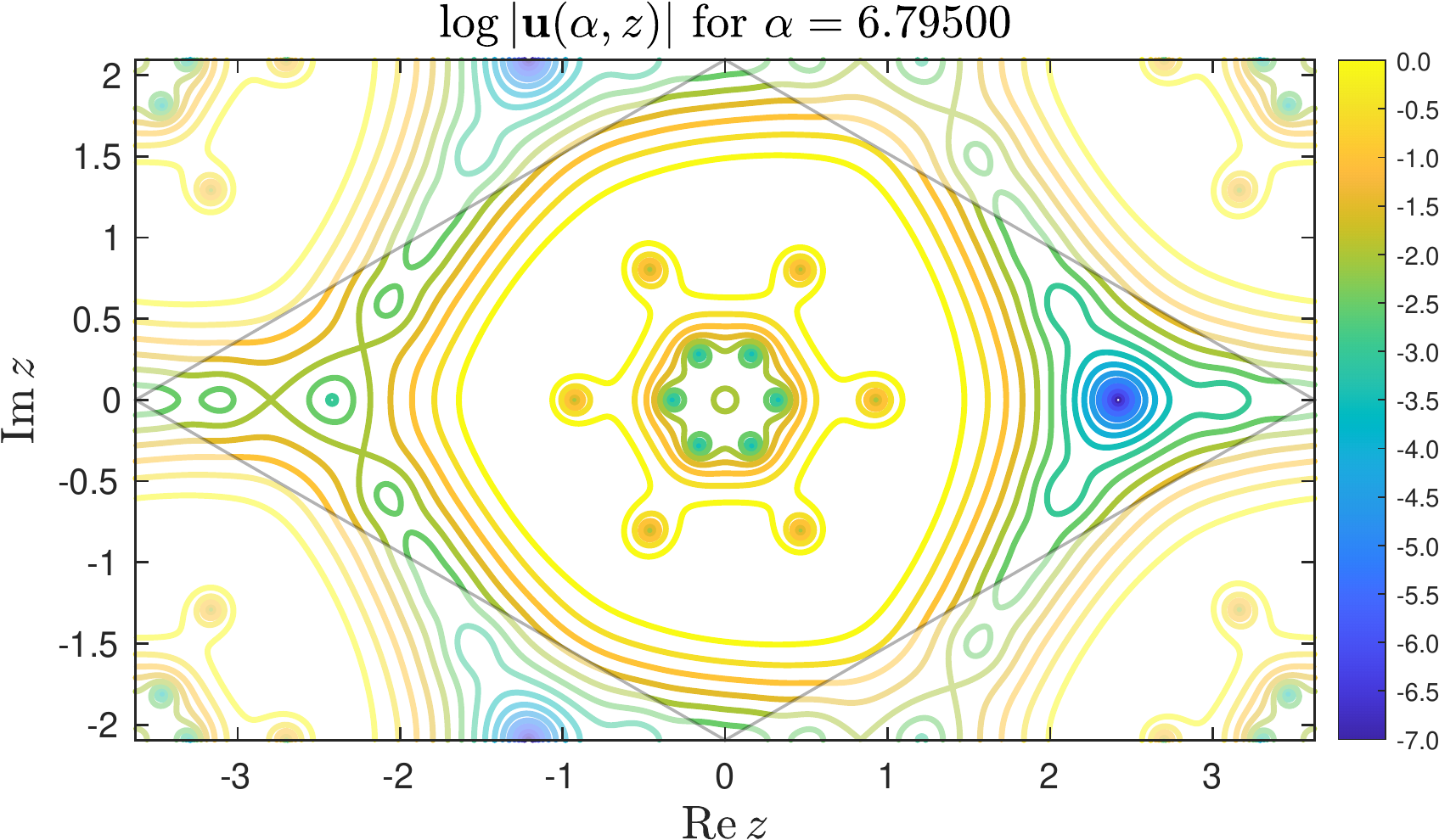}\quad \includegraphics[width=7.5cm]{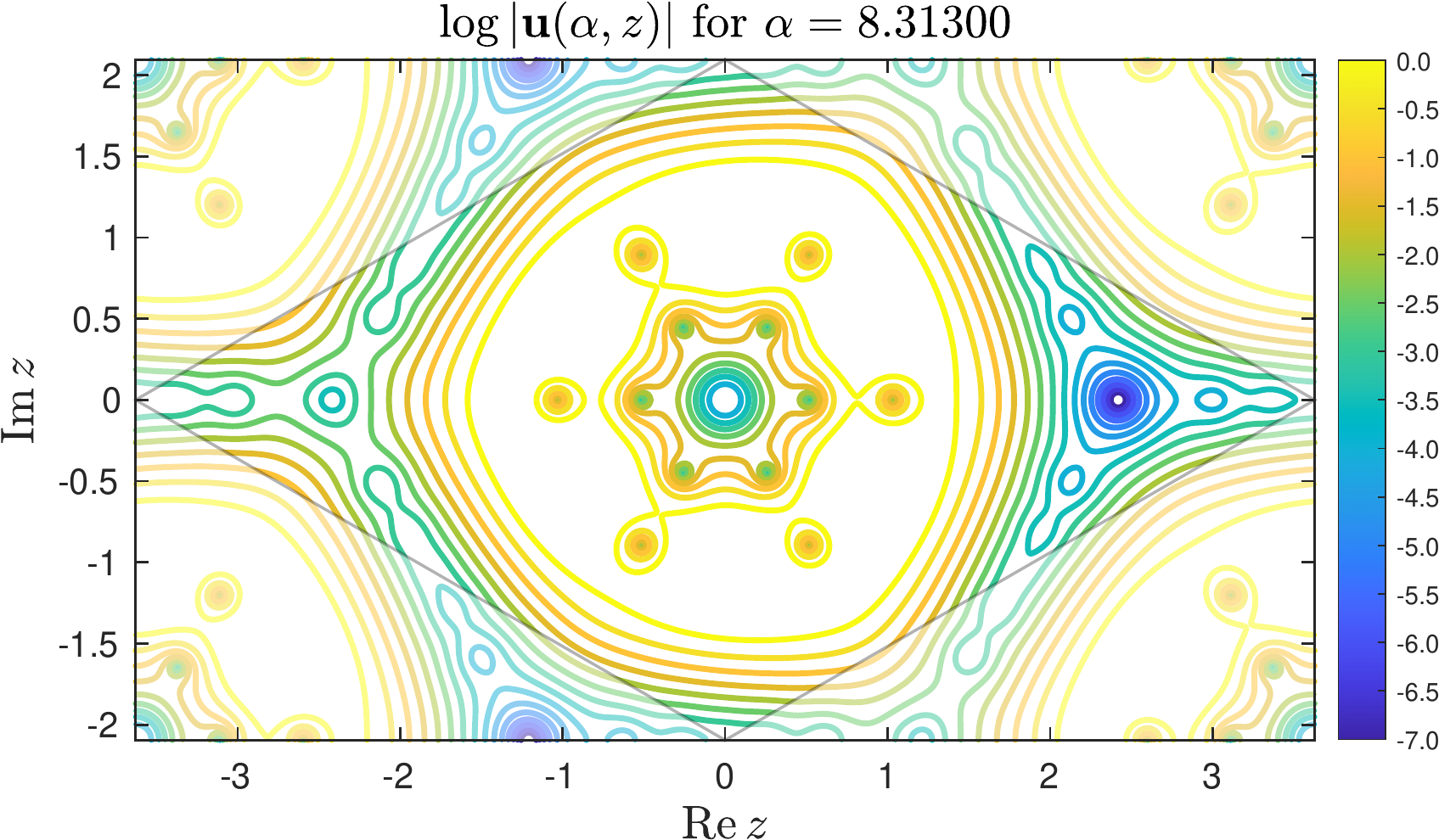}\\[5pt]
\includegraphics[width=7.5cm]{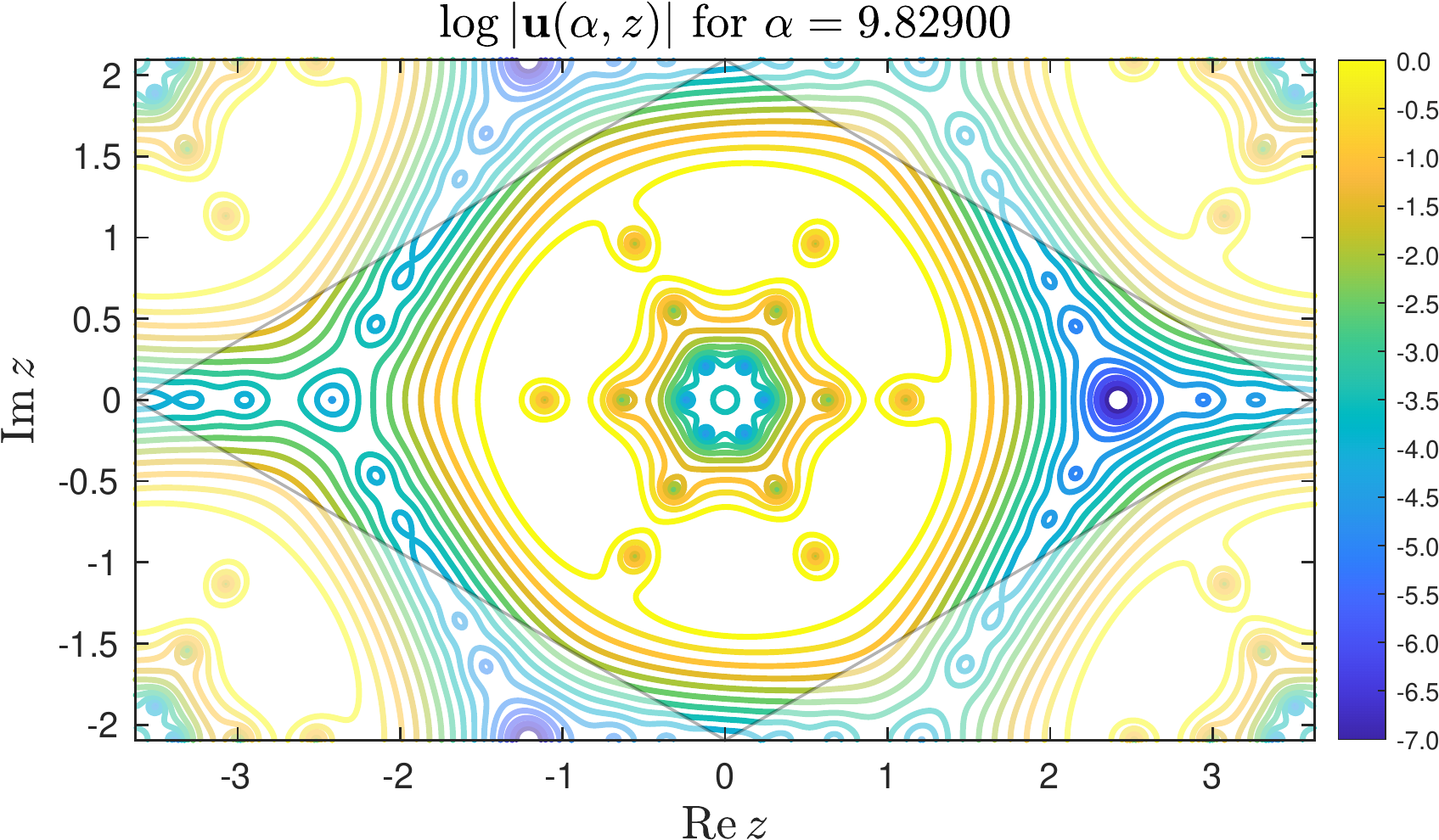}\quad \includegraphics[width=7.5cm]{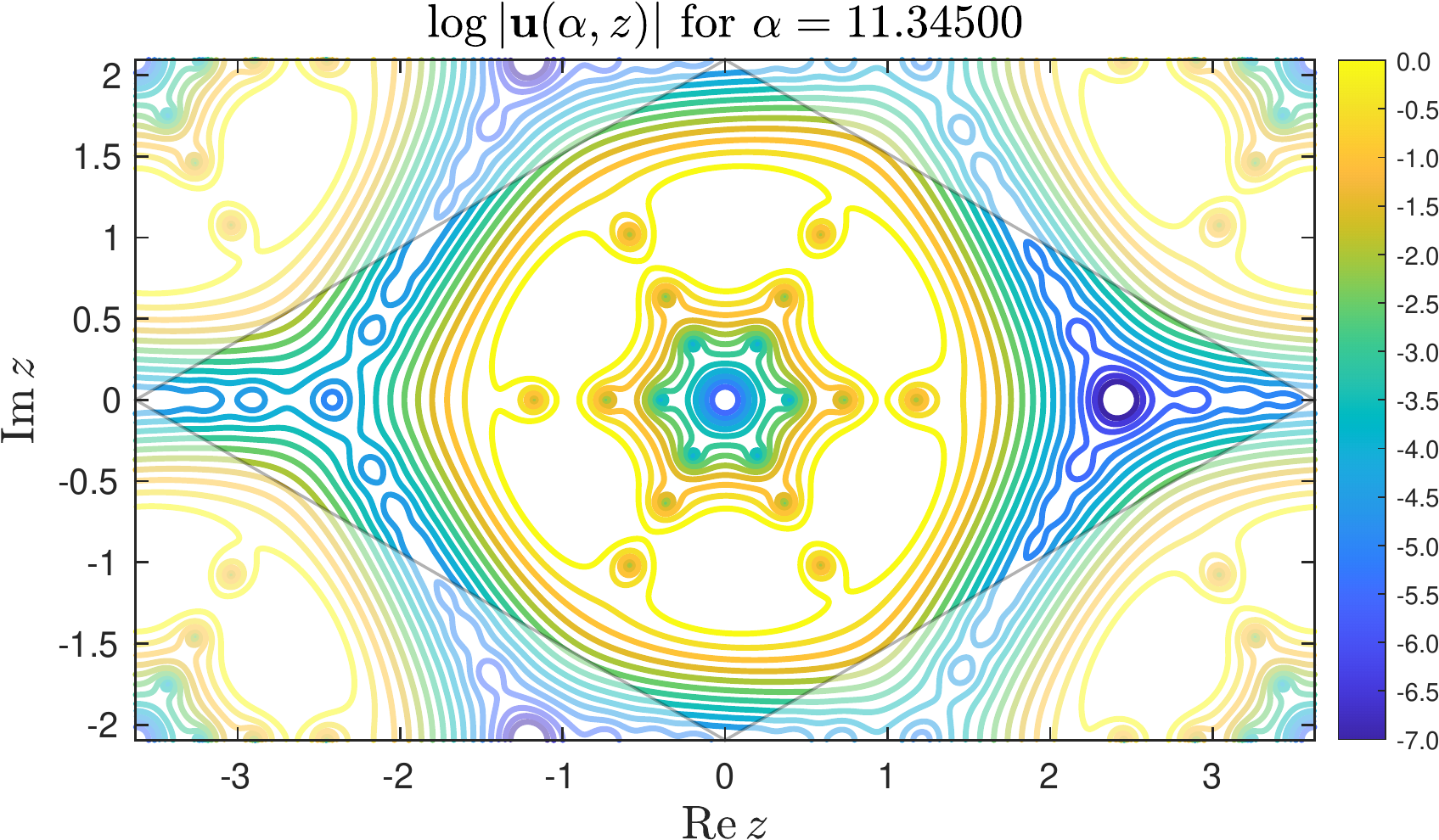}
\end{center}
\caption{\label{f:eig} Plots of $ z \mapsto 
\log| \mathbf u ( \alpha , z ) | $ (in the notation of Proposition \ref{p:Ak}) for $ \alpha $ {\em close} to magic values (due to pseudospectral effects it is difficult to compute the exact eigenfunction 
at a magic angle) showing that the value of $ \mathbf u $ at $ z_S=\frac{4\sqrt{3}}{9} \pi$ is close to $ 0 $.}
\end{figure}

\subsection{A theta function argument}
\label{s:theta}

We first review basic definitions and properties of $ \theta $
functions -- see \cite{tata}. We have
\begin{equation}
\label{eq:theta1}
\begin{gathered} 
\theta_{a,b} ( z| \tau ) := \sum_{ n \in \ZZ } \exp ( 
\pi i ( a  + n)^2 \tau + 2 \pi i ( n + a ) ( z + b ) ) , \ \ \Im \tau >0 , \\
\theta_{a,b} ( z + 1 | \tau ) = e^{ 2 \pi i a} \theta_{a,b} ( z| \tau), \ \ 
\theta_{a,b} ( z + \tau | \tau ) = e^{ -2 \pi i ( z + b ) - \pi i \tau }
\theta_{a,b} ( z| \tau ) , \\
\theta_{ a + 1, b }  ( z| \tau ) = \theta_{a,b} ( z| \tau ) , \ \ \
\theta_{ a, b+1} ( z| \tau ) = e^{ 2 \pi i a } \theta_{a,b} ( z | \tau) .
\end{gathered}
\end{equation}
The (simple) zeros of the (entire) function $ z \mapsto \theta_{a,b} ( z| \tau ) $
are given by 
\begin{equation}
\label{eq:theta2}
z_{ n,m} = ( n - \tfrac12 - a ) \tau + \tfrac12 - b - m .
\end{equation}
If 
\begin{equation}
\label{eq:goz} g ( z):=  \frac{ \theta_{a',b'} ( z/ \tau'| \tau )}{\theta_{a,b} ( z/\tau'| \tau ) }, \end{equation}
then \eqref{eq:theta1} 
shows that
\begin{equation}
\label{eq:gper}   g ( z + \tau' ) = e^{ 2 \pi i ( a' - a ) } g ( z ) , \ \
g ( z + \tau \tau' ) = e^{ - 2 \pi i ( b' - b) } g ( z ) , \end{equation}
and from \eqref{eq:theta2} we know the zeros and poles of $ g $.

With this in place we can prove
\begin{prop}
\label{p:magic}
In the notation of Propositions \ref{p:Ak} and \ref{p:inverse}
we have
\begin{equation*}
\label{eq:magic3}
v( \alpha ) = 0, \ \ \alpha \in \RR \ \Longrightarrow \ \alpha \in \mathcal A . \end{equation*}
\end{prop}
\noindent
{\em Proof.} If $ \mathbf u ( \alpha ) = ( \psi_1 , \psi_2 ) $ then 
\[ v ( \alpha ) = \psi_1 ( z ) \psi_1 ( -z ) + \psi_2 ( z ) \psi_2 ( -z ) .\]
As remarked after \eqref{eq:Wr}, $ v ( \alpha ) $ is independent of 
$z  $. 

The observation made in \cite{magic} is that $ \psi_2 $ vanishes
at special {\em stacking} points. These are fixed points of the action 
$ z \mapsto \omega z $ on $ \CC/\Gamma_3 $ (see \eqref{eq:defGam3}):
\begin{equation}
\label{eq:stack}
\begin{gathered}
\psi_2 ( \alpha, \pm z_S  ) = 0 , \ \ z_S :=  \tfrac13 ( \mathbf a_2 
- \mathbf a_1 ) = \tfrac{ 4 \sqrt {3} } 9 \pi , \ \  \mathbf a_j = \tfrac 43 \pi i \omega^j . 
\end{gathered}
\end{equation}
To see this, note that (with the action of $ \mathscr C $ identified with 
the action on $ ( \mathbf u , 0_{\CC^2}  )^t \in L^2 ( \CC/\Gamma; \CC^4 )$)
\[ \begin{split} \mathbf u ( \alpha , \pm z_S ) 
& = \mathscr C   \mathbf u ( \alpha , \pm z_S ) = 
 \mathbf u  ( \alpha , \pm \omega z_S ) = 
 \mathbf  u ( \alpha , \pm z_S \mp \mathbf a_2 ) \\
 & = 
 \begin{pmatrix}  \omega^{\pm 1 } & 0 \\
 0 & 1 \end{pmatrix} \mathscr L_{\mp \mathbf a_2 } \mathbf u ( \alpha, \pm z_S) =  \begin{pmatrix}  1 & 0 \\
 0 & \omega^{\mp 1 } \end{pmatrix} \mathbf u ( \alpha , \pm z_S) . 
\end{split}  \]
Hence $ \psi_2 ( \pm z_S  ) = \omega^{\mp 1 } \psi_2 ( \pm z_S ) $,
which proves \eqref{eq:stack}.

We conclude that if $ v ( \alpha ) = 0 $ then $ \psi_1 ( z_S ) \psi_1 ( - z_S ) = 0 $, and hence $ \mathbf u ( \alpha, z_S ) = 0 $ or
$ \mathbf u ( \alpha, - z_S ) = 0 $.  Assume the former holds (otherwise 
we replace $ \mathbf u $ with $ \mathscr E \mathbf u $). We can then 
construct a periodic solution to $ ( D ( \alpha ) - \mathbf k ) \mathbf v_{\mathbf k }  
= 0 $ for any $ \mathbf k \in \CC $, and in particular for $ 
\mathbf k \notin \Gamma^* $, implying, in view of \eqref{eq:SpecD}, that
$ \alpha \in \mathcal A $.

In fact, if $ f_{\mathbf k } $ is holomorphic with simple poles 
at the zeros of $ \mathbf u $ allowed (we note that the equations
$ 2 D_{\bar z } \psi_1 + U ( z ) \psi_2 = 2 D_{\bar z } \psi_2 + 
U(-z ) \psi_1 = 0 $ imply that  $ \partial_{
\bar z }^\ell \psi_j ( z_S ) = 0 $ and hence $ \mathbf u = 
( z - z_S ) \widetilde {\mathbf u } $, where $ \widetilde {\mathbf u } $
is smooth near $ z_S $) then
\begin{equation*}
\label{eq:magic1}   ( D ( \alpha ) - \mathbf k ) \mathbf v_{\mathbf k} = 
0 , \ \ \ \mathbf v_{\mathbf k} ( z ) =  e^{  \frac i 2 ( z \bar {\mathbf k }+  \bar z {\mathbf k}) } f_{\mathbf k}  ( z ) \mathbf v (z ) . 
\end{equation*}
To obtain periodicity we need
\begin{equation*}
\label{eq:magic2} 
\begin{gathered}
f_{\mathbf k} ( z + \mathbf a ) = e^{  - \frac i 2 ( \mathbf a \bar {\mathbf k }+  \bar {\mathbf a}  {\mathbf k}) } f_{\mathbf k} ( z ) , \ \ \mathbf a \in \Gamma , \ \ 
\tfrac 1 2 ( \mathbf a \bar {\mathbf k }+  \bar {\mathbf a}  {\mathbf k}) 
= 2 \pi ( a_1 k_1 + a_2 k_2 ) , \\ 
\mathbf a = 4 \pi ( a_1 i\omega  + a_2 i\omega^2  ) , \ \ 
\mathbf k = \tfrac{1}{\sqrt 3 } ( k_1 \omega^2 - k_2 \omega ) .
\end{gathered}
\end{equation*}
But now, \eqref{eq:theta2}--\eqref{eq:gper} show that we can take
\[ \hspace{1.65in} f_{\mathbf k} ( z ) = \frac{ \theta_{ -\frac16 + k_1/3, \frac16 - k_2/3 } ( 3 z/ 4 \pi i 
\omega | \omega )}{ \theta_{ -\frac16, \frac16} ( 3 z/ 4\pi i \omega | \omega ) }  
.
\hspace{1.65in} \Box \] 

\begin{proof}[Proof of Theorem \ref{t:spec}]
The lack of dependence of the spectrum of $ T_{\mathbf k }$ on $ \mathbf k \notin  \Gamma^* $ and equivalence of statements~(1) and (2) are the content of 
Proposition \ref{p:Ak}. The definition of $ H_{\mathbf k } ( \alpha ) $
in \eqref{eq:defmal} immediately shows their equivalence to statement~(3).
\end{proof}

\begin{proof}[Proof of Theorem \ref{t:magic}]
In Proposition \ref{p:Ak} we already obtained a (real) analytic family
$ \alpha \mapsto \mathbf u ( \alpha ) $. Then 
$ v ( \alpha ) = W ( \mathbf u ( \alpha ) , \mathscr E \mathbf u ( \alpha ) ) $ and the equivalence of $ v( \alpha ) $ to (1) in 
Theorem \ref{t:spec} follows from Proposition \ref{p:inverse} and
\ref{p:magic}.
\end{proof}

\noindent
{\bf Remarks.} 1. The zero of $ \mathbf u ( \alpha ) \in \ker_{ L^2_{\rho_{1,0}} ( \CC/\Gamma, \CC^2 )} D ( \alpha ) $
seems to occur at $ z_S $ only -- see Figure~\ref{f:eig}. This is
also suggested by the following argument: from $ v ( \alpha ) = 0 $ we see that
$ \mathscr E \mathbf u ( z) = f ( z ) \mathbf u ( z)  $,
where, using $ v ( \alpha ) = 0 $ again, 
\begin{equation}
\label{eq:deff} 
f ( z ) := \frac{ \psi_2 ( -z ) }{\psi_1 ( z ) }  = 
- \frac{ \psi_1 ( -z ) } { \psi_2 ( z ) } = 
\frac{ \alpha U ( z ) \psi_1 ( - z )}{ 2 D_{\bar z } \psi ( z ) }
, \end{equation}
is holomorphic away from $ \psi_1^{-1} ( 0 ) \cap (D_{\bar z } \psi_1)^{-1} 
( 0 ) $. We also see that $ f $ is meromorphic: in fact, near any point
$ z_0 $, 
$\psi_1 ( z_0 + \zeta ) = F_1 ( \zeta, \bar \zeta )$, 
$ \psi_2 ( - z_0 - \zeta ) = F_2 ( \zeta , \bar \zeta ) $, where
$ F_j : B _{\CC^2 } ( 0, \delta ) \to \CC $ are holomorphic functions (this follows
from real analyticity of $ \psi_j $, which follows in turn from the 
ellipticity of the equation -- see \cite[Theorem 8.6.1]{H1}). 
The definition of $ f $ and the fact that $ \partial_{\bar z } f = 0 $
away from zeros of $ \psi_1 $ shows
 that $ F_2 ( \zeta, \xi ) = f ( z_0 + \zeta ) F_1 ( \zeta , \xi ) $. We can then choose $ \xi_0 $ such that $ F_1 ( \zeta, \xi_0 ) $ is not 
identically zero (if no such $ \xi_0 $ existed, $ \psi_1 \equiv 0$, and hence, from the equation, $ \mathbf u \equiv 0 $). But then 
$ \zeta \mapsto f ( z_0 + \zeta ) = F_2 ( \zeta, \xi_0 )/F_1 ( \zeta, \xi_0 ) $ is meromorphic near $ \zeta = 0 $ and, as $ z_0 $ was
arbitrary, everywhere.
In addition, 
\begin{equation*}
\label{eq:propf} 
f ( z + \mathbf a ) = \omega^{-a_1 - a_2} f ( z ) ,  \ \ \mathbf a \in \Gamma_3, \ \  f ( \omega z ) = f ( z ) , \ \ 
f ( z ) f ( -z ) = -1 . 
\end{equation*}
These symmetries also show that 
$ f ( z_S + \omega \zeta ) = \omega^{-1} f ( z_S + \zeta ) $, 
which means that $ f ( z_S + \zeta ) = \sum_{ k \geq k_0 } \zeta^{-1+ 3k} f_k $ and $ f ( -z_S - \zeta ) = \sum_{ \ell \geq 1- k_0 } \zeta^{
-2 + 3 \ell } g_{ \ell }$, for some $ k_0 \in \ZZ $.
 Hence, if $ f $ has only poles of order 1, 
we have $ \mathbf u ( \alpha, z_S ) = 0 $. We formulate this bold guess
as follows:
\begin{equation}
\label{eq:guess}  \mathbf u ( \alpha ) \in \ker_{ L^2_{\rho_{1,0}} ( \CC/\Gamma, \CC^2 )} D ( \alpha ), \ \ \mathbf u ( \alpha ) \not \equiv 0 \ 
\Longrightarrow \ \mathbf u ( \alpha, z ) \neq 0 , \ \ z \notin z_S + \Gamma_3. \end{equation}
This is related to the following fact, which seems to hold as well:
\begin{equation}
\label{eq:dim}  \dim \ker_{ L^2_{\rho_{1,0}} ( \CC/\Gamma, \CC^2 )} D ( \alpha ) = 1 , \ \ \alpha \in \CC .
\end{equation}
\begin{proof}[Proof of \eqref{eq:guess} $ \Rightarrow $ \eqref{eq:dim}]
Suppose that $ \mathbf u = ( \psi_1, \psi_2 )^t $ and 
$ \mathbf v = ( \varphi_1 , \varphi_2 )^t $ are two elements of 
the kernel in $ L^2_{\rho_{1,0}} $. We then define the (constant) Wronskian $ w := \psi_1 \varphi_2 - \psi_2 \varphi_1 $. Since 
$ \varphi_2 ( \pm z_S ) = \psi_2 ( \pm z_S ) = 0 $ (see \eqref{eq:stack}),
we have $ w = 0 $ and hence $ \mathbf v = g \mathbf u $, where
$ g ( z ) = \varphi_1 ( z ) /\psi_1 ( z ) $. As in the discussion of
$ f$ given after \eqref{eq:deff}, we see that $ g ( z ) $ is a meromorphic
function periodic with respect to $ \Gamma_3 $. From \eqref{eq:guess}
applied to $ \psi_1 $ we see that $ g $ can only have poles at 
$ z_S + \Gamma_3 $, and applied to  $\varphi_1 ( z ) $
we see that $ g $ can only have zeros at the same place. But this implies that $ g $ is constant.
\end{proof}

\noindent
2. The elements of the kernel of $ D ( \alpha ) - \mathbf k$
can be obtained from the (finite rank) residue of the operator
\eqref{eq:inv}, and theta functions are already implicitly present there.
On one hand (see \S \ref{s:num}) the operator $ ( 2 D_{\bar z } - 
\mathbf k )^{-1} $ can be described using Fourier expansion, but on the other hand it can be represented using theta functions: it is the convolution with the fundamental solution of $ 2 D_{\bar z } - \mathbf k $
on $ \CC/\Gamma $. To obtain the convolution kernel (in a construction which works for any torus) we seek a function $ G_{\mathbf k } $ such that
\begin{equation*}
\label{eq:Gk} 
\begin{gathered}   ( 2 D_{\bar z } - \mathbf k ) G_{\mathbf k } = \delta_0 ( z ) , \ \ \
G_{\mathbf k} = e^{ \frac i 2 ( \mathbf k \bar z + \bar{\mathbf k} z)}
g_{\mathbf k } ( z ) , \ \ \
 \partial_{\bar z} g_{\mathbf k }|_{ \CC \setminus \Gamma}  =0 , \\
g_{\mathbf k } ( z + \mathbf a ) = 
e^{ - \frac i 2 (\bar{\mathbf k } 
\mathbf a + \mathbf k \bar{\mathbf a } )} g_{\mathbf k} ( z ) , \ \ \
{\rm{Res}}_{z = w} g_{\mathbf k} ( z ) = \left\{ \begin{array}{ll} i/(2 \pi), & w \in \Gamma;\\
\ \ 0, & w \notin \Gamma . \end{array} \right. \end{gathered}
\end{equation*}
(The last condition gives $ 2 D_{\bar z } g_{\mathbf k} ( z ) = \sum_{ \mathbf a \in \Gamma } \delta_{\mathbf a } ( z ) $, as $
\partial_{\bar z } ( 1/ (\pi z )) = \delta_0 ( z ) $.)

To find $ g_{\mathbf k } $ we return to \eqref{eq:theta2} 
and \eqref{eq:goz} and choose
\begin{equation*}
\label{eq:choicetau1} \tau' = 4 \pi i \omega  , \ \   \tau \tau' = 4 \pi i \omega^2  , \ \ 
 a =  \tfrac12 , \ \
b = \tfrac12 , \ \ a'  =  \tfrac 12 - k_1, \ \ b' =  \tfrac12 + k_2 .
\end{equation*}
Hence we have
\begin{equation}
\label{eq:gkz} 
\begin{gathered} g_{\mathbf k } ( z ) :=  
\frac{  e^{ - \pi i k_1^2 + 2 \pi i k_1 ( \frac12 + k_2 )} \theta_{ \frac12, \frac12 }' (  0 | \omega ) }  { 2 \pi i \theta_{ \frac12 , \frac12  } (  \omega k_1 + k_2 ) | \omega )}
\frac{ \theta_{ \frac12 - k_1, \frac12 + k_2 } ( z/ 4 \pi i 
\omega | \omega )}{ \theta_{ \frac12, \frac12 } (  z/ 4\pi i \omega |\omega ) } , 
\\  \mathbf k = \tfrac{1}{\sqrt 3 } ( k_1 \omega - k_2 \omega^2 ) , 
\ \  ( k_1 , k_2 ) \notin \ZZ^2. 
\end{gathered}
\end{equation}
It would be interesting to derive \eqref{eq:recipe} from \eqref{eq:inv}
and \eqref{eq:gkz}.
\qed

\subsection{Existence of magic $ \alpha$'s} 
\label{s:trace}
We now give a proof of Theorem \ref{t:trace} which amounts to calculating
$ \tr T_{\mathbf k}^4 $. For that it is convenient to switch to 
rectangular coordinates, which are also used in numerical computations
(see \S \ref{s:num}): $z = x_1 + i x_2 =  2i \omega y_1 + 2 i \omega^2 y_2$.
We have 
$ U ( z ) 
 = e^{ -i ( y_1 + y_2 ) } 
+ \omega e^{ i (  2 y_1 - y_2 ) } + \omega^2 e^{ i ( - y_1 + 2 y_2 )  }$
and $ 2 D_{\bar z } =  D_{x_1} + i D_{x_2} =
 \left(  \omega^2  D_{y_1 } - \omega D_{y_2} \right)/\sqrt{3} $.
We are then studying
\begin{equation}
\label{eq:Dkal} \begin{gathered}    D_{\mathbf k } ( \alpha ) := D ( \alpha ) + 
\mathbf k = \tfrac{1}{ \sqrt 3 } \begin{pmatrix} 
\mathscr D_{\mathbf k }  & 
 \alpha \mathscr V ( y )  \\
\alpha \mathscr V ( - y ) & \mathscr D_{\mathbf k} 
\end{pmatrix} , \\ \mathscr D_{\mathbf k } :={ \omega^2  
(D_{y_1} + k_1) - \omega (D_{y_2} + k_2 ) }, \\
\mathscr V ( y ) := \sqrt 3 ( e^{ -i ( y_1 + y_2 ) } 
+ \omega e^{ i (  2 y_1 - y_2 ) } + \omega^2 e^{ i ( - y_1 + 2 y_2 )}  ),
\end{gathered} 
\end{equation}
with {\em periodic} periodic boundary conditions (for $ y \mapsto
y + 2 \pi \mathbf n  $, $ \mathbf n  \in  \ZZ^2 $). {In the following, we shall write $\mathscr V_{\pm}(y):=\mathscr V(\pm y).$}
The operator $ T_{\mathbf k } $, $ \mathbf k = ( \omega^2 k_1 -
 \omega k_2 )/\sqrt 3$, $ ( k_1, k_2 ) \notin \ZZ^2 $, is given by 
\[   T_{\mathbf k} := \begin{pmatrix} 0
  &  \mathscr D_{\mathbf k }^{-1}  \mathscr V_+  \\
 \mathscr D_{\mathbf k }^{-1}  \mathscr V_-  & 0
\end{pmatrix}. \]
In this notation,
\begin{equation}
\label{eq:defA}   \tr T_{\mathbf k}^4 = 18 \tr A^2, \ \ 
 A := A_{\mathbf k } :=  \tfrac13 \mathscr D_{\mathbf k }^{-1}  \mathscr V_+
\mathscr D_{\mathbf k }^{-1}  \mathscr V_- , \end{equation}
where we note that $ A^2 $, a pseudodifferential operator of
order $ -4 $, is of trace class (see for instance \cite[Theorem B.21]{res}).

By taking the (discrete) Fourier transform on 
$ \RR^2 / 2 \pi \ZZ^2 $ we consider the operator $ D_{\mathbf k} ( \alpha ) $ as acting on $  \ell^2 ( \ZZ ) \otimes \ell^2 ( \ZZ ) $. 
With  $ D :=  {\rm{diag}}\, ( \ell )_{\ell \in \ZZ}$
and $ J ( (a_n)_{ n \in \ZZ}  ) ) = ( a_{n+1} )_{n \in \ZZ } $, 
we have 
\begin{equation}
\label{eq:DJ} 
\begin{split}
 \mathscr D_{\mathbf k } &= \omega^2 ( D + k_1  )
\otimes I -
\omega  I \otimes ( D + k_2 I ) ,\\
\mathscr  V_+ / \sqrt 3 &= J \otimes J + \omega \, J^{-2}  \otimes 
J + \omega^2 J \otimes J^{-2}, \\ 
\mathscr  V_- / \sqrt 3 &= J^{-1} \otimes J^{-1} + \omega J^2 \otimes 
J^{-1}  + \omega^2 \, J^{-1} \otimes  J^{2} .
 \end{split}
 \end{equation}

The numerical value in Theorem \ref{t:trace} will come from the following,
 surely classical, computation:
\begin{lemm}
\label{l:sum}
For $ \Gamma := \omega \ZZ \oplus \ZZ $, $ \omega := e^{ 2 \pi i/3} 
$ and $ \gamma_0 \in \Gamma \setminus \{ 0 \} $ define
\begin{equation} K ( \gamma_0 ) := \sum_{\gamma \in \Gamma \setminus \{ 0,  \gamma_0 \} } 
\gamma^{-2} ( \gamma - \gamma_0 )^{-2} . \label{eq:defK} \end{equation}
Then 
\begin{equation}
\label{eq:sum}
K ( \omega m + n ) = - \frac{ 4 \pi i ( \omega ( 2n - m ) + n + m ) }
{ 3 ( \omega m + n)^3 } . 
\end{equation}
\end{lemm} 
\begin{proof}
We notice that $ K ( \omega \gamma_0 ) = \bar \omega K ( \gamma_0 ) $.
Hence it is enough to evaluate
\begin{equation}
\label{eq:defg} g ( \gamma_0 ) := \sum_{ j=0 }^2 \omega^j K( \omega^j \gamma_0 ) = 3 K( \gamma_0 ) .
\end{equation}
Also, if we define $ F  ( z , \gamma_0 ) := \sum_{\gamma \in \Gamma }
( \gamma - z )^{-2} ( \gamma - \gamma_0 - z )^{-2} $, 
then $ F $ is a meromorphic $ \Gamma$-periodic function with 
the singularity at $z= 0 $ given by $ 2 / (z  \gamma_0)^2 $. Hence,
\[ F ( z, \gamma_0 ) = 
2 \gamma_0^{-2} \wp ( z ) + K( \gamma_0 ), \ \ \ 
 \wp ( z ) := \sum_{ \gamma \in \Gamma } 
\left( \frac 1 {(\gamma - z )^{2}} 
- \frac {1 - \delta_{  \gamma , 0}} { \gamma^2 } \right)  . \]
Using the partial fraction expansion, the fact that 
$ \sum_{j=0}^2 \omega^j = 0 $ and the above series for the 
$ \wp $-function, we obtain
\[\begin{split}  g( \gamma_0 ) & =  
\sum_{ j =0}^2 \omega^j F ( z , \omega^j\gamma_0 )  =  \gamma_0^{-2} 
\sum_{\gamma \in \Gamma } \sum_{j=0}^2 \bar \omega^{j} 
\left( \frac 1 {(\gamma - \omega^j \gamma_0 - z )^{2}} - \frac{ 2 \bar\omega^{j} } {\gamma_0} \frac 1 { ( \gamma - \omega^j \gamma_0 - z ) }
\right) \\
& =  \gamma_0^{-2} \sum_{j=0}^2 \bar \omega^j \wp ( z ) + 
\sum_{\gamma \in \Gamma } \sum_{j=0}^2 \frac{\bar \omega^j }{ \gamma_0^2 } 
 \frac {1 - \delta_{  \gamma , \omega^j \gamma_0 }} {( \gamma - \omega^j \gamma_0 )^2 } 
- \sum_{\gamma \in \Gamma } \sum_{j=0}^2 \frac{2 \omega^j }{ \gamma_0^3 }   \frac 1 { ( \gamma - \omega^j\gamma_0 - z ) },
\end{split} \]
where the first term on the right hand side vanishes and both series converge absolutely (this can be checked by 
taking a common denominator using $ \prod_{ j=0}^2 ( \zeta - \omega^j \gamma_0 ) = \zeta^3 - \gamma_0^2  $). 
We now have 
\[ \begin{split} \sum_{\gamma \in \Gamma } \sum_{j=0}^2 {\bar \omega^j }
 \frac {1 - \delta_{  \gamma , \omega^j \gamma_0 }} {( \gamma - \omega^j \gamma_0 )^2 } & =  \lim_{ N \to \infty } \sum_{ j=0}^2 
\bar \omega^j \sum_{ |\gamma - \omega^j \gamma_0 | \leq N } 
( 1 - \delta_{ \gamma , 0} ) { \gamma^{-2} } \\
& = \mathcal O ( 1 ) \lim_{ N \to \infty } \sum_{ N - | \gamma_0|
\leq | \gamma | \leq N + |\gamma_0 | }  N^{-2} = 0 .
\end{split}  
   \]
Hence, using the fact that $ \sum_{ n \in \ZZ } ( (n - a)^{-1} - 
( n - b)^{-1} ) = \pi \cot \pi b - \pi \cot \pi a $,
\begin{equation}
\label{eq:sumg} \begin{split}  g ( \gamma_0 ) 
& = 
2 \pi \gamma_0^{-3} \lim_{ M \to \infty } \sum_{m=-M}^M 
\sum_{ j=0}^2 \omega^j \left(    \cot \pi ( m \omega + \omega^j \gamma_0 + z ) - \cot \pi ( m \omega + z ) \right) .
\end{split} \end{equation}
Since 
$ \cot \pi x - \cot \pi y = 2 i ( ( e^{ 2\pi i x  } - 1 )^{-1}  - ( e^{ 2 \pi i y} - 1 )^{-1}) $,
$   e^{ 2 \pi i n \omega } = (-1)^n e^{ - n \pi \sqrt 3 } $, $ n \in \ZZ $, we obtain, with 
$a_m := ( e^{ 2 \pi i (m \omega  + z) } - 1 )^{-1}$, 
\[  \begin{split} 
&  \sum_{m=-M}^M (  \cot ( (m+m_0) \omega + n_0 - z )  - \cot \pi ( m \omega  - z )  )  \\\ & \ \ \ \ \ \ = 
2 i \sum_{ m=-M}^M (a_{ m + m_0 } - a_{m } )   
=
2i \sum_{m=M+1}^{M+m_0} a_m -  2i \sum_{m=-M}^{-M+m_0+1} a_m  
  \\
& \ \ \ \ \ \  
=  2i \sum_{m=M+1}^{M+ m_0} (  -1 + \mathcal O (e^{-M} ))  
-  2i \sum_{m=M-m_0+1}^{M} \mathcal O (e^{-M} )
=  - 2i  m_0 +
\mathcal O ( e^{-M} ) . \end{split}
 \]
Inserting this in \eqref{eq:sumg} with  $ \gamma = \omega m_0 + n_0 $ 
(and calculating the corresponding $ \omega^j \gamma $)
gives
\[ \begin{split} g ( \omega m_0 + n_0 ) & = 
- 4 \pi i (\omega m_0 + n_0 )^{-3}  ( \omega ( 2 n_0 -  m_0 ) + n_0 + m_0 ), 
\end{split} \]
which, in view of \eqref{eq:defg}, proves \eqref{eq:sum}.
\end{proof}

We can now give the 
\begin{proof}[Proof of Theorem \ref{t:trace}]
To simplify calculations we introduce the following notation:\begin{equation}
\label{eq:Jor}
J^{p,q} := J^p \otimes J^q, \ \ p,q \in \ZZ . 
\end{equation}
Also, for a diagonal matrix $ \Lambda = ( \Lambda_{ij} )_{ i,j \in \ZZ } $
acting on $ \ell^2 ( \ZZ ) \otimes \ell^2 ( \ZZ ) $ we define
a new diagonal matrix with the following basic properties:
\begin{equation}
\label{eq:Lam}
\begin{gathered}   \Lambda_{p,q} := ( \Lambda_{ i+p , j+q } )_{ i, j \in \ZZ },  \\
  ( \Lambda \Lambda' )_{ p, q } = 
 \Lambda_{p, q}  \Lambda'_{p, q} , \ \ \ \ 
( \Lambda_{p',q' }  )_{ p,q } =  \Lambda_{ p+p', q+q' }, \end{gathered}
\end{equation}
where $\Lambda'$ is just another diagonal matrix. 
To express powers of $ A $ in \eqref{eq:defA} we will use the following
simple fact:
\begin{equation}
\label{eq:l1} J^{p,q} \Lambda J^{p',q'} 
= \Lambda_{ p, q }  J^{ p+p', q+q' } = J^{ p+p', q+q'} \Lambda_{ -p',-q'}.
\end{equation}
If we put 
\[ \Lambda := D_{\mathbf k }^{-1}, \ \ 
\Lambda_{mn} = (\omega^2 ( m + k_1 ) - \omega ( n + k_2 ) )^{-1}, \ \ 
( k_1, k_2 ) \notin \ZZ^2 , \]
then, in the notation of \eqref{eq:defA},
\[ \begin{split} A & = 
\Lambda ( J^{1,1} + \omega J^{-2,1} + \omega^2 J^{1,-2} ) \Lambda 
( J^{-1,-1} + \omega J^{2,-1} + \omega^2 J^{-1,2} ) 
\\ & =  \Lambda \Lambda_{1,1}  + \omega \Lambda \Lambda_{1,-2} + 
\omega^2 \Lambda \Lambda_{ -2,1}  + \omega \Lambda \Lambda_{1,1}
  J^{3,0} + \omega^2 \Lambda \Lambda_{1,1}  J^{0,3} \\
& \ \ \  + 
\omega \Lambda \Lambda_{-2,1} J^{-3,0} 
 +  \omega^2 \Lambda \Lambda_{1,-2} J^{0,-3} + \Lambda \Lambda_{-2,1}  J^{-3,3} + 
\Lambda \Lambda_{1,-2} J^{3,-3} .
\end{split} \]
The diagonal part of $ A^2 $ is then given by (note
that the matrices are diagonal and commute)
\begin{equation}
\label{eq:B} 
\begin{split} B  & :=  \Lambda^2  \Lambda_{1,1}^2 + \omega^2 \Lambda^2 \Lambda_{1,-2}^2  + 
\omega \Lambda^2 \Lambda_{-2, 1} ^2 
+ 2 \omega \Lambda^2 \Lambda_{1,1} \Lambda_{1,-2} + 
2 \omega^2 \Lambda^2 \Lambda_{1,1} \Lambda_{-2,1} \\
& \ \ \ + 2 \Lambda^2 \Lambda_{-2,1} \Lambda_{1,-2} + \omega^2 \Lambda_{1,1}^2 \Lambda \Lambda_{3,0}
+ \omega^2 \Lambda_{-2,1}^2 \Lambda \Lambda_{-3,0} + 
\omega \Lambda_{1,1}^2 \Lambda \Lambda_{0,3} \\
& \ \ \ 
+ 
\omega \Lambda_{1,-2}^2 \Lambda  \Lambda_{0,-3} + 
\Lambda_{-2,1}^2 \Lambda \Lambda_{-3,3} + \Lambda_{1,-2}^2 \Lambda \Lambda_{3,-3}.
 \end{split}
 \end{equation}
Since 
$ \tr \Lambda^2_{k,\ell} \Lambda_{p,q} \Lambda_{p',q'} = 
\tr \Lambda^2_{k+r,\ell+s} \Lambda_{p 
{+} r, q {+}s} \Lambda_{p'{+} r, q'{+} s} $, 
we have 
\[ \begin{split} \tr A^2 & = \tr  \Lambda^2 ( 
\Lambda_{1,1}^2 + 2 \Lambda_{-2,1} \Lambda_{1,-2} + 
2\Lambda_{2,-1} \Lambda_{-1,2} ) \\
& \ \ \ +\omega  \tr  \Lambda^2 ( 
 \Lambda_{-2,1}^2 + 2 \Lambda_{1,1} \Lambda_{1,-2} + 
2 \Lambda_{-1,2} \Lambda_{-1,-1} ) \\
& \ \ \ +
\omega^2 \tr  \Lambda^2 (  \Lambda_{1,-2}^2 + 2 \Lambda_{1,1} \Lambda_{-2,1} + 2\Lambda_{-1,-1} \Lambda_{2,-1} ) .
\end{split}  \]
We now find that 
\[  \Lambda_{\pm2,\mp1} \Lambda_{\mp1,\pm2} + \omega  \Lambda_{\mp1,\pm2} \Lambda_{\mp1,\mp1}+\omega^2 \Lambda_{\mp1,\mp1} \Lambda_{\pm2,\pm1}=0. \]
In fact, using 
\[ \frac{1}{ab}+\frac{1}{bc}+\frac{1}{ca} = \frac{a+b+c}{abc}\]
it suffices to show, say for the $ + $ case, that for all $n \in \ZZ^2$
\[  \omega (\Lambda_{2,-1})_{n,n}^{-1}+ \bar \omega (\Lambda_{-1,2})_{n,n}^{-1}+ (\Lambda_{-1,-1})_{n,n}^{-1}=0\]
which follows from a direct computation.
Hence, the expression for the trace simplifies further to 
\begin{equation}
\label{eq:simtr}  \begin{split} \tr A^2  = \tr  \Lambda^2 ( 
\Lambda_{1,1}^2 +
\omega  \Lambda_{-2,1}^2 +\omega^2 \Lambda_{1,-2}^2   ), 
\end{split}  \end{equation}
and this expression can be calculated using Lemma \ref{l:sum}. 
The singular terms of the sum in \eqref{eq:simtr} cancel, as the proof of Lemma \ref{l:sum} shows, so we can remove them, and put $ k_1 = k_2 = 0$. 
Noting that $ \omega^2 m - \omega n = \omega \gamma $, $ \gamma = 
\omega m - n $, and $ \omega^2 ( m + p ) - \omega ( n + q ) = 
\omega ( \gamma - \gamma_0 ) $, $ \gamma_0 = -\omega p + q $,
\[
\begin{split}  \tr A^2 & = \bar \omega K ( - \omega + 1 ) + K ( 2\omega + 1) 
+ \omega K ( -\omega - 2 ) \\
& = 
K ( 2 \omega + 1 ) + \omega K ( \omega  ( 2 \omega +1 ) )
+\omega^2 K ( \omega^2 ( 2 \omega + 1 )) = 
3 K ( 2 \omega + 1 )\\
&  ={4 \pi}/ \sqrt 3, \end{split} \]
where we used \eqref{eq:sum} and \eqref{eq:defg}.
In view of \eqref{eq:defA}, this concludes the proof.
\end{proof}

\noindent
{\bf Remark.} Similar arguments can be used to show that
$ \sum_{\alpha \in \mathcal A} \alpha^{-8} = \tr T_{\mathbf k}^4 = 
740 \pi / \sqrt 3 $.

\section{Exponential squeezing of bands}
\label{s:exp}

Here we prove a more general version of Theorem \ref{t:squeeze} valid for   potentials with symmetries \eqref{eq:symmU}. Theorem \ref{t:squeeze} is then obtained as a special case by choosing the potential as in \eqref{eq:defU}. As mentioned in the introduction, in order to see exponential squeezing of bands as $\alpha\to\infty$ for general potentials, it is necessary to impose an additional non-degeneracy assumption. 

To introduce our class of potentials, let
\begin{equation}\label{eq:fn}
f_n(z)=f_n(z,\bar z):=\sum_{k=0}^2 \omega^ke^{\frac{n}{2}(z\bar \omega^k-\bar z\omega^k)}, \ \ \ n \in \ZZ. 
\end{equation}
Then $f_n(\omega z)=\omega f_n(z)$ and
$$
f_n(z+\mathbf a)=\bar\omega^nf_n(z),\quad \mathbf a = \tfrac{4}3 \pi i \omega^\ell,  \ \ell = 1, 2 .
$$
Hence, $f_n$ satisfies \eqref{eq:symmU} only when $n\equiv 1$ mod 3. We shall therefore consider potentials given by \begin{equation}\label{genpot}
U(z)=U(z,\bar z)=\sum_{n\in 3\ZZ+1} a_{n}f_{n}(z,\bar z), \ \ \ \ 
| a_n | \leq c_0 e^{-c_1 |n| } ,
\end{equation}
for some constants $ c_0, c_1 > 0 $. The condition on $ a_n $ is equivalent 
to real analyticity of $ U $.

Special cases of this type of potential have appeared in \cite{guinea2019continuum} and \cite{walet2019emergence}, where the strength of the potential at certain points based on orbital positions and shapes is taken into account to obtain a model  different from \eqref{eq:defU} that still satisfies the desired symmetries.
Note that the potential in \eqref{eq:defU} is obtained from \eqref{genpot} by taking $a_1=1$ and $a_{n}=0$ for all $n\ne1$. The potential $U_\mu$ appearing in Figure \ref{f:mu} is obtained by taking $a_1=1$, $a_{-2}=\mu$ and $a_{n}=0$ for $n\ne 1,-2$.

Since $\overline{f_n(\bar z)}=f_n(z)$ for all $n$, the symmetry relation $\overline{U(\bar z)}=U(z)$ (used in Proposition \ref{p:sym} to achieve
$ \mathcal A = \overline {\mathcal A }$)  
is equivalent to $ \Im a_n = 0 $ for all $ n $.

We now impose a generic non-degeneracy  assumption that
\begin{equation}\label{eq:conditionbracket}
\sum_{n\in3\ZZ+1}n\Re(a_{n})\ne0.
\end{equation}
This is trivially satisfied by the standard potential in \eqref{eq:defU}, and for the potential $U_\mu$ appearing in Figure \ref{f:mu} it holds as long as $\mu\ne\frac12$. For such potentials we have the following strengthened version of Theorem \ref{t:squeeze}.

\begin{theo}
\label{t:squeezegenpot}
Suppose that $ H_{\mathbf k} ( \alpha ) $ is given by 
\eqref{eq:defH} and \eqref{eq:defmal} with $ U $ given by 
\eqref{genpot} and that
\[  \Spec_{ L^2 ( \CC/\Gamma ) } H_{\mathbf k} ( \alpha ) = \{  E_j ( \mathbf k , \alpha ) \}_{ j \in \ZZ }  , \ \ \  
  E_{ j } ( \mathbf k , \alpha ) \leq E_{ j+1 } ( \mathbf k , \alpha ) ,  \ \  \mathbf k \in \CC , \ \ \alpha > 0 , \]
with the convention that $  E_0 ( \mathbf k , \alpha ) = 
\min_{ j} |E_j ( \mathbf k , \alpha ) | $.
If $U$ satisfies \eqref{eq:conditionbracket}, then there exist positive constants $c_0$, $c_1$, and $c_2$ such that
for all $ \mathbf k \in \CC $, 
\begin{equation*}
 | E_j ( \mathbf k, \alpha) | \leq c_0 e^{ - c_1 \alpha } , 
\ \  | j | \leq c_2 \alpha, \ \  \alpha > 0 .
\end{equation*}
\end{theo}

\noindent
{\bf Remark.} If in \eqref{genpot} we assumed instead that 
$ |a_n| \leq C_N |n|^{-N} $ for all $ N $, that is, that the potential is {\em smooth},
then the conclusion would be replaced by $ | E_j ( \mathbf k, \alpha ) |
\leq C_N \alpha^{-N} $ for any $ N$. That follows essentially from 
H\"ormander's original argument -- see \cite[Theorem 2]{dsz} and references
given there.

To prove Theorem \ref{t:squeezegenpot} it is natural to consider $ h = 1/\alpha $
as a semiclassical parameter. This means that
\[ H_{\mathbf k }  ( \alpha ) = h^{-1} \begin{pmatrix} 
0 & P ( h )^* - h \bar {\mathbf k } \\
P ( h ) - h  {\mathbf k} & 0 \end{pmatrix} ,  \ \ \   P = P ( h ) = 
\begin{pmatrix} 2 h D_{\bar z } & U ( z ) \\ U ( - z) & 2 h D_{\bar z } 
\end{pmatrix} , \]
where $ U ( z) $ is a potential given by \eqref{genpot} that satisfies \eqref{eq:conditionbracket}.

The semiclassical principal symbol of $ P (h) - h \mathbf k $ (see \cite[Proposition E.14]{res}) is given by 
\begin{equation}
\label{eq:littlep}
p(z,\bar z,\bar \zeta) = \begin{pmatrix}2 \bar{\zeta}  &   U(z,\bar z) \\  U(-z,-\bar z) &2 \bar{\zeta} \end{pmatrix} , 
\end{equation}
where we use the complex notation $ \zeta = \frac 12 ( \xi_1 - i \xi_2 ) $, 
$ z = x_1 + i x_2 $. The Poisson bracket can then be expressed as 
\begin{equation}
\label{eq:defPo} \{ a , b \} = \sum_{ j=1}^2 \partial_{\xi_j } a \partial_{x_j} b -
\partial_{\xi_j } b \partial_{x_j} a = 
\partial_\zeta a \partial_z b  - \partial_\zeta b \partial_z a + \partial_{\bar \zeta } a \partial_{ 
\bar z } b - \partial_{\bar \zeta } b \partial_{ 
\bar z } a . \end{equation}

The key fact we will use is the analytic version \cite[Theorem 1.2]{dsz} of H\"ormander's construction based on the bracket condition:
suppose that $ Q = \sum_{ |\alpha| \leq m } 
a_\alpha ( x, h ) ( h D)^\alpha $ is a differential operator such that 
$ x \mapsto a_\alpha ( x, h ) $ are real analytic near $ x_0 $, and let
$ q ( x, \xi ) $ be the semiclassical principal symbol of $ Q $. If there exists
\begin{equation}
\label{eq:Pa}  q ( x_0 , \xi_0 ) = 0 , \ \  \{ q , \bar q \} ( x_0 , \xi_0 ) \neq 0 , 
\end{equation}
then there exists a family $ v_h \in C^\infty_{\rm{c}} ( \Omega ) $, $ \Omega $ a
neighbourhood of $ x_0 $, such that
\begin{equation}
\label{eq:quas}
| (h \partial)^\alpha_x Q  v_h (x ) | \leq C_\alpha e^{ - c / h } , \ \
\| v_h \|_{L^2}  = 1, \ \ | (h\partial_x)^\alpha v_h ( x ) | \leq C_\alpha e^{ - c | x- x_0|^2/ h } , 
\end{equation}
for some $ c > 0$. The formulation is different than in 
the statement of \cite[Theorem 1.2]{dsz}, but \eqref{eq:quas} follows
from the construction in \cite[\S 3]{dsz} -- see also 
\cite[\S 2.8]{HiS}.

\begin{figure}
\begin{center}
\includegraphics[width=7cm]{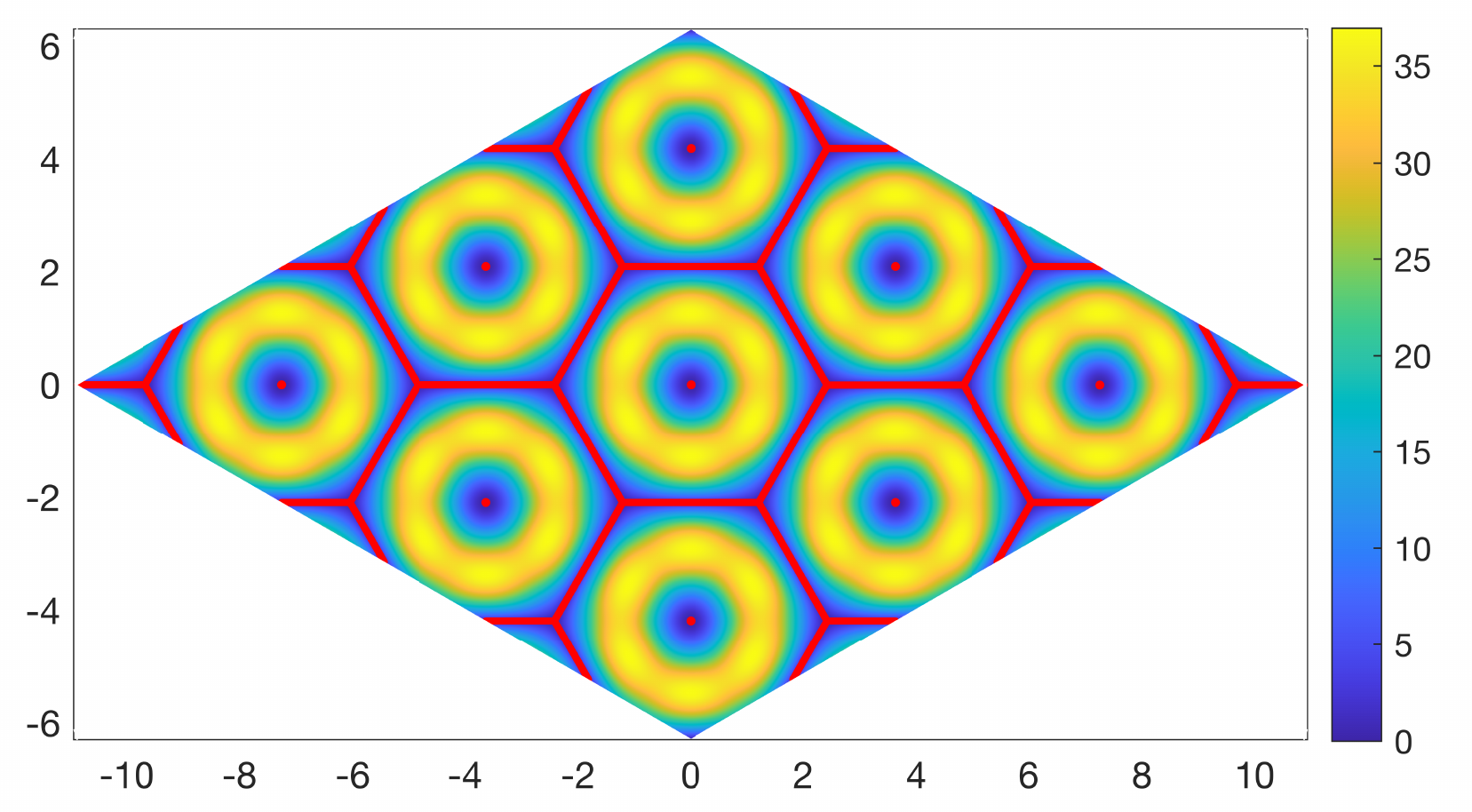} \quad 
\includegraphics[width=7cm]{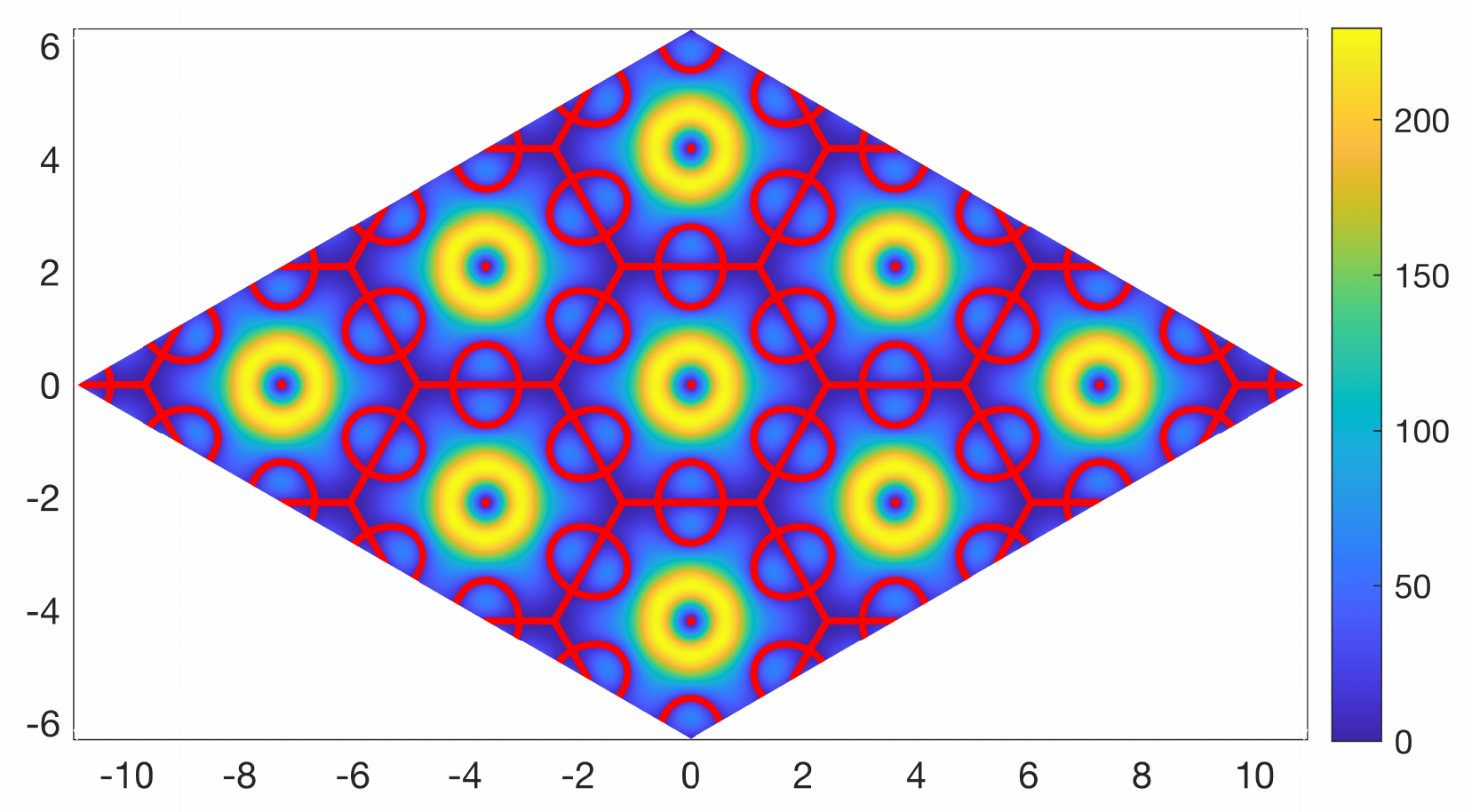} 
\end{center}
\caption{\label{f:Horm} A contour plot of $  | \{ q , \bar q \} | = |\Im ( {\overline{ V}}^{\frac12}   \partial_z V ) | $  -- see \eqref{eq:Horm}. Here $V(z)=U(z)U(-z)$ with $U$ given by \eqref{genpot} so that $U=\sum_n a_nf_n$, $n\equiv 1$ mod 3. In the left panel, $U=a_1f_1$ with $a_1=1$ so that $U$ coincides with \eqref{eq:defU}. In the right panel, $U=a_1f_1+a_{-2}f_{-2}+a_4f_4$ with $a_1=1$, $a_{-2}=-0.75$ and $a_4=0.15$. 
The bracket $ i \{ q, \bar q \} $ is non-zero except on a one-dimensional graph and on a set of points given by the {red} set, and can take 
any sign by choosing a branch of the square root $ V^{\frac12} $. The punctured domain around the origin where $| \{ q , \bar q \} |\ne0$ is clearly visible.}
\end{figure}

We will use this result to obtain
\begin{prop}
\label{p:dsz}
There exists an open set $ \Omega \subset \CC $ and a constant $ c $
such that for any $ \mathbf k \in \CC $ and $ z_0 \in \Omega $ there
exists a family  $ h \mapsto \mathbf u_h \in C^\infty ( \CC/\Gamma ; \CC^2 ) $ such that for
$ 0 < h < h_0 $, 
\begin{equation}
\label{eq:Phz}  | ( P ( h ) - h \mathbf k ) \mathbf u_h ( z ) | \leq e^{ - c /h }, 
\ \ 
\| \mathbf u_h\|_{L^2}  = 1, \ \  | \mathbf u_h ( z ) | \leq e^{ - c | z- z_0|^2/h } . \end{equation}
\end{prop}
\begin{proof}
To apply \eqref{eq:quas} we reduce to the case of a scalar equation,
and for that we look at points where
$ U ( z_0, \bar z_0 ) \neq 0 $. In that case, 
existence of $ \mathbf u_h $ follows from the existence of 
$ v_h \in C^\infty_{\rm{c}} ( \Omega'; \CC )  $, $ \Omega' $ a small 
neighbourhood of $ z_0 $ on which $ U ( z, \bar z ) \neq 0 $, such that
\begin{equation*}
\label{eq:defvh} 
\begin{gathered} 
Q v_h = \mathcal O ( e^{-c/h} ) , \ \ v_h (z_0 )  = 1, \ \ 
|v_h ( z ) | \leq e^{ - c | z-z_0|^2/h } , \\
 Q  :=  U ( z, \bar z ) ( 2 h D_{\bar z} - h \mathbf k ) \left( U ( z, \bar z )^{-1} ( 2 h D_{\bar z} - h \mathbf k )  \right)  - 
U ( -z, - \bar z ) U ( z , \bar z ) , 
\end{gathered}
\end{equation*}
with estimates for derivatives as in \eqref{eq:quas}. We then put
\[ \mathbf u_h := ( v_h , -U ( z , \bar z )^{-1} ( 2 h D_{\bar z } - 
h \mathbf k  ) v_h )  \]
and normalize to have $ \| \mathbf u_h \|_{L^2}= 1$.
Since such $ v_h $ are supported in small neighbourhoods, this defines an 
element of $ C^\infty ( \CC/ \Gamma , \CC^2 ) $.
The principal symbol of $2hD_{\bar z}-h\mathbf k$ is $2\bar\zeta$, and 
basic algebraic properties of the principal symbol map (see \cite[Proposition E.17]{res}) imply that the semiclassical 
principal symbol of $ Q $ is given by
\begin{equation*}
\label{eq:defV}
q(z,\bar z, \bar \zeta) :=\operatorname{det}(p(z,\bar z,\bar \zeta) )=4 \bar \zeta^2 - V ( z, \bar z ), \ \ \
V ( z, \bar z)  := U ( z, \bar z ) U ( - z , - \bar z ). 
\end{equation*}
To use \eqref{eq:quas} we need to check 
 H\"ormander's bracket condition \eqref{eq:Pa}: for $ z $ in an open
 neighbourhood of $ z_0 $, $ U ( z_0, \bar z_0 ) \neq 0 $, there exists
 $ \zeta $ such that
\[  q ( z, \bar z , \bar \zeta ) = 0, \ \ \{ q , \bar q \} ( z, \zeta ) 
\neq 0 .\]
Since $ q =4 \bar \zeta^2 - V ( z, \bar z ) $,  we can take
$ \zeta =  \frac12 \overline V^{\frac 12 } $ (for either branch of the 
square root) so that,
using \eqref{eq:defPo},
\begin{equation}
\label{eq:Horm} 
  \begin{split} i \{ q , \bar q \} & = i (  8 \bar \zeta \bar \partial_z + 
\partial_z V \partial_\zeta ) (4 \zeta^2 - \overline V ) = 
 8 i (  \zeta \partial_z V - \overline{ \zeta \partial_z V }   ) \\ 
& = -  16  \Im ( \zeta \partial_z V )  
 = - 8  \Im ( {\overline{ V}}^{\frac12}   \partial_z V )  . 
\end{split}
\end{equation}
We need to verify that the right-hand side is non-zero at some point $ z_0$,
as that will remain valid in an open neighbourhood of $z_0 $.

To do so we write the expression $\Im ( {\overline{ V}}^{\frac12}   \partial_z V ) $ from \eqref{eq:Horm} as a Taylor series at the origin.  With $f_n$ given by \eqref{eq:fn} we observe that $f_n(0)=0$ for all $n$, and that
\begin{align*}
\partial_z f_n(0)=\tfrac{n}{2}\sum_{k=0}^2e^{\frac{n}{2}(z\bar \omega^k-\bar z\omega^k)}\Big\vert_{z=0}=\tfrac{3n}2,\qquad
\partial_{\bar z} f_n(0)=-\tfrac{n}{2}\sum_{k=0}^2\omega^{2k}e^{\frac{n}{2}(z\bar \omega^k-\bar z\omega^k)}\Big\vert_{z=0}=0,
\end{align*}
since $\omega^4=\omega$ and $1+\omega+\omega^2=0$.
Hence,
\begin{equation}\label{eq:TaylorU}
U(z,\bar z)=\partial_z U(0)z+O(\lvert z\rvert^2),\qquad \partial_z U(0)=\tfrac{3}{2}\sum_{n=3\ZZ+1}na_{n}.
\end{equation}
Recall that $V(z)=U(z)U(-z)$. Since $U(0)=\partial_{\bar z} U(0)=0$, we have $V(0)=\partial_z V(0)=\partial_{\bar z} V(0)=0$, and
$$
\partial_z^2 V(0)=-2(\partial_z U(0))^2,\qquad \partial_z\partial_{\bar z} V(0)=\partial_{\bar z}^2 V(0)=0.
$$
It follows that
$$
V(z)=-z^2(\partial_z U(0))^2(1+O(\lvert z\rvert)),\qquad \partial_zV(z)=-2z(\partial_z U(0))^2(1+O(\lvert z\rvert)),
$$
which gives
\begin{align*}
{\overline{ V}}^{\frac12}(z)\partial_zV(z)&=\overline{\sqrt{-z^2(\partial_z U(0))^2}}(-2z(\partial_z U(0))^2)(1+O(\lvert z\rvert))\\
&=2i\lvert z\rvert^2\lvert \partial_z U(0)\rvert^2\partial_z U(0)(1+O(\lvert z\rvert)).
\end{align*}
From this we see that $\Im ( {\overline{ V}}^{\frac12}   \partial_z V )\ne0$ in a punctured neighbourhood of the origin if $\Re \partial_z U(0)\ne0$, which in view of \eqref{eq:TaylorU} holds by virtue of the non-triviality assumption \eqref{eq:conditionbracket}. This completes the proof.
\end{proof}

\noindent
{\bf Remark.} The open set on which the right-hand side of \eqref{eq:Horm} does not vanish can be easily determined numerically, and it is 
a complement of a one dimensional set -- see Figure~\ref{f:Horm}.

To prove Theorem \ref{t:squeezegenpot} we will use the following fact,
with the proof left to the reader:
\begin{prop}
\label{p:indie}
Suppose that $ g_n \in L^2 ( \CC/\Gamma ) $, $ n \in \ZZ^2 $, 
$ |n| \leq N $ satisfy $ |\langle g_n, g_m \rangle |
\leq  e^{ - M | n-m|^2 } $, $  \langle g_n, g_n \rangle = 1 $. If $ M > 3 $ 
then the set $ \{ g_n \}_{ |n| \leq N } $ is linearly independent
in $ L^2 ( \CC/\Gamma ) $. \qed
\end{prop}

We can now give
\begin{proof}[{Proof of Theorem \ref{t:squeezegenpot}}]
In the notation of Proposition \ref{p:dsz},
let $ C = [ a, b ] \times [c,d]\Subset \Omega $ and consider
the finite set $ \mathscr Z_h :=  K \sqrt h \ZZ^2 \cap C $, 
$ | \mathscr Z_h | \sim 1/h $. Then \eqref{eq:Phz} 
gives $ \mathbf u^w_h $, $ w \in \mathscr Z_h $ (with $ z_0 $ replaced by 
$ w$). Let $ M \gg 1 $. Using  
$| w - z |^2 + |w' - z|^2 = \tfrac12 | w - w'|^2 + 2 | z - \tfrac12 ( w + w')|^2 $, and taking $ K $ large enough, we obtain from \eqref{eq:Phz}
\begin{equation} 
\label{eq:ww} | \langle \mathbf u_h^w , \mathbf u_h^{w'} \rangle | \leq e^{ - M |n-n'|^2}, \ \  n := \tfrac{w}{K \sqrt h},  \ \ n' := \tfrac{w'}{K\sqrt h} \in  \ZZ^2, 
 \ \  \| \mathbf u_h^w \|_{L^2} = 1.
\end{equation}
Abusing notation, let us identify $ \mathbf u_h^w $ with 
$ ( \mathbf u_h^w , 0_{\CC^2 } ) \in L^2 ( \CC/\Gamma; \CC^4 ) $, with 
\eqref{eq:ww} unchanged. We then have 
\begin{equation}
\label{eq:quas2}  \| H_{\mathbf k}  ( \alpha ) \mathbf u_h^w \|_{ L^2 ( \CC/\Gamma ) } 
\leq e^{ - c' /h } , \ \ h = 1/\alpha . \end{equation}
Using self-adjointness of $ H_{\mathbf k} $ and in the notation 
of Theorem \ref{t:squeezegenpot}, write 
\[ H_{\mathbf k } ( \alpha ) \mathbf v = \sum_{ j \in \ZZ } E_j ( \mathbf k , \alpha ) 
\mathbf g_j \langle \mathbf v , \mathbf g_j \rangle , \ \ \ 
H_{\mathbf k } ( \alpha)  \mathbf g_j = E_j ( \mathbf k, \alpha ) 
\mathbf g_j , \ \ \langle \mathbf g_j, \mathbf g_i \rangle = \delta_{ij} .\]
Then \eqref{eq:quas2} implies that
$ \sum_{ | E_{ j } (\mathbf k , \alpha ) | \geq e^{ -c'/2h } }
\mathbf g_j \langle \mathbf u_h^w , \mathbf g_j \rangle = 
\mathcal O ( e^{ -c'/2h } )_{L^2 } $, which gives
\[  \dim {\rm span} \{ \mathbf g_j \}_{ |E_j ( \mathbf k , \alpha ) | < 
e^{ -c'/2h } } \geq  \dim {\rm span} \{ \mathbf u_h^w \}
_{ w\in \mathscr Z_h } . \]
But \eqref{eq:ww} and Proposition \ref{p:indie} show that 
 the right hand side is given by $ \mathscr Z_h \sim 1/h $. This completes
the proof.
\end{proof}

\noindent
{\bf Remark.} This simple argument showing exponential squeezing of bands
does not apply to the more realistic Bistritzer--MacDonald model of 
twisted bilayer graphene \cite{BM11}. In that case,
a more complicated non-self-adjoint system can be extracted from the analogue of $ H( \alpha ) $, but whenever  eigenvalues of the symbol
(the analogue of \eqref{eq:littlep}), $ \lambda $, are simple, 
the Poisson bracket $ \{\lambda, \bar \lambda \}|_{\lambda = 0 } $ vanishes
\cite{suppl}.


\section{Numerical results}
\label{s:num}

The results are numerically implemented using rectangular coordinates
$ z = x_1 + i x_2 =  2i \omega y_1 + 2 i \omega^2 y_2 $, 
see \S \ref{s:trace}. 
We then consider
\[  H_{\mathbf k } ( \alpha) = \begin{pmatrix} 0 & D_{\mathbf k} ( \alpha)^*
\\ D_{\mathbf k} ( \alpha) & 0 \end{pmatrix} , \ \ 
{\mathbf k = ( \omega^2 k_1 - \omega k_2)/\sqrt 3,} \]
where $ D_{\mathbf k} ( \alpha ) $ is given in 
\eqref{eq:Dkal},
with{\em periodic} boundary conditions (for $ y \mapsto
y + 2 \pi \mathbf n  $, $ \mathbf n  \in  \ZZ^2 $).
For a fundamental domain in $ \mathbf k $ we choose
$  \Omega := \{  (k_1, k_2 ) ;  -\tfrac12 \leq k_j < \tfrac12 \}$. 

\subsection{Numerical implementation} \label{s:sum} The discretization is given 
using a Fourier spectral method; see~\cite[Chapter~3]{Tre00}.
Using the tensor structure of
$ \mathscr D_{\mathbf k} $ and $ \mathscr V $ we start with the standard orthonormal basis of $ L^2 (\RR^2 / 2 \pi \ZZ^2) $: $  e_{\mathbf n } (y) := e_{n_1} \otimes e_{n_2} ( y ) := e_{n_1} ( y_1 ) e_{n_2 } ( y_2 ) $, 
$ e_\ell ( t ) := ( 2 \pi)^{-\frac12} e^{ i \ell t } $. 
Using the identification  $ [ -N , N ] \cap \ZZ \simeq \ZZ_{2N+1} $, we define
\[  \begin{gathered} \Pi_N :  L^2 ( \RR^2 / 2\pi \ZZ^2; \CC^2 ) \to \ell^2 ( \ZZ_{2N+1}^2 ; \CC^2 ) = \ell^2 ( \ZZ_{2N+1} ; \CC^2 ) 
\otimes \ell^2 ( \ZZ_{2N+1} ; \CC^2 ) 
, \\ \Pi_N \left( \sum_{ \mathbf n \in \ZZ^2 } 
a_{\mathbf n } e^{ i \langle y , \mathbf n \rangle } \right) = 
\{ a_{(n_1, n_2 ) }\}_{ |n_j| \leq N } , \ \  a_{\mathbf n } \in \CC^2 , \ \ \mathbf n = ( n_1, n_2 ) \in \ZZ^2 , \end{gathered} \]
and $  D_{\mathbf k}^N ( \alpha ) :=  \Pi_N D_{\mathbf k } ( \alpha ) 
\Pi_N^* $. Hence, 
\begin{equation*}
\label{eq:DkN}  D_{\mathbf k}^N ( \alpha ) = 
\tfrac{1}{ \sqrt 3 } \begin{pmatrix} 
\mathscr D^N_{\mathbf k }  & 
 \alpha \mathscr V^N_+  \\
\alpha \mathscr V^N_-  & \mathscr D_{\mathbf k}^N 
\end{pmatrix},
\end{equation*}
where (with $ D^N :=  {\rm{diag}}\, ( \ell )_{ - N \leq | \ell | \leq N }$
and $ J_N $ the $2N+1$ dimensional Jordan block)
\begin{align*}
 \mathscr D^N_{\mathbf k } &:= \omega^2 ( D^N + k_1 I_{ \CC^{2N+1} } )
\otimes I_{ \CC^{2N + 1 }} -
\omega  I_{\CC^{2N+1} } \otimes ( D^N + k_2 I_{\CC^{2N+1}} ) ,\\
\mathscr V_+^N &:= \sqrt 3 ( J_N \otimes J_N + \omega \, (\!J^{2}_N)^t \otimes 
J_N + \omega^2 J_N \otimes (J^{2}_N)^t), \\ 
\mathscr V_-^N &:= \sqrt 3 ( (J_N)^t \otimes (J_N)^t + \omega J_N^2 \otimes 
(J_N)^t + \omega^2 \, (J_N)^t \otimes  J^{2}_N ).
 \end{align*}
The matrix $D_{\mathbf k}^N(\alpha)$ has dimension $2(2N+1)^2$.
To obtain reasonable accuracy up through the second magic $\alpha$,
one should at least use $N=16$ (giving a matrix of dimension $\mbox{2,178}$);
for the range $\alpha\in[0,15]$ in Figures~\ref{f:squeeze} and~\ref{f:res3}, 
we use $N=96$ (giving dimension $\mbox{74,498}$).
It is expedient in the former case, and essential in the latter,
to use sparse-matrix algorithms that take advantage of the many
zero entries in $D_{\mathbf k}^N(\alpha)$.
To compute the smallest singular values of $D_{\bf k}^N(\alpha)$,
we use Krylov subspace methods, either the inverse Lanczos algorithm
adapted from~\cite{Tr99,Wr02} or the augmented implicitly 
restarted Lanczos method~\cite{BR05} implemented in MATLAB's 
{\tt svds} command.

\begin{figure}
\begin{center}
\includegraphics[width=12cm]{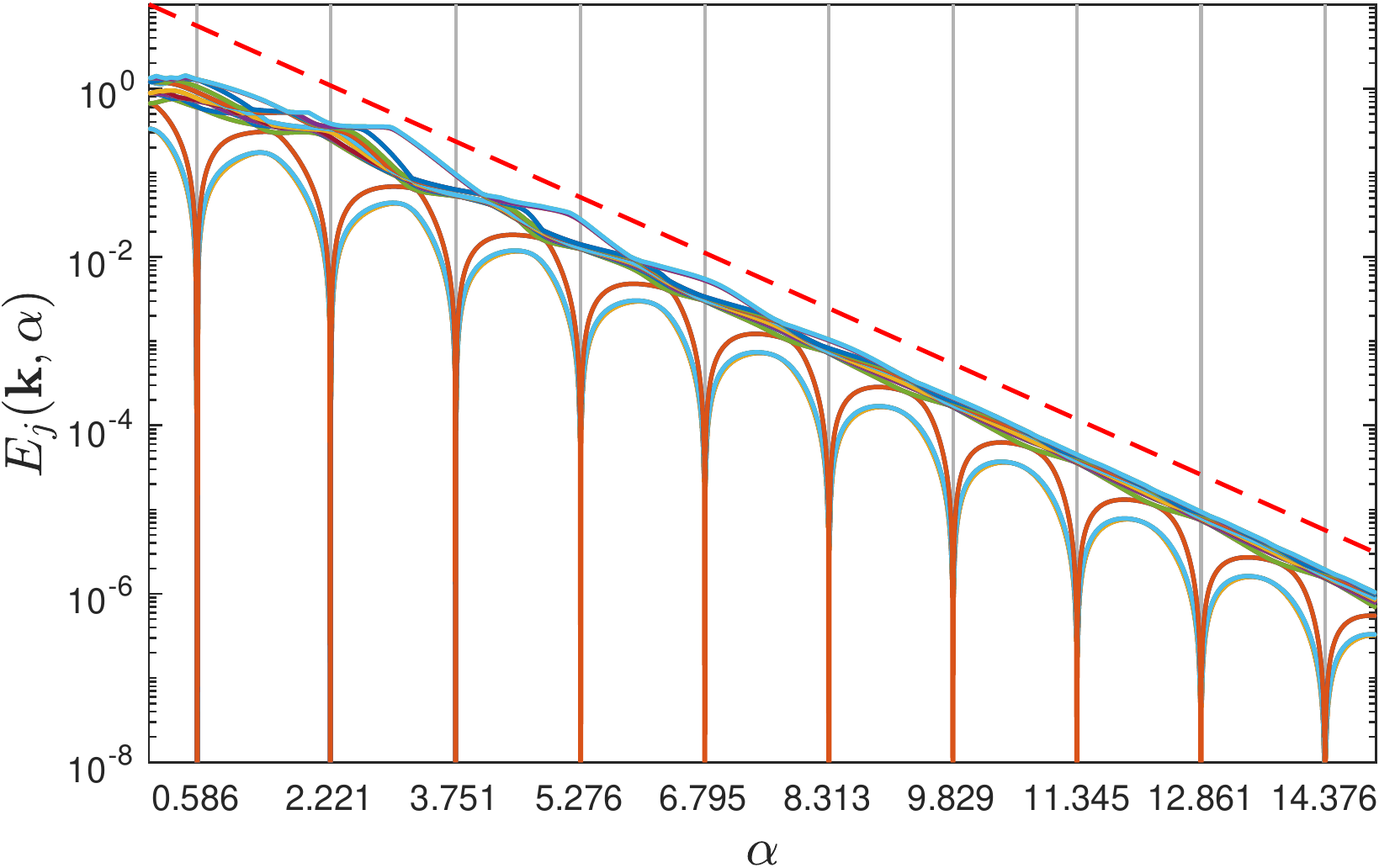}
\begin{picture}(0,0)
\put(-119,120){\rotatebox{-24}{\color{white}{\rule{28pt}{10pt}}}}
\put(-119,120){\rotatebox{-24}{\small $10 e^{-\alpha}$}}
\end{picture}
\end{center}
\caption{\label{f:squeeze}
Numerical confirmation for Theorem~\ref{t:squeeze}:
Computed eigenvalues $E_0({\mathbf k},\alpha)$, \ldots, $E_{40}({\bf k},\alpha)$
of $H_{\bf k}(\alpha)$ for ${\mathbf k}_* = 1/(2\sqrt{3}) + i/6$ (see Figure~\ref{f:res3}).  Numerous eigenvalues are quite close together or have high multiplicity.
}
\end{figure}
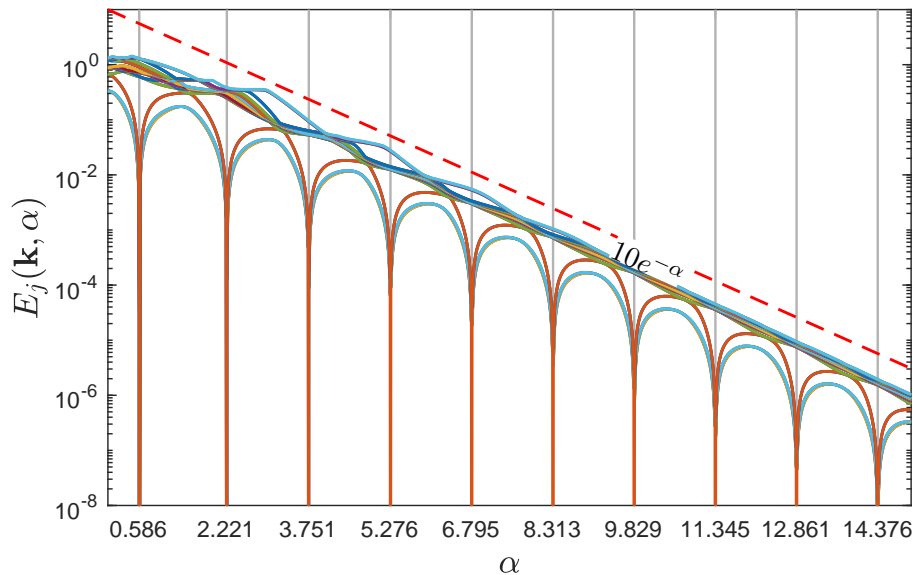

Figure~\ref{f:squeeze} shows numerical calculations of the first 41~non-negative eigenvalues of $H_{\bf k}(\alpha)$.  As required by Theorem~\ref{t:squeeze}, these eigenvalues decay exponentially, apparently no slower than $e^{-\alpha}$.
The vertical lines in the figure indicate the magic $\alpha$ values.
We pursue two approaches to 
locating these magic $\alpha\in {\mathcal A}_{\rm mag} $ (see \eqref{eq:defResa} and Theorem~\ref{t:spec}).
The spectral characterization of the set $ \mathcal A $ of resonant $\alpha$'s
via the operator $ T_{\mathbf k } $ enables the precise calculation of many
points in $ \mathcal A $ as reciprocals of eigenvalues of the discretisation 
\[  T_{\mathbf k}^N := 
\begin{pmatrix} 
0 & (\mathscr D_{\mathbf k}^N)^{-1} \mathscr V^N_+ \\ (\mathscr D_{\mathbf k}^N)^{-1} \mathscr V^N_- & 0 \end{pmatrix}. \]
To reduce dimensions (and multiplicities) we consider these operators
in the decomposition of $ L^2 ( \RR/ 2 \pi \ZZ ) $ 
in terms representations of $ \Gamma_3/\Gamma \simeq \ZZ_3^2 $ (we did not 
use the full symmetry group $ G_3 $ -- see \eqref{eq:defG3}). 
We used this approach to compute Figure~\ref{f:resa} and to get
initial estimates of the values in Table~\ref{tbl:magic};
note however that for large $|\alpha|$ the non-self-adjointness of
$T_{\mathbf k}^N$ limits the precision to which these eigenvalues 
can be computed.  (This pseudospectral effect is a more significant
obstacle to high precision than the
errors introduced by truncation to finite $N$.)  

\begin{figure}
\begin{center}
\includegraphics[height=2.2in]{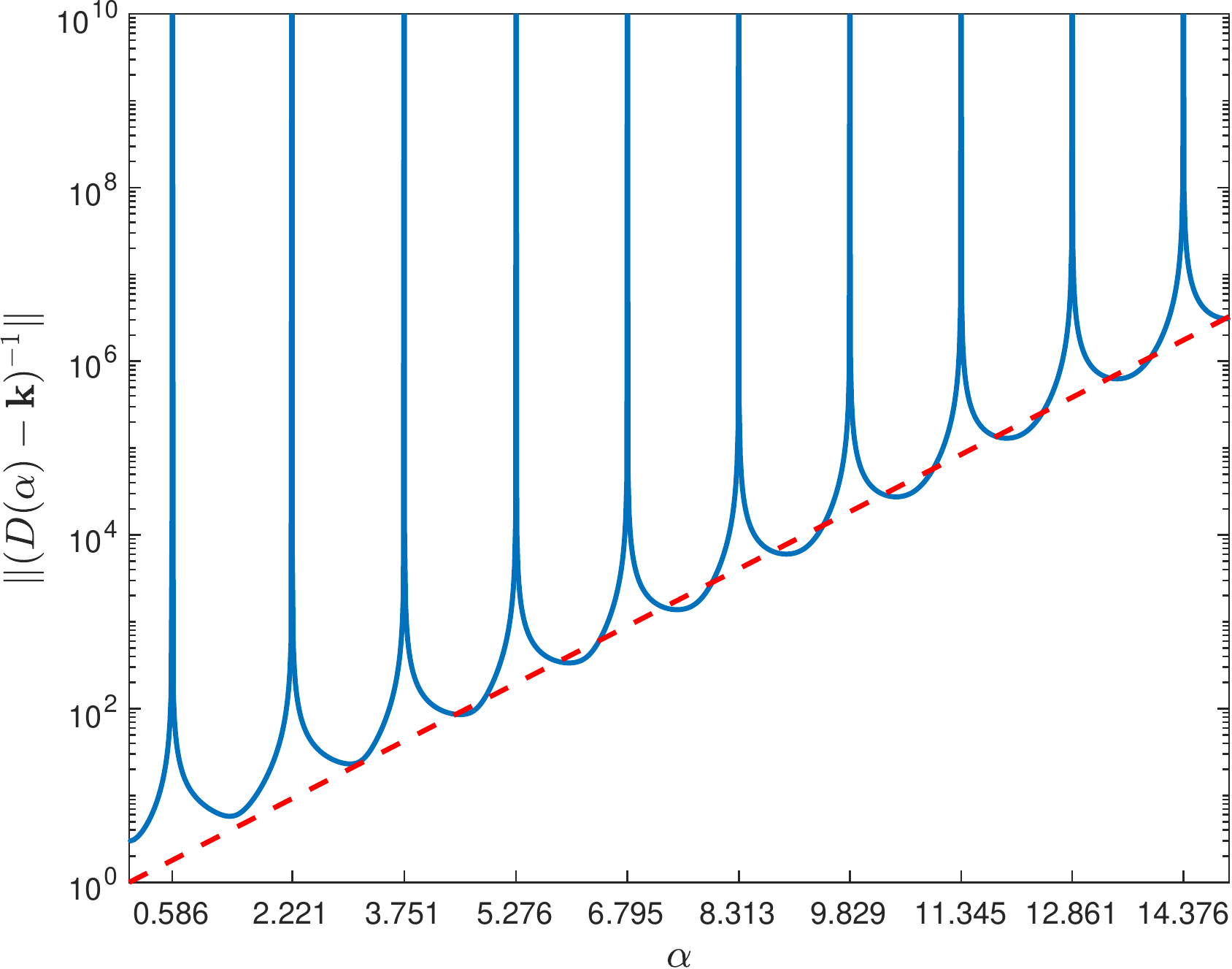}\quad
\raisebox{6pt}{\includegraphics[height=2.1in]{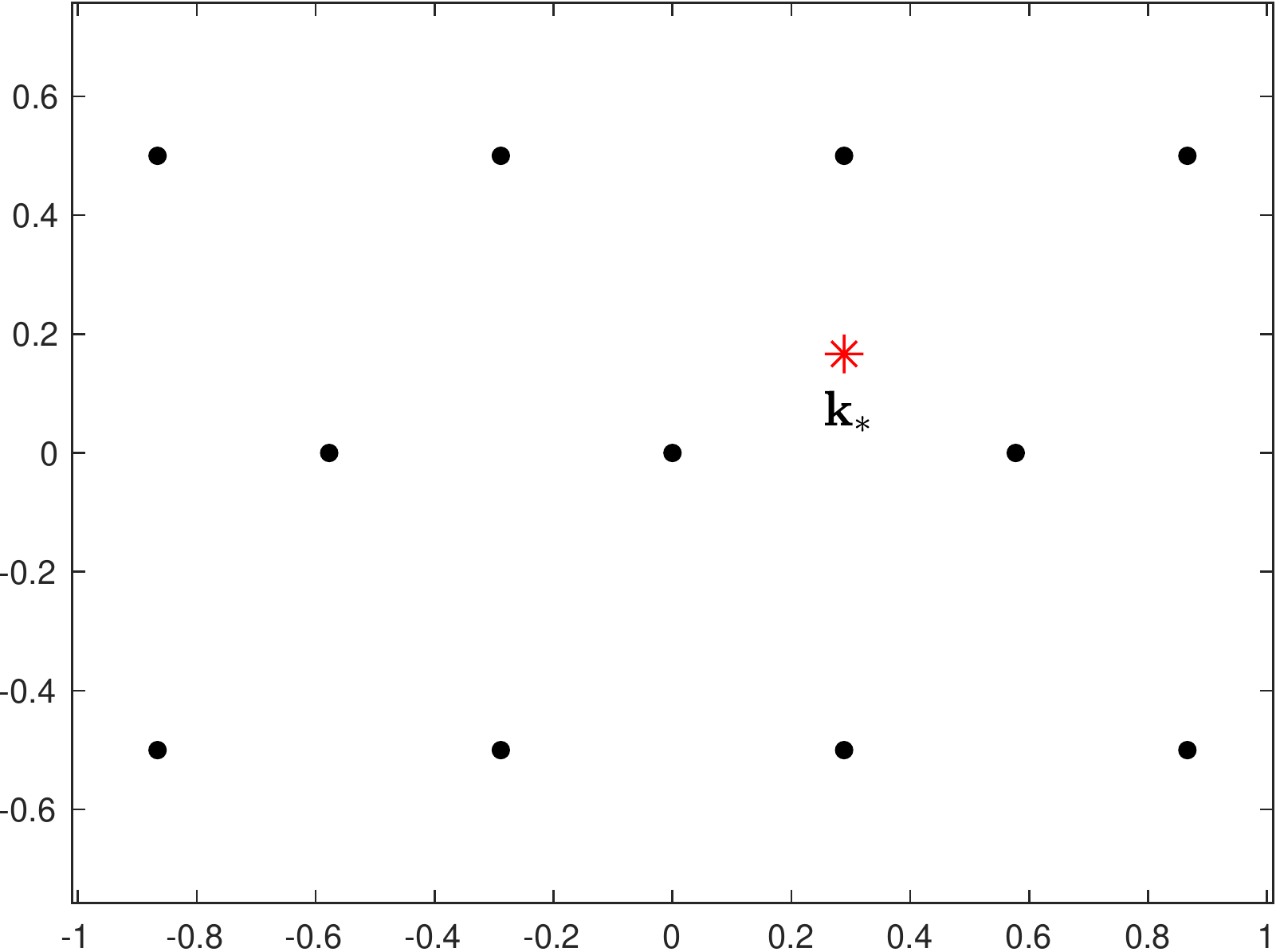}}
\end{center}
\caption{\label{f:res3}
On the left, the norm of the resolvent $(D(\alpha)-{\bf k})^{-1}$ at ${\bf k}_* = 1/(2\sqrt{3}) + i/6$,
a point equidistant from three eigenvalues of $D(\alpha)$ for $\alpha\not \in {\mathcal A}$.
The red dashed line shows $e^\alpha$.
The right shows a portion of ${\mathrm Spec}_{L^2(\CC/\Gamma)} D(\alpha) = \Gamma^*$ for $\alpha \not\in {\mathcal A}$.
}
\end{figure}

\begin{table}
\caption{\label{tbl:magic}
Estimates of the first thirteen magic $\alpha$'s, truncated 
(not rounded) to digits supported with high confidence by our numerics.  
The last column shows the difference between consecutive magic $\alpha$'s,
which seem to converge a bit above the conjecture of $3/2$ in~\cite{magic}.
}


\begin{center}
\begin{tabular}{rclcc}
\multicolumn{1}{c}{$k$} & &
\multicolumn{1}{c}{$\alpha_k$} &
 & $\alpha_{k}-\alpha_{k-1}$ \\[2pt] \hline
1  &\ &   \phantom{0}0.58566355838955 &\ &               \\
2  &&   \phantom{0}2.2211821738201  && 1.6355        \\
3  &&   \phantom{0}3.7514055099052  && 1.5302        \\
4  &&   \phantom{0}5.276497782985   && 1.5251        \\
5  &&   \phantom{0}6.79478505720    && 1.5183        \\
6  &&   \phantom{0}8.3129991933     && 1.5182        \\
7  &&   \phantom{0}9.829066969      && 1.5161        \\
8  &&    11.34534068       && 1.5163        \\
9  &&    12.8606086        && 1.5153        \\
10 &&    14.376072         && 1.5155        \\
11 &&    15.89096          && 1.5149        \\
12 &&    17.4060           && 1.5150        \\
13 &&    18.920            && 1.5147        \\
\end{tabular}
\end{center}

\end{table}

To understand the accuracy of the values in Table~\ref{tbl:magic},
we studied $\|(D_{\mathbf k}^N (\alpha ) )^{-1}\|$ near
the putative magic $\alpha$ values.  
Figure~\ref{f:res3} reveals the computational challenge of resolving large
magic angles to high fidelity.  
One can characterize the magic $\alpha$'s as points where 
$  (D(\alpha) - {\bf k})^{-1} $ does not exist, and hence they are
approximated by $\alpha$'s for which
$\|D_{\bf k}^N ( \alpha)^{-1}\|$ is very large for generic ${\bf k}$.
Careful scanning for $\alpha$'s around magic values (using 
$N=96$ and $N=128$) refines the estimates and indicates their accuracy. 
Overall, as $\alpha$ increases
$\|D_{\mathbf k}^N(\alpha)^{-1}\|$ grows exponentially
(as guaranteed by Theorem~\ref{t:squeeze}, since 
$\|D_{\mathbf k}^N(\alpha)^{-1}\| = 1/E_0({\mathbf k},\alpha)$), so that
precisely locating large $\|D_{\mathbf k}^N(\alpha)^{-1}\|$ values
against this growing background becomes increasingly challenging.
Indeed, this numerical struggle nicely parallels the presumed
diminishing physical significance of large magic $\alpha$ values 
(corresponding, as they do, to reciprocals of angles of twisting).

\subsection{Error bounds} \label{s:error}
{ Assuming accuracy of matrix calculations it is possible to give error bounds
for the approximation of the actual magic $ \alpha$'s. We consider the general
situation in which $  {B} \in \mathcal L_1 ( H ) $ (a trace class operator on a Hilbert space) is 
approximated by a  $m(N)$-by-$m(N)$ matrix, (in our case 
$ m(N) = ( 2 N + 1)^2 $) where
 \begin{equation}
\label{eq:approx}
B =  B_N + E_N , \ \  \| E_N \|_{1} \leq \rho_1 (N)  / N^6 , \ \ \| E_N \| \leq \rho_0 ( N )  /N^8 , 
\end{equation}
where $ \| \bullet \|_{1} $ and $ \| \bullet \| $ are trace class and operator
norms, respectively. (The strange look of the estimates is explained by the statement
of Proposition \ref{l:imp8}.)

Suppose that the matrix $ B_N $ has a simple eigenvalue $ \mu_N \in \RR $
(computed numerically) 
and that (by a numerical calculation)
\begin{equation}
\label{eq:eps0}       \| ( B_N - \lambda_j )^{-1} \| \leq  C^0_N ( \epsilon)  , \ \ \lambda_j := 
\mu_N + \epsilon  {e^{  2 \pi i j/ J }}, \ \ \ j = 0 , 1 , \cdots, J-1 .\end{equation}
We then have, for {\em all} $ \lambda $ with $ | \lambda - \mu_N | = \epsilon $, 
\begin{equation}
\label{eq:eps1}
 2\epsilon  {C_N^0} ( \epsilon ) \sin ( \pi/  {2}J)  < \delta
\ \Longrightarrow \ \| ( B_N - \lambda )^{-1} \| \leq C_N ( \epsilon ) := C^0_N ( \epsilon ) 
( 1 - \delta)^{-1} .
\end{equation}
We then note that for $ |\lambda - \mu_N | = \epsilon $, 
\begin{equation}
\label{eq:eps2}  
\begin{gathered} 
C_N ( \epsilon )   \rho_0 ( N ) / N^8 < \delta \ \Longrightarrow \ 
 ( B - \lambda)^{-1} = ( B_N - \lambda)^{-1} ( {I} - D_N ( \lambda ) ) , \\ 
 D_N ( \lambda ) := E_N ( \lambda ) ( B_N - \lambda)^{-1} ( I + 
 E_N ( \lambda )  ( B_N - \lambda)^{-1} )^{-1} , \\
 \| D_N ( \lambda ) \|_{1} <  C_N ( \epsilon ) \rho_1 ( N ) / N^6 ( 1 - \delta ) .
 \end{gathered}
 \end{equation} 
 These bounds lead to an estimate of the trace class norm: if the assumptions in 
 \eqref{eq:eps1}, using here the larger constant $C_N$ instead of $C_N^0$, and \eqref{eq:eps2} hold:
\begin{equation}  
\label{eq:eps3} 2\epsilon C_N ( \epsilon ) \sin ( \pi/  {2}J)   < \delta, \ \ \ C_N ( \epsilon )   
\rho_0 ( N ) / N^8 < \delta,
\end{equation}
where $ \rho_0 ( N ) $ is defined in \eqref{eq:approx} and $ C_N ( \epsilon ) $ in \eqref{eq:eps1}, then  
\begin{equation}
\label{eq:trace}
\| ( B - \lambda)^{-1} - ( B_N - \lambda )^{-1} \|_{1} <  C_N( \epsilon )^2  \rho_1 (N ) / N^6 ( 1- \delta)  .
\end{equation}
If we define spectral projectors
\begin{equation}
\label{eq:PN} 
\begin{split}  P ( \epsilon )  := \frac{1}{2 \pi i} \oint_{ |\lambda - \mu_N | = \epsilon } ( \lambda - B )^{-1} d \lambda , \ \ \ 
P_N ( \epsilon ) := \frac{1}{2 \pi i} \oint_{ |\lambda - \mu_N  | = \epsilon } ( \lambda - B_N )^{-1} 
d \lambda , \end{split}  \end{equation}
we see that if \eqref{eq:eps3} holds then
\begin{equation}
\label{eq:eps4} 
 \epsilon C_N( \epsilon)^2  \rho_1 ( N ) / N^6 ( 1- \delta ) < 1 \ \Longrightarrow \   \tr P = \tr P_N = 1 , 
\end{equation}
that is, we have a simple eigenvalue of  {$B$} within $ \epsilon $ of $ \mu_N $:
\begin{equation}
\label{eq:eps5}
| \Spec ( B ) \cap D ( \mu_N , \epsilon ) | = 1. 
\end{equation}
If we know that the eigenvalues of $ B $ are symmetric with respect to $ \RR $ 
it follows that $ B $ has a real eigenvalue in $  ( \mu_N - \epsilon , \mu_N + \epsilon ) $.

We now implement this for the operator $ B = B_{\mathbf k}  = 3 A_{\mathbf k } $,
$ \mathbf k \notin \Gamma^* $, where $ A_{\mathbf k }$  is the operator defined
in \eqref{eq:defA}. The Hilbert space is the symmetry reduced $ L^2 $:
\begin{equation}
\label{eq:L2G} H =  L^2_0 ( \CC/ \Gamma ) := \{ u \in L^2 ( \CC/\Gamma ) : u ( z + 
\gamma ) = u ( z ) , \gamma \in \Gamma_3/\Gamma \} , \end{equation}
where $   \Gamma_3 = 
\tfrac 43 \pi i ( \omega \ZZ \oplus \omega^2 \ZZ ) $, $ \Gamma = 3 \Gamma_3 $ -- see \eqref{eq:defGam3}.

We start with the computation of the constants in \eqref{eq:approx}. 
Let $T$ be a compact operator and $  \| T \|_p $ its $p$-Schatten norm:
\[    \| T \|_p= \| T \|_{\mathcal L_p ( H ) }  := \left( \sum_{j=0}^\infty s_j ( T )^p \right)^{\frac1p} , \ \ \ 
T \in \mathcal L^p ( H ) \ \Longleftrightarrow \ \| T \|_p < \infty , \]
where $ s_j ( T ) $ are the singular values of $ T $ -- see \cite[\S B.3]{res}.
In the notation of  {\S}\ref{s:sum}, we let $\pi_{N} := I - \Pi_N $. 
For $p \ge 3,$ $M\ge 2$, and $ \mathbf k =  ( \omega^2 k_1 -
 \omega k_2 )/\sqrt 3$, $ ( k_1, k_2 ) \in ( 0,1 )^2 $,  we claim
\begin{equation}
\label{eq:gammap}
\gamma_p:= \sup_{M \ge 2} \frac{\Vert \pi_M D(\mathbf k)^{-1} \Vert_p^p} {{(M-1)^{2-p}} }
\leq \frac{ 2 \pi  \, 6^{p/2}  } {\sqrt 3 (p-2)} .
\end{equation}
In fact, 
\begin{equation} 
\label{eq:gam} \begin{split} \| \pi_{M} D(\mathbf k )^{-1} \|_p^p & = 
 3^{p/2} \sum_{|m|> M \vee |n| > M } | ( m + k_1 ) - \omega^2 ( n + k_2 ) |^{-p} 
\\ 
&\le 3^{p/2}  \sum_{|m|\geq M \vee |n| \geq M } | m^2 +mn
 + n^2  |^{-p/2} \\
&\le 3^{p/2} \int_{M-1}^{\infty} \int_{0}^{2\pi} \frac{1}{r^{p-1}(1+\cos(\varphi)\sin(\varphi))^{p/2}} \ d\varphi \ dr \\
&= 3^{p/2} \frac{ (M-1)^{2-p}}{p-2} \int_{0}^{2\pi} \frac{1}{(1+\tfrac{1}{2}\sin(2\varphi))^{p/2}} \ d\varphi 
\\ & \le 
\frac{ 2 \pi  6^{p/2}  } {\sqrt 3}\frac{ (M-1)^{2-p}}{p-2}
\end{split} \end{equation}
where we used, with $ f ( \varphi ) := ( 1 + \tfrac12 \sin 2 \varphi )^{-1/2 } $, 
\[ \| f \|_2^2 = \frac{ 4 \pi} {\sqrt 3} , \ \ 
\| f \|_{\infty} = 2^{\frac12} ,  \ \   \| f \|_p^{p} \leq \| f \|_2^2 \| f \|_\infty^{p-2} .\]
(The integral can also be estimated very accurately using the method of steepest descent.) 
In addition, we observe that for the operator norm and $M \ge 1$,
\begin{equation}
\begin{split}
\label{eq:op_norm}
\Vert \pi_{M}D(\mathbf k )^{-1} \Vert \le \sqrt 3  \sup_{ \vert m \vert \ge M \vee \vert n \vert \ge M} 
( m^2 +mn
 + n^2  )^{-\frac12} \le {2 }/{M}.
\end{split}
\end{equation}

We used these estimates to compare finite rank operators used in numerical calculations
to powers of $ T_{\mathbf k}^p $:
\begin{prop}
\label{l:trace} 
In the notation of \S \ref{s:sum}, and with $k_1,k_2 \in (-1,1)$, $N\ge 2p\ge 6$, we have 
$$\Vert T_{\mathbf k}^p - \Pi_N\, T_{\mathbf k}^p  \, \Pi_N \Vert_1 \le 
\frac{ 4 \pi 54^{ p/2}  \rho_1(N,p)} { \sqrt 3 ( p-2) N^{p-2} }$$
and in operator norm 
$$\Vert T_{\mathbf k}^p - \Pi_N \,  T_{\mathbf k}^p \, \Pi_N \Vert \le 
 6^p 2 \rho_0(N,p) N^{-p},
$$
where
\begin{equation}
\label{eq:rhoj}  \rho_j ( N, p ) = \prod_{\ell=0}^{p-1} ( 1 - ( 2 \ell + j) /N )^{-1+\frac{2j}p } .
\end{equation}
\end{prop}
\begin{proof}
We first observe that
\[ \begin{split} \Vert T_{\mathbf k}^p - \Pi_N\, T_{\mathbf k}^p  \, \Pi_N \Vert_1& = 
 \Vert T_{\mathbf k}^p - ( I - \pi_N) T_{\mathbf k}^p + ( I - \pi_N) T_{\mathbf k}^p 
 \pi_N \Vert_1 
\\ & \leq \| \pi_N T_{\mathbf k }^p \|_1 + \| T_{\mathbf k }^p \pi_N \|_1 , \ \ \ \pi_N = I - \Pi_N .
\end{split} \]
We will estimate the first term, with a same argument applicable to the second term.

Letting $ T = T_{\mathbf k} $, we write $T = D(\mathbf k)^{-1} V $,
where $ V $ is the potential with $ U ( z ) $ and $ U ( - z ) $ on 
the antidiagonal.  We note that $ \| V \| \leq 3 $. 
By analysing the potential in \eqref{eq:defU} we find that
\begin{equation}
\label{eq:pint} 
\pi_N T = \pi_N T \pi_{N-2}  . 
\end{equation}
Hence (using Schatten norms)
\begin{equation}
\begin{split}
&\| \pi_N T^p \|_{1} \leq \prod_{ \ell=0}^{p-1} \| \pi_{N-2 \ell } T \|_{p} \leq
3^p  \prod_{ \ell=0}^{p-1} \| \pi_{N-2 \ell} D(\mathbf k )^{-1} \|_p .
\end{split}
\end{equation}
For $M\ge 2$, \eqref{eq:gammap} gives
\begin{equation}
\label{eq:gamp}  \| \pi_{M} D(\mathbf k )^{-1} \|_p \leq \gamma_p^{\frac1p}   (M-1)^{-1+\frac 2 p } , \ \ 
p \geq 3, \ \ 
\end{equation}
and hence we have, using \eqref{eq:gammap} and  \eqref{eq:rhoj}, 
$$ \|  \pi_N T^p   \|_{1 }
 \leq \frac{ 2 \pi 54^{ p/2}  \rho_1 ( N, p ) } { \sqrt 3 ( p-2) N^{p-2} }. $$
Combined with the same estimate for $ \| T^p \pi_N \|_1 $
this implies the result.
The operator norm estimate is fully analogous, using \eqref{eq:op_norm}.
\end{proof}

We recall that $\mathscr L_{\mathbf a} $ 
commutes with $ D_{\mathbf k}  (0 )$ and $  V $, where $ V $ is as in the proof of Proposition \ref{l:trace}. 
It also commutes with $ \Pi_N $ since pull backs by translations and multiplication by constants do not change orders of trigonometric polynomials. This gives an action of $ \ZZ_3^2 $ on 
$ L^2 ( \CC/\Gamma, \CC^2) $ which can then be decomposed using nine irreducible 
representations of that group \eqref{eq:reprZ32}:
\[   L^2_{\mathbf p} ( \CC/\Gamma ; \CC^2 ) = \{ \mathbf u \in 
L^2 ( \CC/\Gamma ; \CC^2 ) :   \mathscr L_{\mathbf a } \mathbf u = \pi_{\mathbf p} ( 
\mathbf a ) \mathbf u \}, \]
where 
$ \mathbf p = ( \omega^2 p_1 - \omega p_2 ) /\sqrt 3, \ \ p_j \in \ZZ_3 $, 
$ \pi_{\mathbf p} ( \mathbf a ) = \exp ( i \Re ( \mathbf a \bar {\mathbf p } )) $.
We then specialize to this symmetry reduced case and power $ p = 8 $. The former 
gives a small improvement:
\begin{prop}
\label{l:imp8}
Suppose that $ B = B_{\mathbf k } = 3 A_{\mathbf k } $, $ \mathbf k = \omega^2/2 \sqrt 3$,
where $ A_{\mathbf k} $ comes from \eqref{eq:defA} 
and $ H $ is given by \eqref{eq:L2G}. Then, with $ \Pi_N $ given in \S \ref{s:sum}, 
and $\rho_j $ defined in \eqref{eq:rhoj},
\begin{equation}
\label{eq:norm8}
\begin{split} 
& \| B^4 - \Pi_N B^4 \Pi_N \|_{  \mathcal L _1( L^2_{{\mathbf 0 }} ( \CC/\Gamma, \CC^2 ) )}^{\frac16}  \leq  10.23 N^{-1} \rho_1 ( 8, N )^{\frac16}  , \\
& \| B^4 - \Pi_N B^4 \Pi_N \|_{  \mathcal L( L^2  ( \CC/\Gamma, \CC^2 )) } \leq   6^8 2 \rho_0 ( 8 , N )  N^{-8}.
\end{split} 
\end{equation}
Moreover, at every magic angle, $\alpha \in \mathcal A$, the Hamiltonian $H(\alpha)$ exhibits at least 18 flat bands. 
\end{prop}
\begin{proof}
We observe that we have unitary equivalence,
\[ U_{\mathbf p } \mathbf u ( z ) : 
L^2_{\mathbf q} ( \CC/\Gamma ; \CC^2 ) \to L^2_{\mathbf p + \mathbf q} 
 ( \CC/\Gamma; \CC^2 ) , \ \ 
  U_{\mathbf p } \mathbf u ( z ) := e^{ - i \Re ( z \bar{\mathbf p })}  \mathbf u ( z ), \]  
 and that,
\[  U_{\mathbf p} T_{\mathbf k} U_{\mathbf p}^* = T_{\mathbf k + \mathbf p } = T_{\mathbf k}  ,\ \ \ 
\mathbf p \in \Gamma^*, \ \ \ \mathbf k \notin \Gamma^* .\]
Hence, in the computation of the trace class norm on $ L^2_0 $ we gain $ 1/9 $
and  Proposition~\ref{l:trace} gives, with $ H $ of \eqref{eq:L2G} and $ p = 8 $ (see \eqref{eq:defA}:
the 8th power of $ T_{\mathbf k } $ corresponds to the 4th power of $ B $),
\[ \lVert B^4-\Pi_NB^4\Pi_N\rVert_{\mathcal L _1( L^2_{{\mathbf 0 }} ( \CC/\Gamma, \CC^2 ) )}^\frac16  \leq 
\left( \frac{ 4 \pi 54^{ 4} \rho_1 ( 8, N ) } { 54 \sqrt 3 }\right)^{\frac16}  N^{-1} =
10.2244 \, \rho_1 ( 8, N )^{{\frac16}}  N^{-1} , 
\]
which gives the desired estimate. The operator norm is estimated using  Proposition \ref{l:trace}
as there is no gain from symmetry reduction.
\end{proof} 

Combining Proposition~\ref{l:imp8} and~\eqref{eq:eps4} provides an error estimate in the
numerical computation of $ \alpha_1 $ and $ \alpha_2 $. In principle, the same methods are applicable
for higher $ \alpha$'s shown in Table \ref{tbl:magic} but that seems to require much larger
matrices and any claim of a ``rigorous" calculation is not feasible. 

\begin{table}
\caption{\label{tbl:error}
The values of $ N $ needed to obtain a rigorous error bound of $ \delta = 10^{-k}$,
as computed using the {\tt guarantee.m} code in the Appendix
(using the default {\tt NN=16}).
The matrices used in calculations then have size $ ( 2N + 1)^2 $-by-$(2N+1)^2 $.
Hence the rigorous error estimates are realistic for $ \alpha_1 $ and for rough bounds
on $ \alpha_2$ and $ \alpha_3 $ but not for higher $ \alpha_j$'s. All the values 
of $N \le 328$ here were certified by a second (long) run of {\tt guarantee.m} with the
procedure described in the Appendix.
}
\begin{center}
\begin{tabular}{rrrrrc}
\multicolumn{1}{c}{$k $} & &
\multicolumn{1}{c}{$\alpha_1$} &
\multicolumn{1}{c}{$\alpha_2$} &
\multicolumn{1}{c}{$\alpha_3$} \\[2pt] \hline
1  &&   21 &    128 &  374 &   \\
2  &&   21 &    159 &  476 &   \\
3  &&   28 &    226 &  689 &   \\
4  &&   38 &    328 & 1011 &   \\
5  &&   51 &    472 & 1480 &   \\
6  &&   71 &    691 & 2168 &   \\
7  &&  100 &   1012 &  &   \\   
8  &&  145 &   1485 & &   \\   
9  &&  211 &        & &   \\   
\end{tabular}
\end{center}
\end{table} 

Replacing $ B $ with $ B^4 $ of   Proposition \ref{l:imp8} we see that 
\eqref{eq:approx} holds for that $ B $. We then have 
\[ | \beta^{-8} - \alpha_j^{-8}| < \epsilon := \beta^{-8} - (\beta+ \delta )^{-8} 
\ \Longrightarrow | \beta - \alpha_j | < \delta. \]
This is particularly favourable in the case of $ \alpha_1 $ as then $ \beta \simeq 0.5$.
(We have to take $ \epsilon  $ sufficiently small to avoid other eigenvalues of $  B $.)

\begin{table}
\caption{\label{tbl:back}
Values needed for the backward error calculation 
guaranteeing $ 10^{-k} $ accuracy for computing $ \alpha_j$ (those errors
are much smaller than those from Proposition \ref{l:imp8}).
We show $ e_j = \| ( B_{N_k^j} - \mu_{32}^j ) u_{32}^j \| /\| u_{32}^j \| $ where
$ N_k^j $ comes from Table~\ref{tbl:error}, 
$ \mu_{32}^j $ is the eigenvalue closest to $ \alpha_j^{-8} $ 
obtained using $ B_{32} ,$ and $ u_{32}^j $ is the corresponding
eigenvector extended by $ 0 $ -- see {\tt backerror.m} in the 
Appendix.  These values, on the order of machine precision,
can vary slightly based on implementation, machine, and MATLAB version.}

\begin{center}
\begin{tabular}{rclccc}
\multicolumn{1}{c}{$k $} & &
\multicolumn{1}{c}{$ e_110^{15}   $} & $ e_2 10^{15} $ &  $ e_3 10^{15} $ & 
  \\[2pt] \hline
1  &&      \ &  4.33 & 3.47 &        \\
2  &&      \ &  4.33 & 3.47 &   \\
3  &&      \ &  4.33 & 3.47 &     \\
4  &&   1.68 &  4.33 & 3.47 &       \\
5  &&   1.68 &  4.33 & 3.47 &        \\
6  &&   1.68 &  4.33 & 3.47 &        \\
7  &&   1.68 &  4.33 &  &        \\
8  &&   1.68 &  4.33   & &        \\
9  &&   1.68 &  &  &        \\
\end{tabular} 
\end{center}
\end{table} 

The method described above is implemented in  {\tt BkN.m} in the Appendix,
which computes $ \Pi_N B_{\mathbf k} \Pi_N$ (see Proposition~\ref{l:imp8}).
The code 
{\tt guarantee.m} then returns an $ N$ for which we obtain an accuracy of $ \delta $.
We have to trust the numerical calculation of the smallest singular value of
 $ ( 2 N + 1)^2$-by-$(2N + 1)^2 $ matrices needed for \eqref{eq:eps0} and \eqref{eq:eps1}.
{To estimate the backward error associated with an approximate eigenpair of $B_N$}, we need to calculate 
$ \| ( B_N - \mu_N ) u_N \| $, where $ \mu_N $ and $ u_N $ are the eigenvalue and eigenvector
returned by MATLAB.\ \ We know then that
 $ \mu_N $ is an {\em exact} eigenvalue of $ B_N + R_N $ where
$ \| R_N \| \leq \| ( B_N - \mu_N ) u_N \|/\| u_N \| $.  In principle $ R_N $ should be added to 
$ E_N $, but those errors are negligible compared to our estimates on $ E_N $.\ \  
We should stress that, for these estimates, we do not need to calculate $ \mu_N $ and
$ u_N $ from $ B_N $ {for the large values of $N$ given in Table~\ref{tbl:error}. 
It is sufficient to compute the eigenpair for $ B_{32} $, then take $ \mu_N = \mu_{32} $ and
build $ u_N \in \CC^{(2N+1)^2} $ by extending $ u_{32} \in \CC^{4225} $ by $ 0$s.}
(This extension is justified by noting that the function approximated by $ u_N $ is a solution of
an elliptic equation with analytic coefficients, hence analytic \cite[Theorem 9.5.1]{H1}. 
Consequently, Fourier coefficients decay exponentially.) We show the resulting 
error in Table \ref{tbl:back}.

Table \ref{tbl:error} gives estimates of values of $ N $ for which calculated $ \alpha$'s
are within $ \delta = 10^{-k} $ of the actual elements of $ \mathcal A_{\rm{mag} }$.
Table \ref{tbl:back} gives the estimates of the deviation of $ B_N $ from the 
matrix with eigenvalues given by a MATLAB calculation.
Hence we can 
claim a rigorous calculation for $ \alpha_1 $ and $ \alpha_2  $ within errors
 $10^{-9} $ and $10^{-3}$, respectively.

\vspace{0.4cm}
\begin{center}
\noindent
{\sc  Appendix}
\end{center}
\renewcommand{\theequation}{A.\arabic{equation}}
\refstepcounter{section}
\renewcommand{\thesection}{A}
\setcounter{equation}{0}

 We include a MATLAB code, {\tt BkN.m}, that constructs a sparse matrix 
of the truncation (as described in  \S \ref{s:sum}) of the operator of $ B_{\mathbf k} :=  3 A_{\mathbf k} $
for the potential 
\begin{equation}
\label{eq:Umu}
U_\mu ( z ) =   \sum_{k=0}^2 \omega^k 
\left( e^{\frac12(\bar z \omega^k-  z \bar \omega^k)} + \mu e^{\bar z \omega^k-  z \bar \omega^k} \right);
\end{equation}
see Figure~\ref{f:mu}.

Approximations of real and complex elements of the magic set $ \mathcal A $ 
are given by computing the spectrum of $ B_{\mathbf k } $:
\begin{equation}
\label{eq:la2al}   \lambda \in \Spec_{L^2_{\mathbf 0}  ( \CC/\Gamma ; \CC^2 ) } ( B_{\mathbf k } ) 
\ \Longrightarrow \ 1/\sqrt \lambda \in \mathcal A , \ \ \ \mathbf k \notin \Gamma^* . \end{equation}
To obtain all $ \alpha$'s with multiplicities we should consider the action on 
all representations of $ \Gamma_3/\Gamma $ rather than just \eqref{eq:L2G} 
-- see \S \ref{s:sH} and the proof of Proposition \ref{l:imp8}.  For instance, in MATLAB, 
\[   \alpha_1 \simeq \text{ \tt real(1{.}/sqrt(eigs(BkN(0.5,8),1))) =  0.585663558389558} . \]
The size of the matrix is 289-by-289 ($ ( 2 N + 1)^2 = 289$, $ N = 8$) and no 
improvement is achieved by taking larger matrices.

{\codefontsize
\begin{verbatim}
function B = BkN(k,N);    % create Pi_N * Bk * Pi_N
  N0 = N; N=N+2; N2 = N;
  Rp=RR(k,N,1); Rm=RR(k,N,-1);
  omega=exp(2i*pi/3); N=2*N+1; n=N^2;
  J1 = spdiags(ones(N,1),1,N,N);
  Vp = speye(n)+omega^2*kron(speye(N),J1')+omega*kron(J1',speye(N));
  Vm = speye(n)+omega^2*kron(speye(N),J1)+omega*kron(J1,speye(N));
  B = Rp*Vp*Rm*Vm/3;
  indx = downsize(N0,N2);
  B = B(indx,indx);
end
function RR=RR(k,N,j)
  kk=-N:1:N; N=2*N+1; n=N^2; kk1=kk-j/6; kk1=spdiags(kk1',0,N,N);
  omega=exp(2i*pi/3);
  RR = omega^2*kron(kk1,speye(N))-omega*kron(speye(N),kk1);
  RR = RR-(omega^2*real(k)-omega*imag(k))*speye(size(RR));
  RR = spdiags(1./diag(RR),0,n,n);
end
function indx = downsize(N1,N2); % indices to truncate from N1 to N2
 n1 = max(N1,N2); n2 = min(N1,N2); dn = n1-n2;
 indx = reshape(1:(2*n1+1)^2,2*n1+1,2*n1+1);
 indx = indx(dn+1:dn+2*n2+1,dn+1:dn+2*n2+1);
 indx = reshape(indx,(2*n2+1)^2,1);
end
\end{verbatim}}

To reproduce (half of) Figure~\ref{f:resa}
one simply calls 
{\codefontsize \begin{verbatim}
plot(1./sqrt(eigs(BkN(0.5,32),800)),'ro','LineWidth',1.5) 
xlim([0,18]), ylim([-9,9])
\end{verbatim}}

The error bounds based on Proposition \ref{l:imp8} are implemented
in {\tt guarantee.m}, which returns an estimate on $ N $ needed
to obtain accuracy $ \delta$ using {\tt BkN.m}. 
{The subroutine {\tt Bk4} uses {\tt BkN} to form $\Pi_N B_{\mathbf k}^4 \Pi_N$,
via~(\ref{eq:pint}).}
As explained in \S \ref{s:error}
the only ``non-rigorous" aspect here {involves the calculation of 
the smallest singular values of sparse matrices (a reliable numerical task)}. 
To find $ N $ for, say, accuracy $ \delta=0.1 $
for computing $ \alpha_2 $, the command {\tt guarantee(0.1,2)} returns
an approximation, $ N = 128$, based on an estimate of those singular values 
with a lower $ N $ (experimentally, always the same). To have a
``rigorous" confirmation, $ N = 128 $ should then be used to 
run {\tt guarantee(0.1,2,116)} {(which again produces $N=116$,
though at a much longer run time)}.  Table~\ref{tbl:error} was produced
using {\tt guarantee($10^{-k}$,p)}, $ p=1,2,3$.
We ran the second refinement step to confirm $N$ for all values
in this table with $N\le 328$.

{\codefontsize
\begin{verbatim}
function N = guarantee(delta,p,NN)
% returns N for which alpha_p is computed within error delta, p = 1,2,3
 if (nargin<2) p=1; end
 if (nargin<3) NN=16; end
 alpha(1)=0.585663; alpha(2)=2.221182; alpha(3)=3.7514055;
 rad(1)=72.2;rad(2)=0.0017;rad(3)=2.3830e-05; % dist to the rest of A.^-8
 bet=alpha(p); epsi=bet^-8-(bet+delta)^-8; epsi=min(rad(p)/5,epsi);
 Cep=circle_norm(epsi,NN,bet); M=16; C0=2*6^8*rhoj(M,0)*M^(-8)*Cep;
 while C0>0.5, M = M+1; C0=Cep*2*6^8*rhoj(M,0)*M^(-8); end
 N=M; C0=Cep*(1-C0)^(-1); C1=10.23*rhoj(N,1)^(1/6);
 while (C0*Cep*epsi)^(1/6)*C1 > N, N=N+1; C1=10.23*rhoj(N,1)^(1/6); end
end
function [C,J] = circle_norm(epsi,N,bet)
% Computes the approximate norm of (B-lambda)^-1 for B=Pi_N*Bk(0.5)^4*Pi_N
% and |lambda-mu|=epsi where mu is an approximate eigenvalue of B
 b=1/bet^8; B4=Bk4(0.5,N); J=10; [C1,del]=Jtest(J,B4,epsi,b);
 while del>0.5, J=2*J; [C1,del]=Jtest(J,B4,epsi,b); end
 C=C1/(1-del);end
function [C1,del]=Jtest(J,T,epsi,mu)
% calculates the maximum of the norm of (T-lambda)^{-1}, T sparse
% at J points on the circle |lambda-mu|=epsi
 mu = eigs(T,1,mu);
 zz = exp(1i*(0:1:J-1)*2*pi/J);   la = mu + epsi*zz;
 for j=1:J, A=T-la(j)*speye(size(T)); CC(j)=1/svds(A,1,'smallest'); end
 C1=max(CC); del=2*max(CC)*epsi*sin(pi/(2*J)); end
function rhoj = rhoj(N,j)
 rhoj=1; for ell=0:7 rhoj=rhoj*(1-(2*ell+j)/N)^(-1+j/4); end
end
function B4 = Bk4(k,N); % create Pi_N * Bk^4 * Pi_N
 Bp8 = BkN(k,N+8);                  % Pi_{N+8} Bk Pi_{N+8}
 Bp4 = BkN(k,N+4);                  % Pi_{N+4} Bk Pi_{N+4}
 Bp8sq = Bp8^2;                     % (Pi_{N+8} Bk Pi_{N+8})^2
 indx_8_4 = downsize(N+4,N+8);
 Bp8sq = Bp8sq(indx_8_4,indx_8_4);  % Pi_{N+4} Bp8sq Pi_{N+4}
 B4    = Bp4*Bp8sq*Bp4;
 indx_4_0 = downsize(N,N+4);
 B4    = B4(indx_4_0,indx_4_0);
end
\end{verbatim}}

Finally we include the code used to obtain Table \ref{tbl:back},
{using the discretization in {\tt BkN.m}}.

{\codefontsize
\begin{verbatim}
function ba = backerror(N2,p,N1)
  if (nargin < 3) N1=32; end
  N1 = min(N2-1,N1);
  alpha(1)=0.585663; alpha(2)=2.221182; alpha(3)=3.7514055;
  al = alpha(p); mu = 1/al^8; B1 = BkN(0.5,N1); B2 = BkN(0.5,N2);
  [v1,lam1] = eigs(B1,1,1/al^2);
% inflate the N1 eigenvector to N2 by:                                           
% - shaping it into a (2*N1+1)-by-(2*N1+1) matrix;                                       
% - padding it with a border of dN := N2 - N1 zeros;                                     
% - reshaping it into a (2*N2+1)^2 length vector.                                        
  dN = N2-N1; 
  V1 = [zeros(dN,2*N2+1);
        zeros(2*N1+1,dN)  reshape(v1,2*N1+1,2*N1+1) zeros(2*N1+1,dN);
        zeros(dN,2*N2+1)];
  v2 = reshape(V1,(2*N2+1)^2,1); ba = norm(B2*v2-lam1*v2)/norm(v2);
end
\end{verbatim}}

\smallsection{Acknowledgements} 
We would like to thank Mike Zaletel for bringing \cite{magic} to our 
attention, Alexis Drouot for helpful discussions, and Michael Hitrik 
for {bringing} \cite{see} to our attention.
SB~gratefully acknowledges support by
the UK Engineering and Physical Sciences Research Council (EPSRC)
grant EP/L016516/1 for the University of Cambridge Centre for Doctoral
Training, the Cambridge Centre for Analysis. ME and MZ were partially 
supported
by the National Science Foundation under the grants DMS-1720257 and DMS-1901462, respectively. 
JW was partially supported by the Swedish Research Council grants 2015-03780 and 2019-04878.


\end{document}